%% file: draft.tex
\documentclass[12pt]{article} %

\input{preamble}

\DeclareMathOperator{\E}{\mathbf{E}}
\DeclareMathOperator{\tr}{\mathrm{tr}}
\DeclareMathOperator{\var}{\mathrm{var}}
\DeclareMathOperator*{\argmin}{arg\,min}

\title{Optimal Shrinkage Estimation of Fixed Effects in Linear Panel Data Models}

\author{Soonwoo Kwon\thanks{email:
    \texttt{\href{mailto:soonwoo\_kwon@brown.edu}{soonwoo\_kwon@brown.edu}}. I
    thank Donald Andrews and Timothy Armstrong for continuous support and
    helpful comments throughout the project. I also thank Jason Abaluck, Joseph
    Altonji, Ian Ball, Barbara Biasi, Xiaohong Chen, Jiaying Gu, John Eric
    Humphries, Yuichi Kitamura, Roger Koenker, Koohyun Kwon, Adam McCloskey,
    Cormac O'Dea, Vitor Possebom, Jonathan Roth, Nicholas Snashall-Woodhams, Suk
    Joon Son, Edward Vytlacil and Conor Walsh for helpful comments. Seminar
    participants at CU Boulder, University College London, Harvard/MIT,
    Pennsylvania State University, University of Toronto, Indiana University,
    and Princeton University have provided invaluable feedback. Juan Yamin Silva
    provided excellent research assistance. The New York City Department of
    Education generously provided the data used in the empirical
    section.}\\Brown University}

\date{September 5, 2025}
  
\begin{document}

\maketitle

\begin{abstract}
  Shrinkage methods are frequently used to improve the precision of least
  squares estimators of fixed effects. However, widely used shrinkage estimators
  guarantee improved precision only under strong distributional assumptions. I
  develop an estimator for the fixed effects that obtains the best possible mean
  squared error within a class of shrinkage estimators. This class includes
  conventional shrinkage estimators and the optimality does not require
  distributional assumptions. The estimator has an intuitive form and is easy to
  implement. Moreover, the fixed effects are allowed to vary with time and to be
  serially correlated, in which case the shrinkage optimally incorporates the
  underlying correlation structure. I also provide a method to forecast fixed
  effects one period ahead in this setting.
\end{abstract}

\newpage

\section{Introduction}
\label{sec:introduction}

Applied economists are often interested in unit-specific effects in linear panel
data models.\footnote{Readers are referred to \cite{walters2024empirical} for an
overview.} Estimation in such settings typically requires a large number of
unit-specific fixed effects. However, relatively small sample sizes at the unit
level yield noisy estimates of these effects. Empirical Bayes (EB) methods,
which shrink the least squares estimates, are frequently used to improve their
precision.

However, commonly used EB approaches guarantee such precision improvements only
under stringent assumptions. Moreover, existing methods that relax these
assumptions are typically limited to settings in which the true effects are
independent across units. %
This restriction reduces their applicability in contexts where disaggregated
effects or effects on multiple outcomes are of interest.\footnote{While there is
  a conceptual distinction between disaggregated effects and effects on multiple
  outcomes, they both lead to the same econometric setting with non-exchangeable
  effects within each unit. For brevity, I use the term disaggregated effects to
  refer to both.} For instance, it has been noted by
\cite{chetty2014MeasuringImpactsTeachers} that allowing for a time drift is
crucial in the context of teacher value-added (TVA). Yet, since value-added is
likely correlated across years within a single teacher, methods based on the
assumption of independent effects cannot be directly
applied.\footnote{Other examples where such disaggregated effects are of
  interest include insurance company-level effects
  (\citealp{abaluck2020MortalityEffectsChoicea}), time-drifting college
  counselor effects (\citealp{mulhern2023beyond}) and teacher effect on multiple
  outcomes (\citealp{rose2022effects}).} The goal of the paper is to provide an
estimation method that 1) guarantees precision improvement under weak
conditions and 2) applies to settings with disaggregated effects. The proposed
method is derived by extending the analysis by \cite{xie2012SUREEstimatesHeteroscedastica}.

Consider the canonical EB approach, which assumes:\vspace{-5pt}
\begin{equation}
  \label{eq:eb_model}
  \theta_{j} \overset{i.i.d.}{\sim} N(0,
  \lambda) \,\,\, \text{and} \,\,\, y_{j} \vert \theta_{j} \overset{indep}{\sim} N(\theta_{j},
  \sigma^{2}_{j}),\vspace{-5pt}
\end{equation}
where $\theta_{j}$ is the true effect and $y_{j}$ is an estimator of
$\theta_{j}$ with known variance $\sigma^{2}_{j}$. An implicit assumption is
that the true effect $\theta_{j}$ is independent of the variance
$\sigma^{2}_{j}$. The importance of this prior invariance assumption has been
discussed in \cite{xie2012SUREEstimatesHeteroscedastica} and
\cite{chen2022gaussian}, with the latter developing a novel nonparametric EB
method that relaxes it. The canonical EB model (\ref{eq:eb_model}) therefore
relies on three components: normality of $\theta_{j}$, normality of
$y_{j}\vert \theta_{j}$ and independence of $\theta_{j}$ and $\sigma_{j}$. I
refer to these collectively as the \textit{parametric/normal EB
  assumptions}.\footnote{I refer to such assumptions simply as \textit{EB
    assumptions} hereafter.} The posterior mean of $\theta_{j}$ is
$\ \hat{\theta}_{j}(\lambda) := E[\theta_{j} \vert y_{j}] =
\frac{\lambda}{\lambda + \sigma_{j}} y_{j}$, and the EB estimator
$\hat{\theta}_{j}(\hat{\lambda})$ is obtained by replacing the unknown
\textit{hyperparameter} $\lambda$ with an estimator $\hat{\lambda}$ based on the
marginal distribution of $y_{j}$ implied by (\ref{eq:eb_model}). The risk
properties of $\hat{\theta}_{j}(\hat{\lambda})$ are therefore inherently
sensitive to these assumptions.%

I propose an alternative shrinkage estimator with optimality properties that do
not rely on the EB assumptions, while retaining a simple form. The approach also
accommodates settings in which each unit-level effect can be decomposed into
multiple, possibly correlated, disaggregated effects. For convenience, I refer
to this disaggregated dimension as “time.” The proposed shrinkage method
incorporates the underlying correlation structure, in contrast to the canonical
EB estimator, which treats all effects as exchangeable. %
In the context of time-varying fixed effects, I develop an optimal forecasting
method to predict the fixed effect one period ahead. A special case of this
forecasting method coincides with the estimator of
\cite{chetty2014MeasuringImpactsTeachers}.

A simple illustration of the main method begins with a multivariate version of
(\ref{eq:eb_model}):
\begin{equation}
   \label{eq:eb_model_main}
  \theta_{j} \overset{i.i.d.}{\sim} N(0,
  \Lambda) \,\,\, \text{and} \,\,\, y_{j} \vert \theta_{j} \overset{indep}{\sim} N(\theta_{j},
  \Sigma_{j}),
\end{equation}
where $\theta_{j}$ and $y_{j}$ are now $T$-dimensional vectors, and $\Lambda$
and $\Sigma_{j}$ are $T \times T$ matrices, reflecting the presence of
disaggregated effects. As before, $\Sigma_{j}$ is assumed to be known. The
posterior mean of $\theta_{j}$ is
$\hat{\theta}_{j}(\Lambda) := \E[\theta_{j} \vert y_{j}] = \Lambda(\Lambda +
\Sigma_{j})^{-1}y_{j}$. Rather than using the marginal likelihood of $y_{j}$
implied by (\ref{eq:eb_model_main}), as in the EB approach, I follow
\cite{xie2012SUREEstimatesHeteroscedastica} and tune the hyperparameter
$\Lambda$ to minimize an estimate of the mean squared error (MSE). I refer to
this risk estimate as the unbiased risk estimate (URE), and the \textit{URE
  estimator} selects hyperparameters by minimizing the URE.\footnote{Both the
strategy and terminology follow
\cite{xie2012SUREEstimatesHeteroscedastica}. Related work using this approach
includes \cite{xie2016OptimalShrinkageEstimation} and
\cite{brown2018EmpiricalBayesEstimates}.} I show
that this leads to an estimator that has minimum asymptotic MSE within the
class of estimators considered, which in this simple case is
$\{(\hat{\theta}_{j}(\Lambda))_{j=1}^{J}: \Lambda \text{ is positive
  semidefinite}\}$. The estimators in this class take the intuitive form of
linearly shrinking $y_{j}$ according to their precision, and the class includes
widely used estimators such as the conventional EB estimators.

While the class of estimators is motivated by the model in
(\ref{eq:eb_model_main}), the optimality results require only a few
bounded-moment conditions and no distributional assumptions.\footnote{That said,
  the method does rely on (\ref{eq:eb_model_main}), or a variation of it
  introduced in later sections, to restrict the class of estimators.} The key reason is that
the URE remains “close” to the true risk under weak conditions, rendering such
assumptions unnecessary. To establish this optimality, I derive new results in a
multivariate normal means setting, which is a frequentist version of
(\ref{eq:eb_model_main})—that is, the problem of estimating the fixed mean
vectors $\{\theta_{j}\}_{j=1}^{J}$ based on observations
$y_{j} \overset{\mathrm{indep}}{\sim} N(\theta_{j}, \Sigma_{j})$ with known
$\Sigma_{j}$ for $j = 1, \dots, J$. Under heteroskedasticity, no estimator has
been shown to be risk-optimal (in the frequentist sense) unless $T = 1$, a case
handled by \cite{xie2012SUREEstimatesHeteroscedastica}. Allowing for $T > 1$ and
general forms of $\Sigma_{j}$, I derive an estimator that achieves the best
possible MSE within a class of estimators that take the form of posterior means
under (\ref{eq:eb_model_main}).

I use the proposed method to estimate a TVA model of public
schools of New York City. %
I revisit the policy exercise of releasing the bottom 5\% of teachers according
to the estimated fixed effects. I find that, relative to the conventional
methods, the composition of released teachers changes by 25\% when using the
proposed estimation method and by 58\% when using the proposed forecast
method. An out-of-sample exercise shows that the average value-added of the
teachers released under the forecast method is about 20\% lower compared to the
case where the conventional estimator is used. Estimates from forecast method
are shown to reduce the MSE by 35\%, indicating that the choice of estimator
makes a significant difference, and that it is crucial to allow for the
value-added to vary with time.

\bigskip

\noindent \textbf{Notation.} Let $\lVert \cdot \rVert$ denote the Euclidean norm
for both vectors and matrices (i.e., the Frobenius norm in the latter case). For
any matrix $A$, $(A)_{ij}$ denotes its $(i,j)$ entry and $s_{k}(A)$ its $k$th
largest singular value. The set of positive semidefinite $k \times k$ matrices
is denoted by $\mathcal{S}_{k}^{+}$, and the $k \times k$ identity matrix is
denoted by $I_{k}.$ For any $\{W_{ijt}\}$, let
$\overline{W}_{jt} = {n^{-1}_{jt}}\sum_{i=1}^{n_{jt}}W_{ijt}$ denote the
$(j,t)$-level sample average.

\section{URE estimators}
\label{sec:ure-estimators}

\subsection{Fixed effects and the normal means model}
\label{sec:fixed-effects-normal}

I consider the following linear panel data model,
\begin{equation}
  \label{eq:model}
  Y_{ijt} = X_{ijt}'\beta + \theta_{jt} + \varepsilon_{ijt}, 
\end{equation}
where $t = 1, \dots, T$, $j = 1,\dots, J$, and $i = 1, \dots, n_{jt}$. Here,
$\{(Y_{ijt}, X_{ijt}')\}$ denotes the observed data, $\varepsilon_{ijt}$ the
idiosyncratic shock, and $\theta_{jt}$ the time-varying fixed effect, which is
the object of interest.\footnote{I treat $\theta_{jt}$ as random to be
  consistent with the hierarchical models used to motivate the proposed
  estimators. However, since no restriction on its potential dependence
  structure with the observed covariates is imposed, it is still considered a
  fixed effect.}  Typically, $i$ denotes the individual level, $j$ the group level,
and $t$ the time dimension. The time-varying fixed effect for $j$,
$\theta_{j}:=(\theta_{j1},\dots \theta_{jT})'$ is assumed to be exchangeable
across $j$. For the idiosyncratic error terms, assume
$\overline{\varepsilon}_{j} = (\overline{\varepsilon}_{j1}, \dots,
\overline{\varepsilon}_{jT})'$ is independent across $j$ and independent of
$\theta_{j}$, with known variance $\Sigma_{j}$. In practice, a consistent
estimator is plugged in, which does not affect the asymptotic properties under
suitable conditions.

\begin{example}[Teacher value-added]
  In the TVA model, $j$ corresponds to teacher, $t$ to year, and
  $i$ to a student assigned to teacher $j$ in year $t$. The outcome variable
  $Y_{ijt}$ is a measure of student achievement and $X_{ijt}$ is a vector of
  student characteristics. The fixed effect $\theta_{jt}$ is the value-added of
  teacher $j$ in year $t$%
  . Assuming that $\varepsilon_{ijt}$ is i.i.d across all $i$, $j$, and
  $t$ with variance $\sigma^{2}_{\varepsilon}$, the variance of the estimate is
  given by
  $\Sigma_{j} = \sigma^{2}_{\varepsilon} \mathrm{diag}(1/n_{j1}, \dots,
  1/n_{jT})$.
\end{example}

I consider the asymptotic regime where $ J \to \infty$ with $T$ and $n_{jt}$
fixed. This captures the common situation where the number of effects to be
estimated is large, with observations for each fixed effect unit being
relatively small. I assume that a consistent estimator $\hat{\beta}$ of $\beta$
is readily available, which is easy to obtain under standard assumptions (see,
for example, \citealp{wooldridge2010EconometricAnalysisCross} for a
textbook-level discussion). For example, the within estimator is consistent in
the present setting under minor regularity conditions.

Let $y_{jt}$ denote the least squares estimator for the fixed
effects%
:
\begin{equation*}
  \label{eq:1}
  y_{jt} := \overline{Y}_{jt} - \overline{X}_{jt}'\hat{\beta}  =
  \overline{X}_{jt}'(\beta - \hat{\beta}) + \theta_{jt} +
  \overline{\varepsilon}_{jt} = \theta_{jt} + \overline{\varepsilon}_{jt} + O_{p}(J^{-1/2}).
\end{equation*}
To see the connection between this estimator and the normal means model, note
that $y_{j}:= (y_{j1}, \dots, y_{jT})' \to_{d}
\theta_{j}+\overline{\varepsilon}_{j}$ for each $j \leq J$.\footnote{Under a mild boundedness condition on $ {X}_{ijt}$
that ensures $\sup_{j}\lVert \overline{X}_{j} \rVert = O_{p}(1)$, this
convergence is uniform over $j$. %
} Further assuming that $\overline{\varepsilon}_{j}$ follows a normal
distribution (with known variance matrix $\Sigma_{j}$), it follows that
$ \left( \theta_{j}+\overline{\varepsilon}_{j} \right) \vert \theta_{j} \sim
N(\theta_{j}, \Sigma_{j})$ so that
\begin{equation}
  \label{eq:normal_means_model}
  y_{j} \vert \theta_{j} \sim N(\theta_{j}, \Sigma_{j}),
\end{equation}
approximately. Note that even if $\varepsilon_{ijt}$ is homoskedastic, the
different cell sizes $n_{jt}$ lead to heteroskedasticity of the estimators,
$\var (\overline{\varepsilon}_{jt}) = \var(\varepsilon_{ijt})/n_{jt}$. Due to
this connection, I now abstract away from the panel data model and focus on the
problem of estimating $\theta = (\theta_{1}', \dots, \theta_{J}')'$ after
observing the data $\{y_{j}\}_{j=1}^{J}$ that is generated according to
(\ref{eq:normal_means_model}). The variance matrix
$\Sigma_{j} \in \mathcal{S}_{T}^{+}$ is assumed to be known, or consistently
estimable.

\subsection{Class of shrinkage estimators}
\label{sec:class-shrink-estim}

The URE estimators will be shown to be optimal within a class of shrinkage
estimators that can be viewed as Bayes estimators under a certain hierarchical
model. To motivate the class of estimators considered, suppose that the true
effects are drawn from a normal distribution:
\begin{equation}\label{eq:second_level}
  \theta_{j} \overset{\mathrm{indep}}{\sim} N(\mu_{j},\Lambda),
\end{equation}
where the location vector $\mu_{j} \in\mathbf{R}^{T} $ and the variance matrix
$\Lambda \in \mathcal{S}_{T}^{+}$ are unknown. Under
\eqref{eq:normal_means_model} and \eqref{eq:second_level}, the conditional
expectation of $\theta_{j}$ given $y_{j}$\footnote{This is the posterior mean of
  $\theta_{j}$ if we interpret (\ref{eq:normal_means_model}) and
  \eqref{eq:second_level} as a Bayesian model, except for the fact that
  $\mu_{j}$ and $\Lambda$ are treated as unknown.} is
\begin{equation}\label{eq:post_mean}
\hat{\theta}_{j}(\mu_{j}, \Lambda) := \E[\theta_{j} \vert y_{j}] =  \left( I_{T} - \Lambda(\Lambda + \Sigma_{j})^{-1} \right)\mu_{j} +  \Lambda(\Lambda + \Sigma_{j})^{-1}y_{j}
\end{equation}
The unknown parameters $\mu_{j}$ and $\Lambda$ are later tuned to achieve
desirable risk properties. The restriction one imposes on $\mu_{j}$ and
$\Lambda$ determines the class of estimators. I denote by
$\mathcal{M} \subset \mathbf{R}^{J \times T}$ and
$\mathcal{L} \subset \mathcal{S}_{T}^{+}$ the (possibly random) sets that
reflect this restriction. For example, setting $\mathcal{M} = \{(0,\dots,0)'\}$
corresponds to the common practice of shrinking to the origin, where the degree
(and direction) of shrinkage is determined by $\Lambda$. Also, if one takes
$\mathcal{L}$ to be the set of positive semidefinite Toeplitz matrices, the
estimators place more weight on estimates from nearby time periods.\footnote{In
  this case, the dimension of $\Lambda$ is reduced to $T$, compared to
  $T(T+1)/2$ when $\Lambda$ is left unrestricted. Hence, the restrictions can be
  imposed for computational considerations as well.}

Analogous to the univariate case, I refer to $\Lambda(\Lambda +
\Sigma_{j})^{-1}$ as the shrinkage matrix. Unlike the univariate case where $T =
1$, this estimator not only shrinks the magnitude of the deviations from
$\mu_{j}$ but also applies a rotation. That is, it differentially shrinks linear
combinations of the elements of $y_{j}$, rather than each element
individually. This reflects the fact that the optimal amount of shrinkage
depends not only on the variances but also on the covariances, either directly
or through the sample sizes $n_{jt}$.

\begin{example}[Independent effects]
  If $\mathcal{L} = \{\lambda I_{T}: \lambda \geq 0\}$ and
  $\Sigma_{j} = \mathrm{diag}(\sigma^2_{j1}, \dots, \sigma^2_{jT})$, then the
  $t$th component of $\hat{\theta}_{j}(\mu_{j}, \Lambda)$ is given as
\begin{equation*}
 \left( 1 - \frac{\lambda}{\lambda + \sigma^{2}_{jt}} \right)
  \mu_{jt} +  \frac{\lambda}{\lambda + \sigma^{2}_{jt}} y_{jt},
\end{equation*}
which is the form of shrinkage estimators\footnote{More precisely, the
  estimators used in the literature take this form without the time-varying
  component, and thus effects are aggregated at the $j$ level so that the
  subscript $t$ disappears.} used in the literature with a specific choice of
$\lambda$ and $\mu$.
\end{example}

\begin{example}[Perfectly correlated effects]
  Let $\mathbf{1} $ denote the $T$-vector with all elements equal to
  $1$. Consider the case where $\Lambda = \lambda \mathbf{1}\mathbf{1}'$. Let
  $\Sigma_{j} = \sigma^{2} \mathrm{diag}(1/n_{j1}, \dots, 1/n_{jT})$ and denote
 the teacher-level sample size by $n_{j} = \sum_{t=1}^{T}n_{jt}$. Then,
  we have
  $\hat{\theta}(0, \Lambda)= \mathbf{1} \frac{ \lambda}{\sigma^{2}/n_{j} +
    \lambda} \left( \frac{1}{n_{j}}\sum_{t=1}^{T} n_{jt}y_{jt} \right)$.  The
  term $\frac{1}{n_{j}}\sum_{t=1}^{T} n_{jt}y_{jt}$ is a weighted mean of the
  least squares estimators of teacher $j$ in year $t$, and thus is equal to the least
  squares estimator for the teacher-level fixed effect without time drift. This
  is the estimator used in the majority of the TVA
  literature with an appropriate choice of $\lambda$.
\end{example}

Now, writing $\mu = (\mu_{1}', \dots, \mu_{J}')'$ and
$\hat{\theta}(\mu, \Lambda) = (\hat{\theta}_{1}(\mu_{1}, \Lambda)', \dots,
\hat{\theta}_{J}(\mu_{J}, \Lambda)')$, the class of estimators considered is
$$\widehat\Theta(\mathcal{M},\mathcal{L}) :=
\{\hat{\theta}(\mu, \Lambda): \mu \in \mathcal{M} \text{ and }
\Lambda \in \mathcal{L} \}.$$ This is precisely the set of estimators that take
the form of ``posterior means'' given in (\ref{eq:post_mean}), under the
restriction that $\mu \in \mathcal{M}$ and $\Lambda \in \mathcal{L}$. I refer to
$\mu$ and $\Lambda$ as the hyperparameters and $\mathcal{M}$ and $\mathcal{L}$
as the hyperparameter spaces.

I consider three different specifications of $\mathcal{M}$, with increasing
flexibility. The first specification, which is the simplest, takes
$\mathcal{M} = \mathcal{M}_{\mathrm{m}}:= \{(\overline{y}'_{J}, \dots,
\overline{y}_{J}')'\}$, where
$\overline{y}_{J} := \frac{1}{J}\sum_{j=1}^{J}y_{j}$. This corresponds to
shrinking to the grand mean. If the covariates in the panel data model
(\ref{eq:model}) included time dummies, then we have $\overline{y}_{J} = 0$,
which corresponds to the common practice of shrinking the least squares estimates
toward the origin.

The second specification takes
$\mathcal{M} = \mathcal{M}_{\mathrm{g}}:= \{(\mu_{0}', \dots, \mu_{0}')':
\mu_{0} \in \mathcal{B}\}$, which shrinks the data toward a general location
$\mu_{0}$. The restriction that $\mu_{0}$ lies in $\mathcal{B}$ is a technical
condition that ensures a certain boundedness property.\footnote{The set
  $\mathcal{B}$ is defined as
  $\mathcal{B} := \{\mu \in \mathbf{R}: \lvert \mu_{t}\rvert \leq
  q_{1-\tau}(\{\lvert y_{jt}\rvert \}_{j=1}^{J}) \,\, \mathrm{for} \,\, t= 1,
  \dots, T \},$ where $q_{1-\tau}(\{ \lvert y_{jt} \rvert \}_{j=1}^{J})$ denotes
  the $1-\tau$ sample quantile of $\{ \lvert y_{jt} \rvert \}_{j=1}^{J}$. A
  similar idea was used by \cite{brown2018EmpiricalBayesEstimates}. I recommend
  a small $\tau$, such as $\tau = .01$. This restricts $\mu_{0}$ to be smaller
  than the 99th percentile of the data in terms of magnitude. This is a mild
  restriction that ensures the data are not shrunk to somewhere with almost no
  observations nearby.} This estimator shrinks the data $y_{j}$ to a common but
general location $\mu_{0}$, which will later be chosen in a data-driven way to
minimize the risk.

For the third specification, consider a setting where we have additional data,
$Z_{jt} \in \mathbf{R}^{k} $. In the panel data model (\ref{eq:model}), these
are exactly the covariates that cannot be included as regressors because of the
inclusion of $(j,t)$-level fixed effects. I consider the specification
$\mathcal{M} = \mathcal{M}_{\mathrm{cov}} := \{(Z_{1}'\gamma, \dots,
Z_{J}'\gamma): \gamma \in \Gamma\}$, where $Z_{j} = (Z_{j1}, \dots, Z_{jT})'$,
which corresponds to a class of estimators that shrinks the least square
estimates $y_{j}$ to a linear combination of the covariates
$Z_{j}'\gamma$.\footnote{While I only consider the simple case of shrinking
  toward a linear combination of the covariates, the estimator can be extended
  in a straightforward manner to shrink toward more general functions of the
  covariates.} Unlike the other two specifications, each estimate $y_{j}$ is
shrunk to a different location. Again, the restriction that $\gamma \in \Gamma$
is a technical restriction that ensures a certain boundedness property holds for
optimality results.\footnote{The set $\Gamma$ is defined as
  $\Gamma := \{\gamma \in \mathbf{R}^{k}: \lVert \gamma \rVert \leq B \lVert
  \hat{\gamma}^{\mathrm{OLS}} \rVert \}$ where $B$ is a large constant and
  $\hat{\gamma}^{\mathrm{OLS}}$ is the OLS estimator obtained by regressing
  $y_{j}$ on $Z_{j}$. This ensures that the OLS coefficient is a feasible
  choice along with other coefficients with much larger magnitude as well.}

\begin{example}[Teacher value-added]
  In a TVA setting, teacher (or teacher-year) level covariates are
  frequently available. Such covariates cannot be included in (\ref{eq:model}),
  since they are absorbed into the teacher-year fixed effects. However, one can use such
  covariates to improve the precision of the teacher fixed-effect
  estimates by using $\hat{\theta}^{\mathrm{URE}, \mathrm{cov}}$. Frequently available teacher level covariates include, for example,
  gender, tenure, and union status of a teacher. Class size, which is almost
  always available, is also an example of such covariate. Asymptotically, the inclusion
  of such covariates are guaranteed to improve the MSE.%
\end{example}

\subsection{Risk estimate and URE estimators}
\label{sec:risk-estimate}

To evaluate the performance of different estimators, I condition on $\theta$ and
use the compound MSE as the performance criterion. To be specific, writing the
compound loss as
$\ell(\theta, \hat{\theta}) := \frac{1}{J}(\hat{\theta} - \theta)' (\hat{\theta}
- \theta),$ the (conditional) compound MSE is given as
$R(\theta, \hat{\theta}) = \E_{\theta} \ell(\theta, \hat{\theta})$. %
This is a frequentist risk criterion since we condition on the parameters. For
what follows, I treat $\theta$ as fixed with the understanding that all claims
are conditional on a fixed sequence $\{\theta_{j}\}_{j=1}^{\infty}$.\footnote{By considering MSE
  conditional on $\theta$, the optimality results established later on will not
  depend on normality of $\theta$ which was used to motivate the class of estimators.} The
subscript $\theta$ in $\E _{\theta}$ has been used to make clear that the
expectation is conditional on this sequence of $\theta$, but I omit the
subscript unless ambiguous otherwise.

Given this performance criterion, an optimal choice of the hyperparameters is
$ (\tilde{\mu}^{\mathrm{OL}}_{\mathcal{M}, \mathcal{L}},
\tilde{\Lambda}^{\mathrm{OL}}_{{\mathcal{M}, \mathcal{L}}}) := \argmin_{(\mu,
  \Lambda) \in \mathcal{M} \times \mathcal{L}} \ell(\theta, \hat{\theta}(\mu,
\Lambda))$, which gives the oracle loss ``estimator''
$\tilde{\theta}^{\mathrm{OL}}(\mathcal{M}, \mathcal{L}) :=
\hat{\theta}(\tilde{\mu}^{\mathrm{OL}}_{\mathcal{M}, \mathcal{L}},
\tilde{\Lambda}^{\mathrm{OL}}_{{\mathcal{M}, \mathcal{L}}})$. Of course, this is
infeasible because the loss function depends on the unobserved true mean
vectors. However, it turns out that the risk is estimable via Stein's unbiased
risk estimate (SURE). The idea is to choose the hyperparameters by minimizing
this risk estimate.\footnote{The idea of minimizing SURE to choose tuning
  parameters has been around since at least \cite{li1985SteinUnbiasedRisk}, and
  has been introduced to this setting by
  \cite{xie2012SUREEstimatesHeteroscedastica}.} %
The risk estimate is defined as
$\mathrm{URE}(\mu, \Lambda) = \frac{1}{J}\sum_{j}\mathrm{URE}_{j}(\mu_{j},
\Lambda)$, where
\begin{equation*}
   \mathrm{URE}_{j}(\mu_{j}, \Lambda) :=  \tr ( \Sigma_{j}) -
     2 \tr ((\Lambda+\Sigma_{j})^{-1}\Sigma_{j}^{2}) + (y_{j}-\mu_{j})'[(\Lambda+\Sigma_{j})^{-1}\Sigma_{j}^{2}(\Lambda+\Sigma_{j})^{-1}](y_{j}-\mu_{j}).
\end{equation*}
It is easy to show that $URE(\mu, \Lambda)$ is indeed an unbiased estimator of
the true risk as in $R(\theta, \hat{\theta}(\mu, \Lambda))$. The aim is to show
that choosing hyperparameters to minimize $\mathrm{URE}(\mu, \Lambda)$ is as
good as, in terms of asymptotic risk, choosing them by minimizing the true
loss. The URE estimator is given as
$\hat{\theta}^{\mathrm{URE}}(\mathcal{M}, \mathcal{L}) :=
\hat{\theta}(\hat{\mu}^{\mathrm{URE}}_{\mathcal{M}, \mathcal{L}},
\hat{\Lambda}^{\mathrm{URE}}_{{\mathcal{M}, \mathcal{L}}})$, where the
hyperparameters are chosen to minimize $\mathrm{URE}(\mu, \Lambda)$:
$ (\hat{\mu}^{\mathrm{URE}}_{\mathcal{M}, \mathcal{L}},
\hat{\Lambda}^{\mathrm{URE}}_{{\mathcal{M}, \mathcal{L}}}) := \argmin_{(\mu,
  \Lambda) \in \mathcal{M} \times \mathcal{L}} \mathrm{URE}(\mu, \Lambda)$.

In the EB framework, the hierarchical model given by
(\ref{eq:normal_means_model}) and (\ref{eq:second_level}) is taken as the true
data generating process, and the hyperparameters are estimated by using the
marginal distribution of the data implied by this model,
$y_{j} \overset{\text{indep}}{\sim} N(\mu_{j}, \Lambda + \Sigma_{j})$, either by
maximum likelihood or the method of moments. I denote the EB maximum likelihood
estimator (EBMLE) by
$\hat{\theta}^{\mathrm{EBMLE}}(\mathcal{M}, \mathcal{L}) =
\hat{\theta}(\hat{\mu}^{\mathrm{EBMLE}}, \hat{\Lambda}^{\mathrm{EBMLE}})$ where
$(\hat{\mu}^{\mathrm{EBMLE}}, \hat{\Lambda}^{\mathrm{EBMLE}})$ maximizes the
marginal likelihood subject to $\mu \in \mathcal{M}$ and
$\Lambda \in \mathcal{L}$.

For each of the three specifications $ \mathcal{M}_{m}$, $\mathcal{M}_{g}$ and
$\mathcal{M}_{\mathrm{cov}}$, I define the corresponding URE estimators as
$\hat{\theta}^{\mathrm{URE}, \mathrm{m}} :=
\hat{\theta}^{\mathrm{URE}}(\mathcal{M}_{\mathrm{m}}, \mathcal{S}_{T}^{+})$,
$\hat{\theta}^{\mathrm{URE}, \mathrm{g}} :=
\hat{\theta}^{\mathrm{URE}}(\mathcal{M}_{\mathrm{g}}, \mathcal{S}_{T}^{+})$ and
$\hat{\theta}^{\mathrm{URE}, \mathrm{cov}} :=
\hat{\theta}^{\mathrm{URE}}(\mathcal{M}_{\mathrm{cov}}, \mathcal{S}_{T}^{+})$. The
corresponding oracle estimators $\tilde{\theta}^{\mathrm{OL}, \mathrm{m}}$,
$\tilde{\theta}^{\mathrm{OL}, \mathrm{g}}$ and
$\tilde{\theta}^{\mathrm{OL}, \mathrm{cov}}$ are defined analogously.

\section{Optimality of the URE estimators}
\label{sec:optim-results-shrink}

I now show that the URE estimators defined in Section \ref{sec:risk-estimate}
achieve the smallest possible asymptotic MSE among all estimators
in the corresponding classes. I provide these optimality results for
$\hat{\theta}^{\mathrm{URE}, \mathrm{m}}$ and
$\hat{\theta}^{\mathrm{URE}, \mathrm{g}}$. An analogous result for
$\hat{\theta}^{\mathrm{URE}, \mathrm{cov}}$ is deferred to Appendix \ref{appsec:opt-cov},
as the arguments are similar. The main step in establishing such optimality
is to show that the corresponding UREs are uniformly close to the true
risk. Since I use an unbiased estimate of risk, this essentially reduces to a
uniform (weak) law of large numbers (ULLN) argument. I first establish a simple
high-level result for a generic URE estimator and then verify that the
conditions for this result hold for each estimator under appropriate lower-level
conditions.

The following lemma shows that if $\mathrm{URE}(\mu, \Lambda)$ is uniformly close to the
true loss in $L^{1}$, in the sense that
\begin{equation}
  \label{eq:URE_loss_l1}
  \sup_{(\mu, \Lambda) \in
    \mathcal{M} \times \mathcal{L}} \,\,
  \lvert \mathrm{URE}(\mu, \Lambda) -  \ell(\theta,
    \hat{\theta}(\mu, \Lambda)) \rvert \overset{L^{1}}{\to} 0,
\end{equation}
then the URE estimator has asymptotic risk as good as the oracle.

\begin{lemma}\label{lem:L1_oracle_generic} Suppose (\ref{eq:URE_loss_l1}) holds. Then,
    \begin{equation}\label{eq:optimal_risk}
      \underset{J \to \infty}{\lim} \,\big( R(\theta,
        \hat{\theta}^{\mathrm{URE}}(\mathcal{M}, \mathcal{L})) - R(\theta, \tilde{\theta}^{\mathrm{OL}}(\mathcal{M}, \mathcal{L})) \big) = 0.
    \end{equation}
\end{lemma}
\begin{proof}
See Appendix \ref{appsec:proof-theorem-l1-gen}.
\end{proof}
\noindent %

Therefore, establishing uniform consistency in (\ref{eq:URE_loss_l1}) for
each estimator class is a key step. Then, since
$ R(\theta, \tilde{\theta}^{\mathrm{OL}}(\mathcal{M}, \mathcal{L})) =
\min_{\hat{\theta} \in \widehat\Theta(\mathcal{M}, \mathcal{L})} R(\theta,
\hat{\theta})$, it follows from Lemma \ref{lem:L1_oracle_generic} that the URE
estimators are asymptotically optimal within the class of estimators
$\widehat\Theta(\mathcal{M}, \mathcal{L})$. The difference between the URE and
the true loss can be decomposed as
\begin{align}
    & \,\, \mathrm{URE}(\mu, \Lambda) - \ell(\theta, \hat{\theta}(\mu, \Lambda)) \nonumber \\
    =  & \textstyle ( \mathrm{URE}(0, \Lambda) - \ell(\theta, \hat{\theta}(0, \Lambda))) - \frac{2}{J} \sum_{j=1}^{J}\left( \mu_{j}' (\Lambda + \Sigma_{j})^{-1} \Sigma_{j}(y_{j}-\theta_{j}) \right).  \label{eq:URE_minus_loss_genmu}
\end{align}
I show that both terms in the last line converge to $0$ in $L_{1}$, uniformly
over each hyperparameter space I consider. Note that the first term of the last
line does not depend on $\mu$, and thus is common for all three estimators.

The following assumption states that $y_{j}$'s are independent with (uniformly)
bounded fourth moments, and that the smallest eigenvalue of the variance of
$y_{j}$ is bounded away from zero. I write $y_{j} \sim (\theta_{j}, \Sigma_{j})$
to denote that $y_{j}$ follows a distribution such that
$ \E y_{j} = \theta_{j} $ and $ \var (y_{j}) = \Sigma_{j}$. The supremum
$\sup_{j}$ is taken over all $j \geq 1$, and likewise for $\inf_{j}$. Hence, the
assumption imposes conditions on the sequences
$ \left\{ \E \lVert y_{j} \rVert \right\}_{j=1}^{\infty}$ and
$ \left\{s_{T}(\Sigma_{j}) \right\}_{j=1}^{\infty}$.

\begin{assumption}[Independent sampling and boundedness]
  \label{assum:bounded}

  $ \mathrm{(i)} \, y_{j} \overset{\mathrm{indep}}{\sim}
  (\theta_{j}, \Sigma_{j})$,  \\$ \mathrm{(ii)} \, \sup_{j} \E \lVert y_{j} \rVert^{4} <
  \infty$ and  $ \mathrm{(iii)}\, 0 < \textstyle\inf_{j}s_{T}(\Sigma_{j}).$
\end{assumption}

In the case where $\Sigma_{j}$ is diagonal for all
$j$, Assumption \ref{assum:bounded} (iii) boils down to assuming that $\var
(y_{jt})$ is bounded away from zero over $j$ and
$t$. Also, in the case where $\Sigma_{j} = \Sigma$ for all
$j$, the assumption trivially holds as long as $\Sigma$ is
invertible. %
It turns out that Assumption \ref{assum:bounded} is enough to ensure uniform
convergence of the first term of \eqref{eq:URE_minus_loss_genmu}.

To show convergence of the second term of (\ref{eq:URE_minus_loss_genmu}), note
that for both $\hat{\theta}^{\mathrm{URE}, \mathrm{m}}$ and
$\hat{\theta}^{\mathrm{URE}, \mathrm{g}}$, the centering term $\mu_{j}$ does not
depend on $j$, so that I can write $\mu_{0} = \mu_{j}$ for all $j$. Hence, the
required convergence result is
\begin{equation}
  \label{eq:sec_term}
\E\Big[\textstyle\sup_{\mu_{0} \in \mathcal{M}_{0}, \Lambda \in \mathcal{S}^{+}_{T}} \left\lvert   \frac{1}{J} \sum_{j=1}^{J} \mu_{0}' (\Lambda + \Sigma_{j})^{-1}
      \Sigma_{j}(y_{j}-\theta_{j}) \right\rvert \Big] \to 0
\end{equation}
where $\mathcal{M}_{0} = \{\overline{y}_{J}\}$ for
$\hat{\theta}^{\mathrm{URE}, m}$ and $\mathcal{M}_{0} = \mathcal{B}$ for
$\hat{\theta}^{\mathrm{URE}, g}$. It is clear that some form of boundedness
condition on $\mathcal{M}_{0}$ is necessary for such a convergence result to
hold. For $\hat{\theta}^{\mathrm{URE}, \mathrm{m}}$, Assumption
\ref{assum:bounded} (ii) ensures this. For $\hat{\theta}^{\mathrm{URE},
\mathrm{g}}$, we require conditions that guarantee boundedness of $\mathcal{B}$
in a suitable sense, and the following assumption provides such a guarantee.
\begin{assumption}[Bounded sample quantiles]\label{assu:bdd_quan}
$ \limsup_{J}\E  q_{1-\tau}(\{ y_{jt}^{2} \}_{j=1}^{J}) < \infty$.
\end{assumption}
This assumption states that the expectation of the sample quantile of
$\{ y_{jt}^{2} \}_{j=1}^{J}$ is uniformly bounded. This is a rather mild
condition that is satisfied, for example, if the tail probabilities of
$y_{jt}^{2}$ vanish uniformly to 0. A lower level condition that is stronger but
easy to interpret is provided in Appendix \ref{sec:suff-cond-assumpt}. We now
state the optimality of $\hat{\theta}^{\mathrm{URE}, \mathrm{m}}$ and
$\hat{\theta}^{\mathrm{URE}, \mathrm{g}}$.
\begin{theorem}[Optimality of URE estimators]\label{thm:URE_opt_m_g} Suppose Assumption
  \ref{assum:bounded} holds. Then,\\[.5ex]
  (i) $\sup_{\Lambda \in \mathcal{S}_{T}^{+}}  \big|
      \mathrm{URE}(\overline y_{J}, \Lambda) - \ell(\theta, \hat{\theta}(\overline y_{J},
      \Lambda)) \big|  \overset{L^{1}}{\to} 0,$ and \\[.5ex]
      (ii) if Assumption \ref{assu:bdd_quan} holds as well, then $\sup_{\mu \in \mathcal{B}, \Lambda \in \mathcal{S}_{T}^{+}}  \big|
      \mathrm{URE}(\mu, \Lambda) - \ell(\theta, \hat{\theta}(\mu,
      \Lambda)) \big|  \overset{L^{1}}{\to} 0.$
    \end{theorem}
    \begin{proof}
      See Appendix \ref{sec:proof-main-theorems}.\footnote{\label{fn:
          diff_with_xie}The proof technique used in related papers (e.g.,
        \citealp{xie2012SUREEstimatesHeteroscedastica,
          xie2016OptimalShrinkageEstimation})
        does not go through in this setting mainly because 1)
        the matrix hyperparameter $\Lambda$ governs the direction of
        shrinkage as well as the magnitude and 2) normality of
        $y_{j}$ is not assumed. %
      }
    \end{proof}

    The proof of this theorem, provided in Appendix
    \ref{appsec:proof-theorem-URE}, relies on establishing a ULLN over
    independent but non-identically distributed sequences of data, followed by
    verifying uniform integrability to strengthen the mode of convergence from
    convergence in probability to convergence in $L^{1}$. Note that the uniform
    convergence is shown over the largest possible hyperparameter space for
    $\Lambda$, $\mathcal{S}_{T}^{+}$, and thus convergence over any subset
    $\mathcal{L} \subset \mathcal{S}_{T}^{+}$ follows immediately. I note that
    Assumption \ref{assum:bounded} is stronger than necessary. However, this
    stronger assumption is not particularly restrictive and has the clear
    advantages of simplifying the proofs and being easy to
    interpret.%

    By Lemma \ref{lem:L1_oracle_generic}, it follows that
    $\hat{\theta}^{\mathrm{URE},\mathrm{m}}$ and
    $\hat{\theta}^{\mathrm{URE},\mathrm{g}}$ are asymptotically optimal within
    $\widehat\Theta(\mathcal{M}_{m}, \mathcal{S}_{T}^{+})$ and
    $\widehat\Theta(\mathcal{M}_{g}, \mathcal{S}_{T}^{+})$, respectively. Note
    that the optimality of the URE estimators requires only mild conditions on
    the moments of the data, which is in contrast with the EB estimators that
    require stringent distributional assumptions. The EB estimators are optimal
    in the sense of \cite{robbins1964EmpiricalBayesApproach}\footnote{The
      estimator obtains the Bayes risk of the model
      (\ref{eq:normal_means_model}) and (\ref{eq:second_level}).} only when 1)
    the normality assumptions for both the least squares estimator and the true
    fixed effect hold and 2) the true fixed effect and variance of the least
    squares estimator are independent.

    The normality assumption on the true fixed effect is typically difficult to
    justify.\footnote{Some evidence on the violation of such assumption in the
      context of teacher value-added is provided in
      \cite{gilraine2020NewMethodEstimatinga}.} The optimality results here are
    conditional on a sequence of true mean vectors that satisfy a mild
    boundedness condition. The independence between the true fixed effect and
    the variance of the least squares estimator can be easily violated in
    empirical settings as well. Since the variance of the least squares
    estimator is inversely proportional to $n_{jt}$, the assumption is violated
    if the fixed effect depends on $n_{jt}$. If teachers with
    higher value-added teach more students, or if the size of the class is
    related to teaching effectiveness, then such independence is unlikely to
    hold.

    The nonparametric EB literature (e.g.,
    \citealp{jiang2009GeneralMaximumLikelihood},
    \citealp{brown2009NonparametricEmpiricalBayes}) provides an alternative
    method to relax the normality assumption on the true effects. In this
    setting, the distribution of the true fixed effect is left unspecified,
    thereby broadening the class of estimators. This approach, however, is
    complementary to—but does not dominate—the URE approach, for two main
    reasons. First, the risk properties of currently available nonparametric EB
    methods still rely on the independence of the true fixed effect and the
    variance of the least squares estimator (or at least a structured
    relationship between the two, as in \citealp{chen2022gaussian}), as well as
    on a normality assumption for the least squares estimators. Second, existing
    approaches address only the case $T=1$. Even when extended to $T>1$, the
    associated computation is likely to be infeasible for even moderate values
    of $T$.

    The URE estimators can also be shown to dominate the unshrunk (unbiased)
    estimator, $\hat{\theta}^{\mathrm{ub}} =y$, which corresponds to using the
    least squares estimators without any shrinkage in the context of fixed
    effects. Because there is no $\Lambda \in \mathcal{S}_{T}^{+}$ such that
    $\hat{\theta}(\mu_{j},\Lambda) = y$, the estimator
    $\hat{\theta}^{\mathrm{ub}}$ is not included in any of the classes of
    estimators I consider. However, a simple approximation argument can be used
    to establish that
    ${\limsup}_{J \to \infty}\big( R(\theta, \hat{\theta}^{\mathrm{URE},
      \mathrm{m}}) - R(\theta, \hat{\theta}^{\mathrm{ub}}) \big) \leq
    0$.%
    This shows that using
    $\hat{\theta}^{\mathrm{URE},m}$ is at least as good as using the unshrunk
    estimator, providing a strong justification for shrinkage when the goal is
    to improve precision. This is a property that EB methods do not enjoy unless
    the EB assumptions are satisfied.

\section{Forecasting $\theta_{T+1}$}
\label{sec:optim-shrink-pred}

In settings where $t$ represents time, forecasts of future effects are often of
interest. For example, in teacher evaluation and retention, policies based on
forecasts of future effects—rather than past effects—can improve future student
outcomes. To this end, I consider the problem of predicting
$\theta_{T+1} = (\theta_{1,T+1}, \dots, \theta_{J,T+1})'$. The approach is
analogous to the URE framework proposed earlier: I first derive a class of
predictors based on a hierarchical model, and then tune the hyperparameters by
minimizing an unbiased estimate of the prediction error (UPE). The resulting
forecasts are referred to as UPE forecasts.

For simplicity, I consider the case where the fixed effects are demeaned, i.e.,
(\ref{eq:second_level}) with $\mu_{j} = 0$. Write $\Lambda$ and $\Sigma_{j}$ in
the following block matrix form:
\begin{equation*}
  \Lambda =
  \begin{pmatrix}
    \Lambda_{-T} & \Lambda_{T,-T} \\
    \Lambda_{T,-T}' & \lambda_{T}
  \end{pmatrix}, \Sigma_{j} =
  \begin{pmatrix}
    \Sigma_{j,-T} & \Sigma_{j,T,-T} \\
    \Sigma_{j,T,-T}' & \Sigma_{j,T}
  \end{pmatrix}
  =
  \begin{pmatrix}
    \Sigma_{j,1} & \Sigma_{j,1,-1}' \\
    \Sigma_{j,1,-1} & \Sigma_{j,-1}
  \end{pmatrix}
\end{equation*}
where $\Lambda_{-T}$, $\Sigma_{j,-T}$ and $\Sigma_{j,-1}$ are
$(T-1)\times (T-1)$ matrices. 

A recommended choice for the hyperparameter space $\mathcal{L} \subset
\mathcal{S}^{+}_{T}$ is
\begin{equation*}
\mathcal{L} := \Big\{ \Lambda \in \mathcal{S}_{T}^{+}: s_{1}(\Lambda) \leq K
  s_{1}\textstyle\Big( \frac{1}{J}\sum_{j=1}^{J}y_{j}y_{j}' \Big) \Big\} 
\end{equation*}
for some large number $K$ that does not depend on $J$. Under the hierarchical
model,
$\frac{1}{J}\sum_{j=1}^{J}\E y_{j}y_{j}' = \Lambda +
\frac{1}{J}\sum_{j=1}^{J}(\theta_{j}\theta_{j}' + \Sigma_{j})$, and thus
$\frac{1}{J}\sum_{j=1}^{J}y_{j}y_{j}'$ gives a sense of the scale of
$\Lambda$. By scaling this up by $K$, the bound becomes less restrictive. This
makes $\mathcal{L}$ bounded in a certain sense, which is necessary for the
optimality
argument. %

The aim is to tune the hyperparameter in a way that it minimizes prediction
error of predicting $\theta_{T+1} := (\theta_{1,T+1}, \dots, \theta_{J,
  T+1})'$. However, the challenge is that an unbiased estimator of this
prediction error is unavailable because we do not observe data for period
$T+1$. The strategy is to tune the hyperparameters by considering the problem of
predicting $\theta_{T} = (\theta_{1T}, \dots, \theta_{JT})'$ using the first
$T-1$ periods of data. Under a suitable stationarity condition, this will lead
to optimal hyperparameter selection for predicting $\theta_{T+1}$ with
$y_{j,-1}$ as well.

Consider the problem of estimating $\theta_{T}$ using observations from the
first $T-1$ periods. Let
$y_{j,-t} = (y_{j1}, \dots, y_{j, t-1}, y_{j,t+1}, \dots y_{j,T})'$ and
$y_{-t} = (y_{1,-t}', \dots, y_{J,-t}')'$ denote the vectors $y_{j}$ and $y$
with period-$t$ observations removed, respectively. The
class of estimators I consider is the posterior mean implied by the hierarchical
model,\vspace*{-5pt}
\begin{equation*}
  \E [\theta_{jT}\vert y_{-T}] = \Lambda_{T,-T}' (\Lambda_{-T} + \Sigma_{j,-T})^{-1}y_{j,-T}.\vspace*{-5pt}
\end{equation*}
\noindent Define the multiplicative factor as
$B(\Lambda,\Sigma_{-T}): = (\Lambda_{-T} + \Sigma_{-T})^{-1} \Lambda_{T,-T}$. The
performance criterion is the mean prediction error,
$\E \mathrm{PE}(\Lambda; T)$, where the prediction error is given as\vspace*{-5pt}
\begin{equation*}
 \mathrm{PE}(\Lambda; T) := \textstyle \frac{1}{J}  \sum _{j=1}^{J}(B(\Lambda,\Sigma_{j,-T})'y_{j,-T} - \theta_{jT})^{2}.\vspace*{-5pt}
\end{equation*}
Similar to the URE estimator, I derive an estimator of the prediction error and
choose $\Lambda$ by minimizing this. An unbiased estimator of the mean
prediction error is given as
 \begin{equation*}
\textstyle  \mathrm{UPE}(\Lambda) = \frac{1}{J}\sum_{j=1}^{J} \left( (B(\Lambda, \Sigma_{j,-T})'y_{j,-T} -y_{jT})^{2}  - \Sigma_{jT} + 2B(\Lambda, \Sigma_{j,-T})'\Sigma_{j,T,-T} \right).
\end{equation*}
Writing $\hat{\Lambda}^{\mathrm{UPE}}$ as the $\Lambda$ that minimizes
$\mathrm{UPE}(\Lambda)$, the proposed estimator for $\theta_{j,T+1}$ is
$B(\hat{\Lambda}^{\mathrm{UPE}}, \Sigma_{j,-T})'y_{j,-1}$.

\begin{remark}[Estimator of \citealp{chetty2014MeasuringImpactsTeachers}] While
  I focus on predicting $\theta_{T}$ with the observations from the first $T-1$
  periods, one can also consider predicting $\theta_{t}$ with observations
  excluding the period $t$ observation. If $\Sigma_{j} = \Sigma$ with $\Sigma$
  being diagonal, the $\Lambda$ that minimizes $\mathrm{UPE}(\Lambda)$ implies
  $B(\Lambda, \Sigma_{-t}) = \hat{\beta}^{\mathrm{\mathrm{OLS}},t}$, which is
  the OLS estimator of regressing $y_{jt}$ on $y_{j,-t}$. This leads to the
  estimator used in \cite{chetty2014MeasuringImpactsTeachers}.
\end{remark}

Since the goal is to forecast $\theta_{T+1}$ rather than $\theta_{T}$, it is
necessary to establish that $\mathrm{UPE}(\Lambda)$ is a good estimator of the
prediction error for the problem of predicting $\theta_{T+1}$, \vspace{-5pt}
\begin{equation*}
  \mathrm{PE}(\Lambda;T+1) = \textstyle \frac{1}{J} \sum_{j=1}^{J}(B(\Lambda, \Sigma_{j,-1})'y_{j,-1} - \theta_{j,T+1})^{2}. \vspace{-5pt}
\end{equation*} 
By the same argument made by Lemma \ref{lem:L1_oracle_generic}, if  \vspace{-5pt}
\begin{equation}\label{eq:UPE_conv}
  \sup_{\Lambda \in \mathcal{L}}\,\, \lvert \mathrm{UPE}(\Lambda) -
  \mathrm{PE}(\Lambda; T+1) \rvert \overset{L^{1}}{\to} 0,  \vspace{-5pt}
\end{equation}
then $B(\hat{\Lambda}^{\mathrm{UPE}}, \Sigma_{j,-T})'y_{j,-1}$ obtains the
oracle mean prediction error, which is the mean prediction error of
$B(\tilde{\Lambda}, \Sigma_{j,-T})'y_{j,-1}$ with $ \tilde{\Lambda} := \argmin_{\Lambda}\mathrm{PE}(\Lambda; T+1)$.

Due to the extrapolative nature of the method, a suitable stationarity
assumption is necessary. To formalize this notion of stationarity, I assume that
the pairs $\{ ((\theta_{j}', \theta_{j,T+1})', \Sigma_{j}) \}_{j=1}^{\infty}$
are drawn randomly from a joint density $f_{(\theta', \theta_{T+1})',
  \Sigma}$. Let $f_{\Sigma}$ denote the marginal density of $\Sigma_{j}$ and let
$\mathrm{supp}(f_{\Sigma})$ denote its support. The following assumption is a
modified version of Assumption \ref{assum:bounded} that accounts for the fact
that both $\theta_{j}$ and $\Sigma_{j}$ are now considered to be random draws.  \vspace{-5pt}

\begin{assumption}[Assumption \ref{assum:bounded} with random parameters\vspace{-.17in}] \label{assum:bounded_rand}
  \begin{align*}
    \mathrm{(i)}\,\, & \,\, y_{j} \vert \theta_{j}, \Sigma_{j} \overset{\mathrm{indep}}{\sim} (\theta_{j}, \Sigma_{j}), \hspace{4.9in}\\
    \mathrm{(ii)}\, & \,\, \textstyle \sup_{j} \E [\lVert
                      y_{j} \rVert^{4} \vert \theta_{j}, \Sigma_{j}] <
                      \infty
                      \text{, and} \hspace{4.9in}\\
    \mathrm{(iii)} & \,\, \mathrm{supp}(f_{\Sigma}) \subset \{\Sigma \in \mathcal{S}_{T}^{+} : s_{T}(\Sigma) > \underline{s}_{\Sigma} \} \text{ for some } \underline{s}_{\Sigma} > 0.
  \end{align*}\vspace*{-35pt}
\end{assumption}
To state the stationarity assumption, let $f_{\theta,\Sigma_{-T}}$ and
$f_{(\theta_{-1}', \theta_{T+1})',\Sigma_{-1}}$ denote the marginal densities
that correspond to $(\theta_{j},\Sigma_{j,-T})$ and
$((\theta_{j,-1}' , \theta_{j, T+1})',\Sigma_{j,-1})$, respectively. The
following assumption states that the distributions of
$(\theta_{j},\Sigma_{j,-T})$ and
$((\theta_{j,-1}' , \theta_{j, T+1})',\Sigma_{j,-1})$ are the same.

\begin{assumption}[Stationarity]  $f_{\theta,\Sigma_{-T}} = f_{(\theta_{-1}',\theta_{T+1})',\Sigma_{-1}}$.\label{assu:stationarity_main}
\end{assumption}
\noindent Note that this assumption does not imply that the estimates $y_{j}$
themselves are stationary, and no restriction is imposed on the joint
distribution of the mean and variance.

The following theorem shows that these assumptions ensure that
\eqref{eq:UPE_conv} holds almost surely, where the almost sure convergence is
with respect to the randomness of the sequence
$\{ ((\theta_{j}', \theta_{j,T+1})', \Sigma_{j}) \}_{j=1}^{\infty}$. An
implication is that $B(\hat{\Lambda}^{\mathrm{UPE}}, \Sigma_{j,-T})'y_{j,-1}$
obtains the oracle (asymptotic) prediction error.

\begin{theorem}\label{thm:forecasting} Under Assumptions
  \ref{assum:bounded_rand} and \ref{assu:stationarity_main},
  (\ref{eq:UPE_conv}) holds almost surely.
\end{theorem}

\section{An application to teacher value-added}
\label{sec:empirical-study}

I now apply the proposed methods to estimate TVA in the public
schools of New York City (NYC). I show that allowing value-added to vary over
time and using the URE estimators (and forecasts) lead to substantially
different empirical results compared to the conventional approach.

\subsection{Baseline model and data}
\label{sec:baseline-model-data}

I use a standard TVA model specified as the linear panel data
model introduced in \eqref{eq:model}, where $Y_{ijt}$ is the (standardized) test
score in either English Language Arts (ELA) or math, and $X_{ijt}$ is a vector
of student characteristics. The covariates include: the previous year’s test
score, gender, ethnicity, special education status (SWD), English language
learner status (ELL), and eligibility for free or reduced-price lunch
(FL).\footnote{The results are not sensitive to which covariates are included
  and/or interacted, as long as the previous year’s test scores are included.}
The teacher fixed effect $\theta_{jt}$ is indexed by time, allowing it to
vary over time. The error term $\varepsilon_{ijt}$ is assumed to be
i.i.d. across $i$, $j$, and $t$.
 
I use administrative data on all NYC public schools from academic years
2012–2013 to 2018–2019. As in \cite{bitler2019TeacherEffectsStudenta}, I
restrict the sample to 4th and 5th grade students. The analysis is based on ELA
scores, though results using math scores are similar. I further restrict the
sample to students whose ELA teachers are observed in all six years. The final
dataset includes $J = 1{,}185$ teachers across $T = 6$ years and 174{,}239
student-year observations. The coefficients are estimated using OLS with fixed
effects.%

The average number of students per teacher per year is approximately 24.5, with
a standard deviation of about 11.7. This substantial variation in class size
implies considerable heteroskedasticity in the least squares
estimates. Moreover, regressing $\hat{\theta}_{jt}$ on class size $n_{jt}$
reveals a significant positive relationship, suggesting a potential dependence
between the variance of the least squares estimator and the true fixed
effect. Also, recent work by \cite{gilraine2020NewMethodEstimatinga} has noted
that the true value-added is unlikely to follow a normal distribution. These
patterns indicate that conventional EB approaches are likely suboptimal in this
context.

\subsection{Estimation results and policy exercise}
\label{sec:estimation-results}

\begin{figure}[t]
\begin{subfigure}{.5\textwidth}
  \centering
  \includegraphics[width=\textwidth]{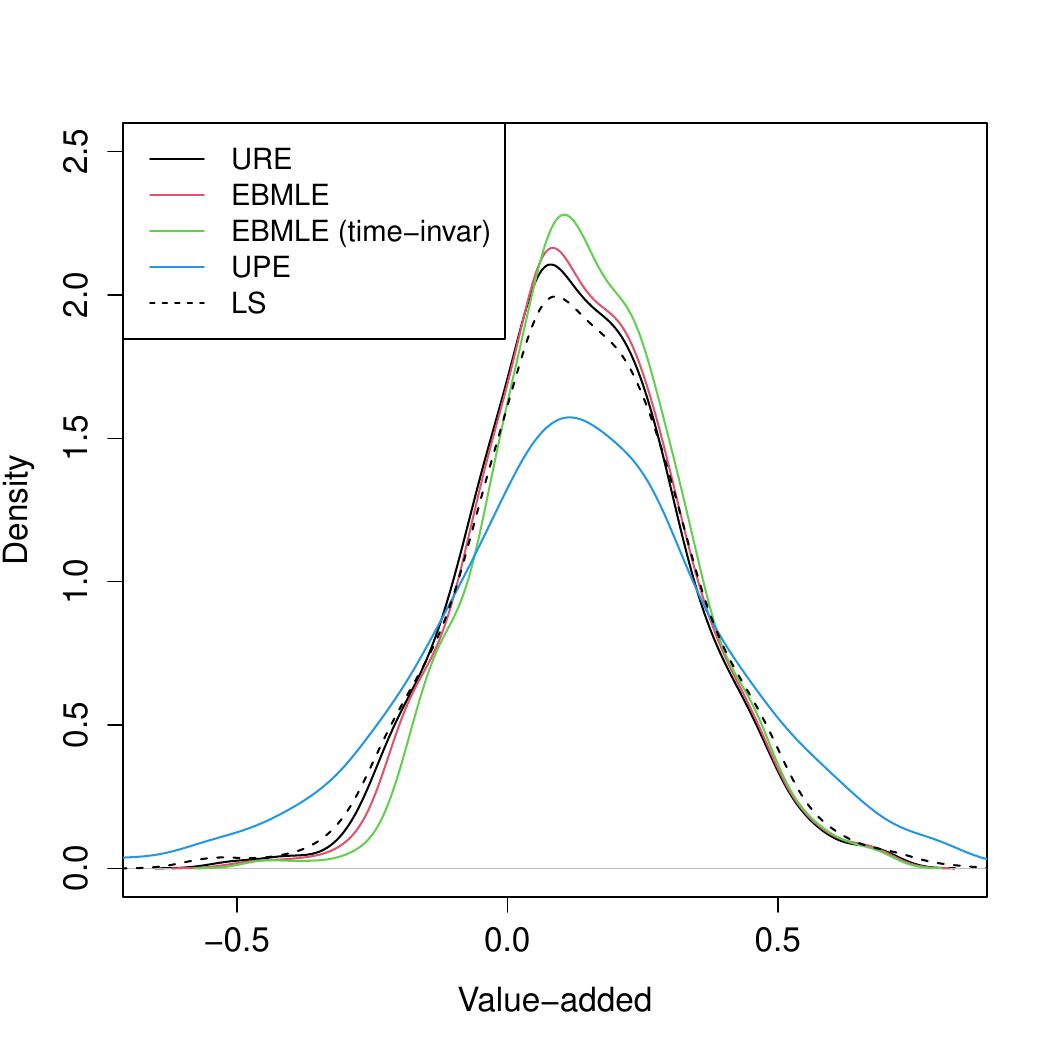}
  \caption{Density plots for shrinkage estimates}
  \label{fig:est_distn}
\end{subfigure}%
\begin{subfigure}{.5\textwidth}
  \centering
  \includegraphics[width=1\textwidth]{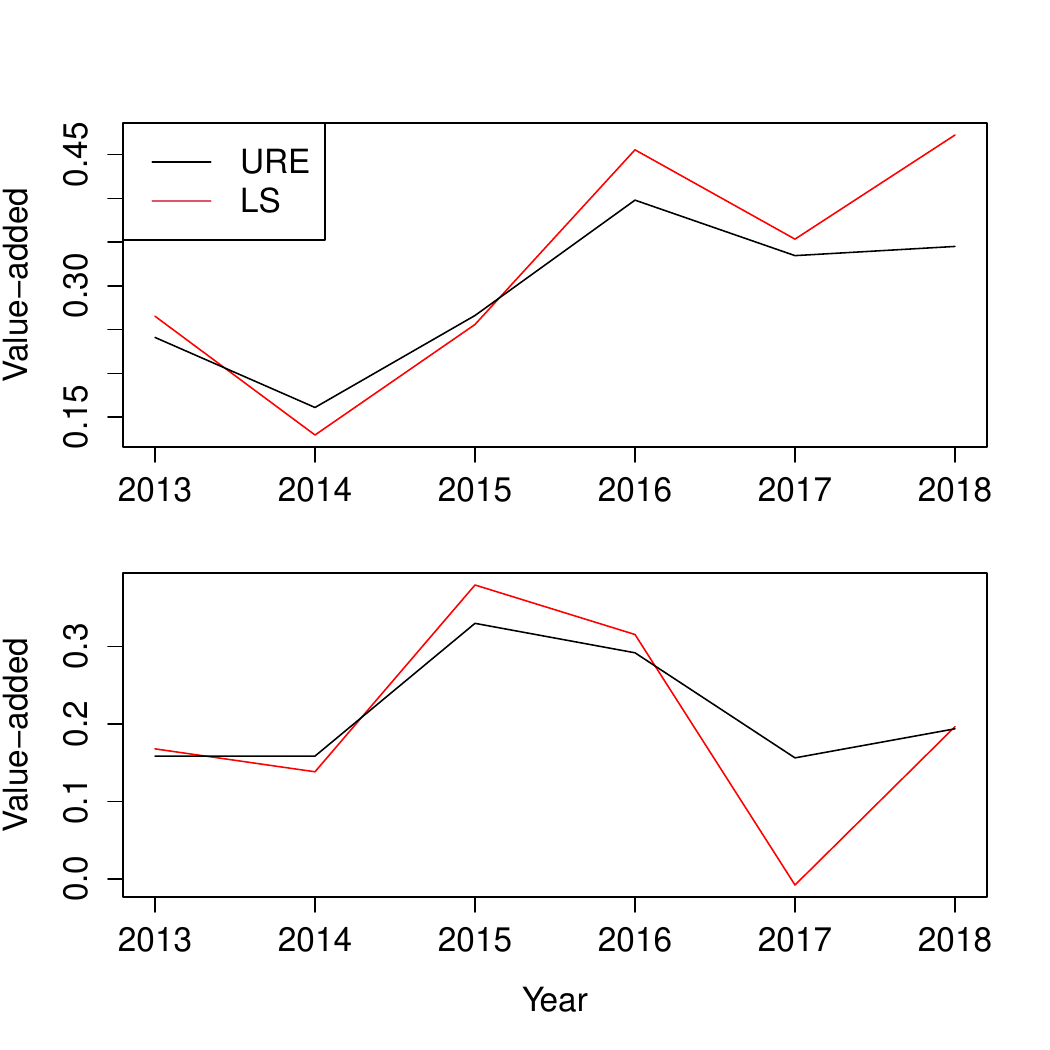}
  \caption{Shrinkage patterns %
  }
  \label{fig:shr_pattern}
  
\end{subfigure}\vspace{-10pt} \label{fig:shr_res}
\caption{Shrinkage results. The left panel shows the distribution of value-added
estimates across four estimators. The right panel shows average trajectories for
a group of teachers.}
\end{figure}

Figure \ref{fig:est_distn} shows the distribution of TVA
estimates under four different estimators: the conventional estimator (EBMLE
assuming time-invariant value-added; green), the EBMLE and URE estimators under
time-varying value-added (red and black), and the optimal UPE forecast
(blue).\footnote{For the EBMLE and URE estimators, I use
  $\mathcal{M} = \mathcal{M}_{\mathrm{g}}.$} For the time-varying estimators,
the average across time within each teacher is used for comparability. All
shrinkage-based estimators, except for the UPE forecast, yield distributions
more concentrated around the mode than the least squares estimator (black dashed
line), reflecting the effect of shrinkage. Notably, the density plots indicate
substantial differences between the conventional method and those allowing for
time variation.\footnote{The UPE forecasts are more dispersed than the other
  estimators, as they do not involve any temporal averaging. Under a
  time-invariant fixed effect model, forecasts and current estimates coincide,
  but here the divergence in distributions suggest that it may be misleading to
  use an average of past performance to predict future value-added.}

Figure \ref{fig:shr_pattern} illustrates the shrinkage pattern of the URE
estimator.\footnote{Due to confidentiality restrictions, estimates are averaged
  across a group of teachers.} Although the estimators are allowed to shrink
toward a general location ($\mathcal{M} = \mathcal{M}_{\mathrm{g}}$), the
optimal location turns out to be close to zero, so the shrinkage can be
interpreted as toward a horizontal line at zero. The URE estimator does not
shrink each point independently, but rather shrinks a smoothed version of the
trajectory, reflecting the structure of $\hat{\Lambda}^{\mathrm{URE}}$, which
has positive off-diagonal entries—consistent with positive serial correlation in
true value-added. In contrast, ignoring this correlation leads to over-shrinkage
by applying independent shrinkage at each time period. This highlights the
importance of allowing for serial dependence in the shrinkage procedure.

A standard policy simulation in the literature involves removing teachers in the
bottom 5\% of the value-added distribution and replacing them with average
teachers. I replicate this exercise, focusing on how the composition of the
bottom 5\% changes depending on the choice of estimator. Specifically, I compare
the sets of 60 teachers flagged for removal under three estimators: the
conventional time-invariant EBMLE, the URE estimator with time-varying
value-added, and the UPE forecast. Switching from the conventional estimator to
the URE changes the list by about 24\% (14 teachers), indicating that
alternative estimation strategies can significantly alter policy
outcomes.\footnote{In contrast, \cite{gilraine2020NewMethodEstimatinga} find
  that a nonparametric EB method (assuming time-invariant value-added) leads to
  minimal changes. This suggests that time variation is likely the main driver
  of these differences.}

Moreover, when the goal of the policy is to improve student outcomes in the
following year, forecasts for the next period's value-added are arguably more
informative than a summary of past performance. When the value-added is allowed
to vary with time, one can use the UPE forecasts in such context. On the other
hand, if one specifies value-added to be time-invariant, past and future
value-added are the same by definition, and thus will release the bottom 5\%
according to the conventional estimator. This consideration turns out have a
large effect, with only 25 teachers (approximately 42\%) being released under
both estimators.

I also conduct an out-of-sample policy exercise: teachers are ranked based on
estimates using data from the first five years, and the bottom 5\% are removed
under each estimator. I then evaluate performance by measuring the average
value-added in year six using the least squares estimator as a proxy for the
true value-added. Again, the composition of removed teachers differs
substantially, with only 60\% overlap between the conventional and forecast-based
methods. Importantly, the average value-added of the removed teachers is 20\%
lower when using the UPE forecasts. This suggests that the different composition
is in the correct direction, highlighting that the choice of estimator leads to
more effective policy. Consistent with this observation, the MSE of the UPE
forecasts is also 35\% lower than that of the conventional estimator, further
underscoring the importance of incorporating time variation and the optimality
of the proposed methods.
\bibliography{EB.bib}

\newpage

\begin{appendices}

\setlength{\abovedisplayskip}{8pt}
\setlength{\belowdisplayskip}{8pt}
\setlength{\abovedisplayshortskip}{8pt}
\setlength{\belowdisplayshortskip}{8pt}

\crefalias{section}{appsec}

\section{Optimality of $\hat{\theta}^{\mathrm{URE}, \mathrm{cov}}$}
\label{appsec:opt-cov}
This estimator $\hat{\theta}^{\mathrm{URE}, \mathrm{cov}}$ shrinks each
observation to a different location, $Z_{j}\gamma$, which depends on the
covariate. The key step in establishing optimality is to show
\begin{equation*}
    \left( \E  \textstyle \sup_{\gamma \in \Gamma}  \lVert \gamma \rVert^{2} \right)^{1/2} 
    \Big( \E  \textstyle \sup_{\Lambda \in \mathcal{S}_{T}^{+}} \left\lVert  \frac{1}{J} \sum_{j=1}^{J} Z_{j}' (\Lambda + \Sigma_{j})^{-1}
      \Sigma_{j}(y_{j}-\theta_{j} ) \right\rVert^{2} \Big)^{1/2} \to 0.
\end{equation*}
Again, the strategy is to show that the first term is bounded and the second
term converges to zero. Define $\varepsilon_{j} = y_{j} - \theta_{j}$. I make
the following assumptions.

\begin{assumption}[Covariates\vspace{-.17in}]\label{assu:cov}
  \begin{align*}
    \mathrm{(i)}\,\, & \,\, \textstyle \{(y_{j},Z_{j})\}_{j=1}^{J} \text{ is an
                       independent sample with } Z_{j}
                       \overset{\mathrm{i.i.d,}}{\sim} P_{Z}, \hspace{4.9in}\\
    \mathrm{(ii)}\, & \,\, \textstyle \sup_{j}s_{1}(Z_{j}'Z_{j}) < \infty
                       \text{ a.s.}, \hspace{4.9in}\\
    \mathrm{(iii)} & \,\, \E [\varepsilon_{j} \vert Z_{j}] = 0  \text{ and }
                      \var (\varepsilon_{j} \vert Z_{j}) =  \Sigma_{j},  \\
    \mathrm{(iv)}   &\,\, \mu_{Z,2} :=  \E  Z_{j}'Z_{j} \text{ is
                       nonsingular, and} \\
  \mathrm{(v)} \, & \,\, \textstyle \sup_j \E  \, [\lVert y_j \rVert^4 \vert Z_j] < \infty
                  \text{ a.s. }
  \end{align*}
\end{assumption}

A sufficient condition for (ii) is that there exists some constant
$\overline{C}_{Z} \in \mathbf{R}$ such that
$ \sup_{j,t}\lVert Z_{jt} \rVert < \overline{C}_{Z} < \infty$ almost surely,
which amounts to assuming that the covariates are uniformly bounded. The first
and second part of (iii) are exogeneity conditions for the first and second
moments of the noise term, with respect to the covariates. The full rank
condition given in (iv) is standard. The boundedness condition for the
conditional expectation given in (v) is a conditional version of Assumption
\ref{assum:bounded} (ii). Again, the boundedness conditions in (ii) and (v) can
be relaxed to a boundedness condition on the averages of the given
quantities.%

Note that there is no assumption that states a linear relationship between the
covariate matrix $Z_{j}$ and the true mean $\theta_{j}$ and/or $y_{j}$. Hence,
``misspecification'' is not a concern as long as the exogeneity
condition (ii) is met. Some specifications yield better risk properties than
others, but as long as time dummies are included in the covariates with $B$
being sufficiently large, any specification (choice of covariates) is guaranteed
to improve upon $\hat{\theta}^{\mathrm{URE}, \mathrm{g}}$ asymptotically.

Now, with some abuse of notation, I condition on a realization
$\{Z_{j}\}_{j=1}^{\infty}$ and treat the covariates as fixed. I assume that this
fixed sequence satisfies $\sup_{j} s_{1}(Z_{j}'Z_{j}) < \infty $,
$\frac{1}{J}\sum_{j=1}^{J}Z_{j}'Z_{j} \to \mu_{Z,2}$, and
$\sup_j \E  \, [\lVert y_j \rVert^4 \vert Z_j] < \infty$ which holds for
almost all realizations due to Assumption \ref{assu:cov}(ii), (iv), and
(v), and the strong law of large numbers. I directly impose these conditions on
the fixed covariates in the following theorem, with the understanding that such
conditions follow from Assumption \ref{assu:cov}. The following theorem shows
that the URE is uniformly close to the true loss function over
$(\gamma, \Lambda) \in \Gamma \times \mathcal{S}_{T}^{+}$ under this
implied assumption on the covariates, along with the maintained Assumption
\ref{assum:bounded}.

\begin{theorem}[Uniform convergence of $\mathrm{URE}^{\mathrm{cov}}(\gamma,
  \Lambda)$]\label{thm:URE_cov} Suppose that $\sup_{j}s_{1}(Z_{j}'Z_{j})
  <\infty$, $\lim_{J \to \infty}\frac{1}{J}\sum_{j=1}^{J}Z_{j}'Z_{j} =
  \mu_{Z,2}$, and
  Assumption \ref{assum:bounded} holds. Then,
  \begin{equation*}\label{eq:URE_loss_diff_cov}
    \sup_{\gamma \in \Gamma, \Lambda \in \mathcal{S}_{T}^{+}}  \left\lvert
      \mathrm{URE}^{\mathrm{cov}}(\gamma, \Lambda) - \ell(\theta, \hat{\theta}^{\mathrm{cov}}(\gamma,
      \Lambda)) \right\rvert  \overset{L^{1}}{\to} 0.
  \end{equation*}
\end{theorem}

\section{Proof of main results}
\label{sec:proof-main-theorems}

I first introduce some more notation. Define $\lambda_{k}(A)$ as the $k$th
largest eigenvalue of a square matrix $A$, so that $s_{k}(A) = \lambda_{k}(A)$
for all $k$ when $A$ is positive semidefinite. Let
$\kappa(A) = s_{1}(A)/s_{k}(A)$ be the condition number of any $k \times k$
matrix $A$. For two real symmetric matrices $A$ and $B$, I write $A \geq B$ to
denote that $A-B$ is positive semidefinite, with strict inequality meaning that
$A-B$ is positive definite. For any $d \in \mathbf{R}^{k}$, let
$\mathrm{diag}(d)$ denote the $k \times k$ diagonal matrix with diagonal
elements $d$.

Since the first term in (\ref{eq:URE_minus_loss_genmu}) is common to all the
estimators I consider, I first state a convergence result for this term.

\begin{theorem}[Uniform convergence of $\mathrm{URE}(0, \Lambda)$]\label{thm:URE} Suppose
  Assumption \ref{assum:bounded} holds. Then,
  \begin{equation}
    \label{eq:10}
    \sup_{\Lambda \in \mathcal{S}_{T}^{+}}  \big\lvert \mathrm{URE}(0,\Lambda) - \ell(\theta, \hat{\theta}(0,
      \Lambda)) \big\rvert  \overset{L^{1}}{\to} 0.
  \end{equation}
\end{theorem}

\subsection{Proof of Lemma \ref{lem:L1_oracle_generic}}
\label{appsec:proof-theorem-l1-gen}
  By definition of $\hat{\theta}^{\mathrm{URE}}(\mathcal{M}, \mathcal{L})$, I have
  $\mathrm{URE}(\hat{\theta}^{\mathrm{URE}}(\mathcal{M}, \mathcal{L}))) \leq
  \mathrm{URE}(\tilde{\theta}^{\mathrm{OL}}(\mathcal{M}, \mathcal{L})))$. This gives
  \begin{align*}
    &\ell(\theta, \hat{\theta}^{\mathrm{URE}}(\mathcal{M}, \mathcal{L}))- \ell(\theta,
      \tilde{\theta}^{\mathrm{OL}}(\mathcal{M}, \mathcal{L})) \\
    \leq & \left(  \ell(\theta, \hat{\theta}^{\mathrm{URE}}(\mathcal{M}, \mathcal{L})) -
           \mathrm{URE}(\hat{\theta}^{\mathrm{URE}}(\mathcal{M}, \mathcal{L}))  \right) + \left(  \mathrm{URE}(\tilde{\theta}^{\mathrm{OL}}(\mathcal{M}, \mathcal{L})))- \ell(\theta,
           \tilde{\theta}^{\mathrm{OL}}(\mathcal{M}, \mathcal{L})) \right) \\
    \leq & \, 2 \sup_{(\mu, \Lambda) \in
    \mathcal{M} \times \mathcal{L}} \,\,
  \lvert \mathrm{URE}(\mu, \Lambda) -  \ell(\theta,
    \hat{\theta}(\mu, \Lambda)) \rvert
  \end{align*}
  Nothing that the first line must be positive due to the definition of the
  oracle, taking expectations and then taking $\limsup_{J \to \infty}$ on all
  sides gives the conclusion.

\subsection{Proof of Theorem \ref{thm:URE}}
\label{appsec:proof-theorem-URE}
The difference between the URE and the loss is given as
\begin{equation*}
\textstyle  \mathrm{URE}(\mu, \Lambda) - \ell(\theta, \hat{\theta}(\mu, \Lambda))   =  \frac{1}{J}\sum_{j=1}^{J}\left( \mathrm{URE}_{j}(\mu, \Lambda) - (
      \hat{\theta}_{j}(\mu, \Lambda)-\theta_{j} )'(
      \hat{\theta}_{j}(\mu, \Lambda)-\theta_{j} ) \right).
\end{equation*}
Expanding the summand gives 
\begin{align}
  \begin{aligned} \label{eq:URE_minus_loss_expansion}
    &\mathrm{URE}_{j}(\mu,\Lambda) -(
    \hat{\theta}_{j}(\mu,\Lambda)-\theta_{j} )'(
    \hat{\theta}_{j}(\mu,\Lambda)-\theta_{j} ) \\
    = & y_{j}'y_{j} - \theta_{j}'\theta_{j}- \tr (\Sigma_{j})
    - 2\tr (\Lambda (\Lambda + \Sigma_{j})^{-1} (y_{j}y_{j}' -y_{j} \theta_{j}'-\Sigma_{j} ) ) \\
    & - 2\mu' (\Lambda + \Sigma_{j})^{-1} \Sigma_{j}(y_{j}-\theta_{j}).
  \end{aligned}
\end{align}
Taking $\mu = 0$, I obtain
\begin{align*}
  \begin{aligned}
    &\mathrm{URE}_{j}(\Lambda) - ( \hat{\theta}_{j}(\Lambda)-\theta_{j} )'(
    \hat{\theta}_{j}(\Lambda)-\theta_{j} ) \\
    = & y_{j}'y_{j} - \theta_{j}'\theta_{j}- \tr (\Sigma_{j}) -
    2\tr (\Lambda(\Lambda + \Sigma_{j})^{-1} (y_{j}y_{j}' -
    y_{j} \theta_{j}'-\Sigma_{j} ) ),
  \end{aligned}
\end{align*}
where, for simplicity, I write $\mathrm{URE}_{j}(\Lambda)$ as a shorthand for
$\mathrm{URE}_{j}(0,\Lambda)$, and likewise for $\mathrm{URE}(\Lambda)$ and
$\hat{\theta}(\Lambda)$. It follows that
\begin{align*}
\textstyle   \sup_{\Lambda} \left\lvert \mathrm{URE}(\Lambda) - \ell(\theta, \hat{\theta}( \Lambda)) \right\rvert  
  \leq & \textstyle \left\lvert  \frac{1}{J} \sum\nolimits_{j=1}^{J}  ( y_{j}'y_{j} -
         \theta_{j}'\theta_{j}- \tr (\Sigma_{j})) \right\rvert \\
  & + \textstyle\sup_{\Lambda}\left\lvert   \frac{1}{J} \sum\nolimits_{j=1}^{J} \tr (\Lambda(\Lambda + \Sigma_{j})^{-1}( y_{j}y_{j}' - y_{j}\theta_{j}'
    -   \Sigma_{j}) ) \right\rvert,
\end{align*}
where the inequality follows from the triangle inequality. I show that each of
the two terms in the last expression converges to zero in $L^{1}$.

For the first term, because
$ \E  y_{j}'y_{j} = \theta_{j}'\theta_{j} + \tr (\Sigma_{j})$ for
all $j \leq J$ and $y_{j}$'s are independent, I have
\begin{align*}
\textstyle  \E \left(  \frac{1}{J} \sum_{j=1}^{J}  ( y_{j}'y_{j} -
  \theta_{j}'\theta_{j}- \tr (\Sigma_{j})) \right)^{2}
  &= \textstyle \frac{1}{J^{2}} \sum_{j=1}^{J}\E ( y_{j}'y_{j} -
    \theta_{j}'\theta_{j}- \tr (\Sigma_{j}))^{2} \\
  &= \textstyle\frac{1}{J^{2}} \sum_{j=1}^{J}\var (y_{j}'y_{j}).
\end{align*}
Therefore, if
$\lim \frac{1}{J^{2}} \sum_{j=1}^{J}\var (y_{j}'y_{j}) = 0$, then this
term converges to zero in $L^{2}$ and thus in $L^{1}$. Assumption
\ref{assum:bounded}(ii) ensures that this is the case.

For the second term, note that
\begin{align*}
  & \textstyle \sup_{\Lambda}\left\lvert   \frac{1}{J} \sum_{j=1}^{J} \tr (\Lambda(\Lambda + \Sigma_{j})^{-1}( y_{j}y_{j}' - y_{j}\theta_{j}'
    -   \Sigma_{j}) ) \right\rvert \\
  \leq & \textstyle \left\lvert   \frac{1}{J} \sum_{j=1}^{J} ( y_{j}'y_{j} - \theta_{j}'y_{j}
         -   \tr (\Sigma_{j}) ) \right\rvert + \sup_{\Lambda}\left\lvert   \frac{1}{J} \sum_{j=1}^{J} \tr (\Sigma_{j}(\Lambda + \Sigma_{j})^{-1}( y_{j}y_{j}' - y_{j}\theta_{j}'
         -   \Sigma_{j}) ) \right\rvert%
     .
\end{align*}
Denote the first and second term of the right hand side as $\mathrm{(I)}_{J}$
and $\mathrm{(II)}_{J}$, respectively. To show that $\mathrm{(I)}_{J} \overset{L^{1}}{\to} 0$, I again show $L^{2}$
convergence. Because
$\E (y_{j}'y_{j}- \theta_{j}'y_{j}) = \tr (\Sigma_{j})$ for all
$j \leq J$ and $y_{j}$'s are independent, it follows that
\begin{align*}
 \textstyle \E  \left(  \frac{1}{J} \sum_{j=1}^{J} ( y_{j}'y_{j} - \theta_{j}'y_{j}
  -   \tr (\Sigma_{j}) ) \right)^{2}
  = &  \textstyle\frac{1}{J^{2}} \sum_{j=1}^{J} \E ( y_{j}'y_{j} - \theta_{j}'y_{j}
      -   \tr (\Sigma_{j}) )^{2} \\
  = & \textstyle\frac{1}{J^{2}} \sum_{j=1}^{J} \var ( y_{j}'y_{j} - \theta_{j}'y_{j}).
\end{align*}
Hence, it suffices to establish that
$\lim_{J \to \infty }\frac{1}{J^{2}}\sum_{j=1}^{J}\var (y_{j}'y_{j}-
\theta_{j}'y_{j}) = 0$. The summand is bounded by
\begin{equation*}
  \var (y_{j}'y_{j}- \theta_{j}'y_{j}) \leq 2\var (y_{j}'y_{j}) +
  2 \theta_{j}'\Sigma_{j}\theta_{j} \leq  2\var (y_{j}'y_{j}) +
  2 \tr (\Sigma_{j})\lVert \theta_{j} \rVert_{\infty}^{2}.
\end{equation*}
Hence, if
$ {\lim \sup}_{J \to \infty} \frac{1}{J}
\sum_{j=1}^{J}(\var (y_{j}'y_{j}) + \tr (\Sigma_{j})\lVert
\theta_{j} \rVert_{\infty}^{2}) < \infty $ it follows that
$\mathrm{(I)}_{J} \overset{L^{2}}{\to} 0.$ A sufficient condition for this to
hold is that $\sup_{j} \var (y_{j}'y_{j})$,
$\sup_{j}\tr (\Sigma_{j})$, and
$\sup_{j}\lVert \theta_{j} \rVert_{\infty}^{2}$ are all finite, which is true by
Assumption \ref{assum:bounded} (ii).

To show that $\mathrm{(II)}_{J} \overset{L^{1}}{\to} 0$, define the random
function ${G}_{J}(\Lambda)$ as
\begin{equation*}
 \textstyle {G}_{J}(\Lambda) = \frac{1}{J}\sum_{j=1}^{J} \tr (\Sigma_{j}(\Lambda + \Sigma_{j})^{-1}( y_{j}y_{j}' - y_{j}\theta_{j}'
  -   \Sigma_{j}) )
\end{equation*}
so that the aim is to show
$\sup_{\Lambda}\lvert {G}_{J}(\Lambda) \rvert \overset{L^{1}}{\to} 0.$ I use the
fact that convergence in probability and a uniform integrability condition imply
convergence in $L^{1}$. That is, I show
$\sup_{\Lambda}\lvert {G}_{J}(\Lambda) \rvert \overset{p}{\to} 0$ and that
$\left\{ \sup_{\Lambda}\lvert {G}_{J}(\Lambda) \rvert \right\}_{J \geq 1}$ is
uniformly integrable.

I show $\sup_{\Lambda}\lvert {G}_{J}(\Lambda) \rvert \overset{p}{\to} 0$ by
using the results given by \cite{andrews1992GenericUniformConvergence}.The
results therein require a totally bounded parameter space. However, the
parameter space in consideration, $\mathcal{S}_{T}^{+}$, does not satisfy this
requirement. This can be dealt with by an appropriate reparametrization. Let
$\underline{s}_{\Sigma} = \inf_{j} s_{T}(\Sigma_{j})$ denote the
infimum of the smallest eigenvalues of $\Sigma_{j}$'s for $ j \geq 1 $, which is
bounded away from zero by assumption. Consider the transformation defined by
$h(\Lambda) = (\underline{s}_{\Sigma}I_{T} + \Lambda)^{-1}$, and write the
image of such transformation as
$\tilde{\mathcal{L}}:= \{h(\Lambda) : \Lambda \in S^{+}_{T} \}$. Note that
$h: \mathcal{S}_{T}^{+} \to \tilde{\mathcal{L}}$ is one-to-one and onto, with its
inverse given as
$h^{-1}(\tilde{\Lambda}) = \tilde{\Lambda}^{-1}-
\underline{s}_{\Sigma} I_{T}.$ For
$\tilde{\Lambda} \in \tilde{\mathcal{L}}$, define
$\tilde{G}_{J} := G_{J}\circ h^{-1}$ so that
\begin{equation*}
  \sup_{\Lambda \in \mathcal{S}_{T}^{+}}\lvert {G}_{J}(\Lambda) \rvert =  \sup_{\Lambda
    \in \mathcal{S}_{T}^{+}}\lvert {G}_{J}(h^{-1}(h(\Lambda))) \rvert  = \sup_{\tilde{\Lambda}
    \in \tilde{\mathcal{L}}}\lvert {G}_{J}(h^{-1}(\tilde{\Lambda})) \rvert = \sup_{\tilde{\Lambda}
    \in \tilde{\mathcal{L}}}\lvert \tilde{G}_{J}(\tilde{\Lambda}) \rvert.
\end{equation*}

Hence, showing $\sup_{\Lambda}\lvert {G}_{J}(\Lambda) \rvert \overset{p}{\to} 0$
is equivalent to
$\sup_{\tilde{\Lambda} \in \tilde{\mathcal{L}}}\lvert
\tilde{G}_{J}(\tilde{\Lambda}) \rvert \overset{p}{\to} 0.$ Let
$\mathcal{S}_{T}$ denote the set of all real positive $T\times T$
matrices. While the choice of metric is irrelevant, equip $\mathcal{S}_{T}$ with
the metric $d$ induced by the Frobenius norm for concreteness. Note that
$\tilde{\mathcal{L}} \subset S_{T}$.  I show that the (reparametrized)
parameter space $\tilde{\mathcal{L}}$ is indeed totally bounded. For any
$\tilde{\Lambda} \in \tilde{\mathcal{L}}$, I have
$ 0 \leq \tilde{\Lambda} \leq \underline{s}_{\Sigma}^{-1}I_{T}$ so that
$s_{1}(\tilde{\Lambda}) \leq
\underline{s}_{\Sigma}^{-1}$. Moreover, since the largest singular value
equals the operator norm and all norms on $S_{T}$ are equivalent, this shows
that $\tilde{\mathcal{L}}$ is bounded, and thus totally bounded because
$\tilde{\mathcal{L}}$ can be seen as a subset of the Euclidean space with
dimension $T^{2}$.

It remains to show that a)
$\lvert \tilde{G}_{J}(\tilde{\Lambda}) \rvert \overset{p}{\to} 0$ for
all $\tilde{\Lambda} \in \tilde{\mathcal{L}}$ and b)
$\tilde{G}_{J}(\tilde{\Lambda})$ is stochastically equicontinuous. For
a), I can show $\lvert G_{J}(\Lambda) \rvert \overset{p}{\to} 0$ for all
$\Lambda \in \mathcal{S}_{T}^{+}$ instead because for any
$\tilde{\Lambda} \in \tilde{\mathcal{L}}$, there exists
$\Lambda \in \mathcal{S}_{T}^{+}$ such that
$G_{J}(\Lambda) = \tilde{G}_{J}(\tilde{\Lambda})$. Now, note that
\begin{equation*}
    \E  \tr (\Sigma_{j}(\Lambda + \Sigma_{j})^{-1}( y_{j}y_{j}' - y_{j}\theta_{j}'
    -   \Sigma_{j}) ) =  \tr (\Sigma_{j}(\Lambda + \Sigma_{j})^{-1}\E ( y_{j}y_{j}' - y_{j}\theta_{j}'
      -   \Sigma_{j}) ) =  0,
\end{equation*}
and $y_{j}$'s are independent. This gives
\begin{equation*}
 \textstyle \E   G_{J}(\Lambda)^{2}  = \frac{1}{J^{2}}\sum_{j=1}^{J}\E \tr (\Sigma_{j}(\Lambda + \Sigma_{j})^{-1}( y_{j}y_{j}' - y_{j}\theta_{j}'
  -   \Sigma_{j}) )^{2}
\end{equation*}
I give a bound on
$\lvert \tr (\Sigma_{j}(\Lambda + \Sigma_{j})^{-1}( y_{j}y_{j}' -
y_{j}\theta_{j}' -  \Sigma_{j})) \rvert.$ Let $UDU'$ denote the spectral decomposition of 
$\Sigma_{j}^{-1/2} \Lambda \Sigma_{j}^{-1/2}$ with
$ D = \mathrm{diag}(d_{1}, \dots, d_{T})$. Then, I have 
\begin{equation*}
  \Sigma_{j} (\Lambda +  \Sigma_{j})^{-1}= \Sigma_{j}^{1/2}U(I_{T} +D)^{-1}U'\Sigma_{j}^{-1/2}.
\end{equation*}
It follows that
\begin{align}
  \begin{aligned}\label{eq:tr_decomp}
    & \tr (\Sigma_{j}(\Lambda + \Sigma_{j})^{-1}( y_{j}y_{j}' -y_{j} \theta_{j}'
    -   \Sigma_{j})) \\
    = & \tr ( (I_{T} +D)^{-1}U'\Sigma_{j}^{-1/2}( y_{j}y_{j}' -\theta_{j}
    y_{j}' -  \Sigma_{j})\Sigma_{j}^{1/2}U).
  \end{aligned}
\end{align}
Write
$H_{j} = \Sigma_{j}^{-1/2}( y_{j}y_{j}' -y_{j} \theta_{j}' - 
\Sigma_{j})\Sigma_{j}^{1/2}$, and observe that
\begin{align}
  \begin{aligned}\label{eq:bound_on_tr}
   \textstyle  \left\lvert \tr ( (I_{T} +D)^{-1} U'\Sigma_{j}^{-1/2}( y_{j}y_{j}'
      -y_{j} \theta_{j}'
      -   \Sigma_{j})\Sigma_{j}^{1/2}U ) \right\rvert 
    = &  \textstyle \left\lvert \sum_{t=1}^{T}\frac{1}{1+d_{t}} (U'H_{j}U)_{tt} \right\rvert \\
    \leq &\textstyle \sum_{t=1}^{T} \left\lvert (U'H_{j}U)_{tt} \right\rvert,
  \end{aligned}
\end{align}
where the inequality holds because $0 \leq 1/(1+d_{t}) \leq 1$. Let $U_{t}$
denote the $t$th column of the orthogonal matrix $U$. I have
\begin{equation*}
  \left\lvert  (U'H_{j}U)_{tt}\right\rvert = \left\lvert U_{t}' H_{j} U_{t}  \right\rvert
  \leq  \left\lVert  H_{j} U_{t}  \right\rVert \leq \sup_{U \in \mathbb{R}^{T},\lVert U \rVert = 1} \lVert
  H_{j} U \rVert  = s_{1}(H_{j}),
\end{equation*}
where the first inequality follows from Cauchy-Schwarz, and the last equality
from the fact that the operator norm of a matrix is equal to its largest
singular value.

Combining these results gives
\begin{equation*}
 \textstyle \E   G_{J}(\Lambda)^{2}  \leq \frac{T^{2}}{J^{2}}
  \sum_{j=1}^{J} \E  \,s_{1}(H_{j})^{2}.
\end{equation*}
Now, to derive a bound for $ s_{1}(H_{j})$, observe that
\begin{align*}
  s_{1}(H_{j}) %
  \leq  & s_{1}(\Sigma_{j}^{-1/2})s_{1}( y_{j}y_{j}' -y_{j} \theta_{j}' - 
          \Sigma_{j})s_{1}(\Sigma_{j}^{1/2}) \\
  \leq & \kappa(\Sigma_{j})^{1/2} s_{1}( y_{j}y_{j}' -y_{j} \theta_{j}' - 
         \Sigma_{j}).
\end{align*}
Since the largest singular value of a matrix is bounded by its Frobenius norm,
it follows that
\begin{align*}
   s_{1}( y_{j}y_{j}' -y_{j} \theta_{j}' - 
    \Sigma_{j})^{2} 
  \leq & \tr (( y_{j}y_{j}' -y_{j} \theta_{j}' - 
         \Sigma_{j})'( y_{j}y_{j}' -y_{j} \theta_{j}' - 
         \Sigma_{j})) \\
  = & (y_{j}'y_{j})^{2} + \theta_{j}'y_{j} \theta_{j}'y_{j} +
      \sigma^{4}\tr (\Sigma_{j})^{2} - 2 y_{j}'y_{j} y_{j}'\theta_{j} - 2 
      y'_{j}\Sigma_{j}y_{j} + 2\theta_{j}'\Sigma_{j}y_{j}.
\end{align*}
Taking expectations yields
\begin{align*}
  &\E (y_{j}'y_{j})^{2} + \theta_{j}'\theta_{j}\E  y_{j}'y_{j} +
    \sigma^{4}\tr (\Sigma_{j})^{2} - 2\E  y_{j}'y_{j} y_{j}'\theta_{j} - 2 
    \E y'_{j}\Sigma_{j}y_{j} + 2\theta_{j}'\Sigma_{j}\E y_{j} \\
  = & \var (y_{j}'y_{j}) + 2(\theta_{j}'\theta_{j})^{2} + 3 
      \theta_{j}'\theta_{j}\tr (\Sigma_{j}) - 2 \theta_{j}'\E  (y_{j}
      y_{j}'y_{j}) \\
  \leq & \var (y_{j}'y_{j}) + 2\lVert \theta_{j} \rVert^{4} + 3 
         \lVert \theta_{j} \rVert^{2} \tr (\Sigma_{j}) + 2 \lVert \theta_{j}
         \rVert \E  \lVert y_{j} \rVert^{3}.
\end{align*}
This shows that if
\begin{equation}\label{eq:high_lev_condn_2}
 \textstyle \underset{J \to \infty}{\lim \sup}  \frac{1}{J} \sum_{j=1}^{J}  \kappa(\Sigma_{j}) \left( \var (y_{j}'y_{j}) + 2\lVert \theta_{j} \rVert^{4} + 3 
    \lVert \theta_{j} \rVert^{2} \tr (\Sigma_{j}) + 2 \lVert \theta_{j}
    \rVert \E  (\lVert y_{j} \rVert^{3}) \right) < \infty,
\end{equation}
then $\lvert G_{J}(\Lambda) \rvert \to 0$ in $L^{2}$, and thus in
probability. Hence, if $\sup_{j}s_{1}(\Sigma_{j})/s_{T}(\Sigma_{j})$,
$ \sup_{j} \var (y_{j}'y_{j})$, $\sup_{j}\lvert \theta_{j} \rvert$, and
$\sup_{j}\tr (\Sigma_{j})$ are bounded, the result holds. Note that this
true by Assumption \ref{assum:bounded}.

It remains to show that $\tilde{G}_{J}(\tilde{\Lambda})$ is
stochastically equicontinuous. I do this by showing that
$\tilde{G}_{J}(\tilde{\Lambda})$ satisfies a Lipschitz condition as in
Assumption SE-1 of \cite{andrews1992GenericUniformConvergence}. Specifically,
I show that
$\lvert \tilde{G}_{J}(\tilde{\Lambda}) -
\tilde{G}_{J}(\tilde{\Lambda}^{\dagger}) \rvert \leq B_{J}\lVert
\tilde{\Lambda} - \tilde{\Lambda}^{\dagger} \rVert $ for all
$\tilde{\Lambda}, \tilde{\Lambda}^{\dagger} \in
\tilde{\mathcal{L}} $ with $B_{J}= O_{p}(1)$.  Let
$\tilde{\Lambda}, \tilde{\Lambda}^{\dagger} \in
\tilde{\mathcal{L}}$ be arbitrarily taken. First, I show that
$\tilde{\mathcal{L}}$ is convex. Take any $\theta \in [0,1]$. Note that
$ \theta\tilde{\Lambda} + (1-\theta)\tilde{\Lambda}^{\dagger}$ is
nonsingular because $\tilde{\Lambda}$ and $\tilde{\Lambda}^{\dagger}$
are positive definite and the space of positive definite matrices is
convex. Then, for
$\Lambda_{\theta} = (\theta\tilde{\Lambda} +
(1-\theta)\tilde{\Lambda}^{\dagger})^{-1}- \underline{s}_{\Sigma}I_{T}$,
I have
$h(\Lambda_{\theta}) = \theta\tilde{\Lambda} +
(1-\theta)\tilde{\Lambda}^{\dagger},$ which shows that
$\theta\tilde{\Lambda} + (1-\theta)\tilde{\Lambda}^{\dagger} \in
\tilde{\mathcal{L}}$. The mean value theorem gives
\begin{equation*}
  \tilde{G}_J(\tilde{\Lambda}) - \tilde{G}_J(\tilde{\Lambda}^{\dagger})
  = \nabla \tilde{G}_J(\tilde{\Lambda}^{\theta}) \cdot
  \mathrm{vec}(\tilde{\Lambda} - \tilde{\Lambda}^{\dagger}),
\end{equation*}
where
$\nabla \tilde{G}_J(\tilde{\Lambda}):= \frac{\partial}{\partial
  \mathrm{vec}(\tilde{\Lambda})}\tilde{G}_J(\tilde{\Lambda})$, and
$\tilde{\Lambda}^{\theta} := \theta \tilde{\Lambda} +
(1-\theta)\tilde{\Lambda}^{\dagger}$ for some $\theta \in [0,1]$. This
implies, by Cauchy-Schwarz,
\begin{equation}\label{eq:Lip_condn_G}
  \lvert \tilde{G}_J(\tilde{\Lambda}) - \tilde{G}_J(\tilde{\Lambda}^{\dagger}) \rvert
  \leq \lVert \nabla \tilde{G}_J(\tilde{\Lambda}^{\theta}) \rVert 
  \lVert  \tilde{\Lambda} - \tilde{\Lambda}^{\dagger} \rVert,
\end{equation}
where I use the fact that the Frobenius norm of a matrix and the Euclidean norm
of the vectorized version of it are the same. Note that
$\lVert \nabla \tilde{G}_J(\tilde{\Lambda}) \rVert= \lVert
\frac{\partial}{\partial
  \tilde{\Lambda}}\tilde{G}_J(\tilde{\Lambda}) \rVert$ by definition
of the Frobenius norm.

By the formula for the derivative of a matrix inverse and the derivative of a
trace, and the chain rule for matrix derivatives, I have
\begin{equation*}
   \textstyle \frac{\partial}{\partial
    \tilde{\Lambda}} G_{J}(\tilde{\Lambda}) 
  = \textstyle \frac{1}{J}\sum_{j=1}^{J}
      \tilde{\Lambda}^{-1}(\tilde{\Lambda}^{-1}-\underline{s}_{\Sigma}I_{T}+\Sigma_{j})^{-1} \Sigma_{j} ( y_{j}y_{j}' - y_{j}\theta_{j}'
      -   \Sigma_{j})(\tilde{\Lambda}^{-1}-\underline{s}_{\Sigma}I_{T}+\Sigma_{j})^{-1}\tilde{\Lambda}^{-1}.
\end{equation*}
Write the summand in the second line as $g_{j}(\tilde{\Lambda})$. I first
derive a bound on $s_{1}(g_{j}(\tilde{\Lambda}))$, and use this to
bound
$\lVert \textstyle \frac{\partial}{\partial \tilde{\Lambda}}
G_{J}(\tilde{\Lambda}) \rVert$ by using the fact that
\begin{equation*}
\textstyle \left\lVert    \frac{\partial}{\partial
     \tilde{\Lambda}} G_{J}(\tilde{\Lambda}) \right\rVert  \leq  T^{1/2} s_{1}\left(  \frac{\partial}{\partial
      \tilde{\Lambda}} G_{J}(\tilde{\Lambda}) \right) \leq \frac{1}{J}\sum_{j=1}^{J}s_{1}(g_{j}(\tilde{\Lambda})).
\end{equation*}
Since the operator norm is submultiplicative, it follows that
\begin{equation*}
  s_{1}(g_{j}(\tilde{\Lambda})) \leq s_{1}(\tilde{\Lambda}^{-1}(\tilde{\Lambda}^{-1}-\underline{s}_{\Sigma}I_{T}+\Sigma_{j})^{-1})^{2}s_{1}(\Sigma_{j} ( y_{j}y_{j}' - y_{j}\theta_{j}'
  -   \Sigma_{j})).
\end{equation*}
I proceed by bounding the two singular values that appear on the right hand
side. For the first term, note that
\begin{align}
  \begin{aligned}
    s_{1}(\tilde{\Lambda}^{-1}(\tilde{\Lambda}^{-1}-\underline{s}_{\Sigma}I_{T}+\Sigma_{j})^{-1})^{2} 
    =  s_{1}((I+\tilde{\Lambda}^{1/2}(\Sigma_{j} -
    \underline{s}_{\Sigma}I_{T})\tilde{\Lambda}^{1/2})^{-2}) 
    \leq  1,\label{eq:bound_on_LambdtilAtil}
  \end{aligned}
\end{align}
where the inequality follows because
$\tilde{\Lambda}^{1/2}(\Sigma_{j} -
\underline{s}_{\Sigma}I_{T})\tilde{\Lambda}^{1/2}$ is positive
semidefinite so that
$0 \leq (I+\tilde{\Lambda}^{1/2}(\Sigma_{j} -
\underline{s}_{\Sigma}I_{T})\tilde{\Lambda}^{1/2})^{-2} \leq I_{T}$, and
$ A \leq B$ implies $s_{1}(A) \leq s_{1}(B)$ for any two positive
semidefinite matrices $A$ and $B$. A bound on
$s_{1}(\Sigma_{j} ( y_{j}y_{j}' - y_{j}\theta_{j}' -  \Sigma_{j}))$
is given by
\begin{equation*}
  s_{1}(\Sigma_{j} ( y_{j}y_{j}' - y_{j}\theta_{j}'  -   \Sigma_{j})) \leq s_{1}(\Sigma_{j})(y_{j}'y_{j} +
  (y_{j}'y_{j})^{1/2}(\theta_{j}'\theta_{j})^{1/2} +   s_{1}(\Sigma_{j})).
\end{equation*}
Combining these results, I obtain
\begin{align*}
\textstyle  \sup_{\tilde{\Lambda} \in \tilde{\mathcal{L}}} \, \left\lVert  \frac{\partial}{\partial
  \tilde{\Lambda}} G_{J}(\tilde{\Lambda}) \right\rVert %
  \leq & \textstyle
      \frac{1}{J} \sum_{j=1}^{J} s_{1}(\Sigma_{j})\left(y_{j}'y_{j} +
      \lVert y_{j} \rVert\lVert \theta_{j} \rVert - \E (y_{j}'y_{j} +
      \lVert y_{j} \rVert\lVert \theta_{j} \rVert)\right) \\
                                                              & \textstyle +\frac{1}{J} \sum_{j=1}^{J} \left( \E (y_{j}'y_{j} +
                                                                \lVert y_{j} \rVert\lVert \theta_{j} \rVert) +  s_{1}(\Sigma_{j}) \right).
\end{align*}
The first term on the right hand side is $o_{p}(1)$ because
\begin{equation*}
  \textstyle \sup_{j}\var (s_{1}(\Sigma_{j})(y_{j}'y_{j} + \lVert y_{j} \rVert \lVert
  \theta_{j} \rVert)) \leq 2 \textstyle\sup_{j}s_{1}^{2}(\Sigma_{j})(\E \lVert y_{j} \rVert^{4}
  + \lVert \theta_{j} \rVert^{2}\E \lVert y_{j} \rVert^{2}) < \infty.
\end{equation*}
The term in the last line is bounded as $J \to \infty$ because the summand is
bounded uniformly over $j$. This shows that
$B_{J} := \sup_{\tilde{\Lambda} \in \tilde{\mathcal{L}}} \lVert
\textstyle \frac{\partial}{\partial \tilde{\Lambda}}
G_{J}(\tilde{\Lambda}) \rVert = O_{p}(1)$. Combining this with
\eqref{eq:Lip_condn_G} gives
\begin{equation*}
  \lvert \tilde{G}_J(\tilde{\Lambda}) - \tilde{G}_J(\tilde{\Lambda}^{\dagger}) \rvert
  \leq  B_{J} 
  \lVert  \tilde{\Lambda} - \tilde{\Lambda}^{\dagger} \rVert,
\end{equation*}
for all $\Lambda, \tilde{\Lambda} \in \tilde{\mathcal{L}}$ and
$B_{J} = O_{p}(1)$, which establishes the desired Lipschitz condition. This
completes the proof for
$\sup_{\Lambda \in \mathcal{S}_{T}^{+}}\lvert {G}_{J}(\Lambda) \rvert \overset{p}{\to}
0$.

Now, to strengthen the convergence in probability to convergence in $L^{1}$, I
show that
$\left\{ \sup_{\Lambda} \lvert {G}_{J}(\Lambda) \rvert \right\}_{J \leq 1}$ is
uniformly integrable. A bound on
$\sup_{\Lambda} \lvert {G}_{J}(\Lambda) \rvert$ is given by
\begin{align*}
 \sup_{\Lambda} \lvert {G}_{J}(\Lambda) \rvert & =  \sup_{\Lambda} \left\lvert \textstyle \frac{1}{J}\sum_{j=1}^{J} \tr (\Sigma_{j}(\Lambda + \Sigma_{j})^{-1}( y_{j}y_{j}' - y_{j}\theta_{j}'
                                                             -   \Sigma_{j}) ) \right\rvert \\
                                                           & \leq \textstyle \frac{1}{J}\sum_{j=1}^{J}  \sup_{\Lambda} \left\lvert\tr (\Sigma_{j}(\Lambda + \Sigma_{j})^{-1}( y_{j}y_{j}' - y_{j}\theta_{j}'
                                                             -   \Sigma_{j}) ) \right\rvert \\
                                                           & \leq 
                                                             \textstyle\frac{T}{J}\sum_{j=1}^{J}\kappa(\Sigma_{j})(y_{j}'y_{j}
                                                             + \lVert \theta_{j} \rVert \lVert y_{j} \rVert + s_{1}(\Sigma_{j}))
\end{align*}
where the last inequality follows from \eqref{eq:tr_decomp} and
\eqref{eq:bound_on_tr}. Let $\overline{G}_{J}$ denote the expression in the last
line, and suppose that
${\lim \sup}_{{J \to \infty}}\E  \overline{G}_{J}^{2} < \infty $, which
I verify below. Then, I have
$\sup_{J} \E (\sup_{\Lambda} \lvert {G}_{J}(\Lambda) \rvert)^{^{2}} <
\infty$, from which the uniform integrability follows.  It remains only to show
that ${\lim \sup}_{{J \to \infty}}\E  \overline{G}_{J}^{2} < \infty
$. By Cauchy-Schwarz, it follows that
\begin{equation*}
\textstyle  \E  \,\overline{G}_{J}^{2}   \leq
  \frac{T^{2}}{J}\sum_{j=1}^{J}\E  \left( \kappa(\Sigma_{j})(y_{j}'y_{j}
  + \lVert \theta_{j} \rVert \lVert y_{j} \rVert + s_{1}(\Sigma_{j})) \right)^{2},
\end{equation*}
and the term in the summand is uniformly bounded over $j \geq 1$. This
establishes ${\lim \sup}_{{J \to \infty}}\E  \, \overline{G}_{J}^{2} < \infty $,
and thus that
$\left\{ \sup_{\Lambda} \lvert {G}_{J}(\Lambda) \rvert \right\}_{J \leq 1}$ is
uniformly integrable. This concludes the proof.

\subsection{Proof of Theorem \ref{thm:URE_opt_m_g} (i)}
\label{appsec:proof-theorem-URE_m}

The convergence result to be
established is
\begin{equation*}
  \sup_{\Lambda \in \mathcal{S}_{T}^{+}} \left\lvert
         \frac{1}{J} \sum\nolimits_{j=1}^{J} \overline{y}_{J}' (\Lambda + \Sigma_{j})^{-1}
    \Sigma_{j}(y_{j}-\theta_{j}) \right\rvert \overset{L^{1}}{\to} 0.
\end{equation*}
Two applications of the Cauchy-Schwarz inequality show that the expectation of
the left-hand side is bounded by
\begin{equation*}
  \left( \E   \lVert \overline{y}_{J} \rVert^{2} \right)^{1/2} 
    \Big( \E  \textstyle \sup_{\Lambda} \left\lVert  \frac{1}{J} \sum_{j=1}^{J}  (\Lambda + \Sigma_{j})^{-1}
      \Sigma_{j}(y_{j}-\theta_{j} ) \right\rVert^{2} \Big)^{1/2}.
\end{equation*}
The limit supremum, as $J \to \infty$, of the first term is bounded by
Assumption \ref{assum:bounded} (ii). The fact that the second term converges to
$0$ has been shown in the proof of Theorem \ref{thm:URE}, given in Appendix \ref{appsec:proof-theorem-URE}. 

\subsection{Proof of Theorem \ref{thm:URE_opt_m_g} (ii)}
\label{appsec:proof-theorem-URE_g}

All supremums over $\mu$ are understood to be taken over $\mathcal{B}$,
though for simplicity I write $\sup_{\mu}$. Observe that, by
\eqref{eq:URE_minus_loss_expansion},
\begin{align*}
  &\sup_{\mu, \Lambda} \big(
    \mathrm{URE}(\mu, \Lambda) - \ell(\theta, \hat{\theta}(\mu,
    \Lambda)) \big) \\
  \leq & \sup_{\Lambda} \big(
         \mathrm{URE}( \Lambda) - \ell(\theta, \hat{\theta}(\Lambda))
         \big)
       \textstyle  + \sup_{\mu, \Lambda} \big(  \frac{1}{J} \sum_{j=1}^{J} \mu' (\Lambda + \Sigma_{j})^{-1}
         \Sigma_{j}(y_{j}-\theta_{j}) \big).
\end{align*}
Since Theorem \ref{thm:URE} shows that the first term on the right-hand side
converges to zero in $L^{1}$, it now remains to show
\begin{equation}\label{eq:centering_approx}
\textstyle  \sup_{\mu,\Lambda} \left\lvert  \frac{1}{J} \sum_{j=1}^{J} \mu' (\Lambda + \Sigma_{j})^{-1}
    \Sigma_{j}(y_{j}-\theta_{j}) \right\rvert \overset{L^{1}}{\to} 0.
\end{equation}

By two applications of Cauchy-Schwarz, I have
\begin{align}
  \begin{aligned}
    & \textstyle \,\E  \sup_{\mu,\Lambda} \left\lvert  \frac{1}{J} \sum_{j=1}^{J} \mu' (\Lambda + \Sigma_{j})^{-1}
      \Sigma_{j}(y_{j}-\theta_{j}) \right\rvert \\
    \leq & \textstyle \left( \E  \textstyle \sup_{\mu} \lVert \mu \rVert^{2} \right)^{1/2} 
    \left( \E   \sup_{\Lambda} \left\lVert  \frac{1}{J} \sum_{j=1}^{J}  (\Lambda + \Sigma_{j})^{-1}
      \Sigma_{j}(y_{j}-\theta_{j} ) \right\rVert^{2} \right)^{1/2}.\label{eq:two_app_cs}
  \end{aligned}
\end{align}
Note that
${\limsup}_{J \to \infty} \E  \textstyle \sup_{\mu} \lVert \mu \rVert^{2} <
\infty$ by Assumption \ref{assu:bdd_quan}. It suffices to show
\begin{equation*}
\textstyle  \lim_{J \to \infty}  \E \sup_{\Lambda} \left\lVert  \frac{1}{J} \sum_{j=1}^{J}  (\Lambda + \Sigma_{j})^{-1}
    \Sigma_{j}(y_{j}-\theta_{j} ) \right\rVert^{2} = 0,
\end{equation*}
from which then \eqref{eq:centering_approx} will follow.

Write $H_{J}(\Lambda) :=  \lVert \frac{1}{J} \sum_{j=1}^{J} (\Lambda + \Sigma_{j})^{-1}
\Sigma_{j}(y_{j}-\theta_{j} ) \rVert$. As in the proof for Theorem \ref{thm:URE}, I
show (a) $\sup_{\Lambda} H_{J}(\Lambda) \overset{p}{\to} 0$ and (b) $\sup_{J} \E  \,(\sup_{\Lambda} H_{J}(\Lambda))^{2+\delta} < \infty$
for some $\delta > 0$. Since (b) is a sufficient condition for $\left\{
  \sup_{\Lambda}H_{J}(\Lambda)^{2} \right\}_{J \geq 1}$ being uniformly integrable, (a) and (b)
together imply $ \sup_{\Lambda }H_{J}(\Lambda) \overset{L^{2}}{\to} 0.$

I show $\sup_{\Lambda}H_{J}(\Lambda) \overset{p}{\to} 0$ by again using a ULLN
argument as in \cite{andrews1992GenericUniformConvergence}. First, to show
$H_{J}(\Lambda) \overset{p}{\to} 0$, it is enough to show $\E  \,H_{J}(\Lambda)^{2}
\to 0.$ Note that
\begin{align*}
 \textstyle  \E H_{J}(\Lambda)^{2} %
                               &=   \textstyle\frac{1}{J^{2}} \E   \left(\sum_{j=1}^{J} (y_{j}-\theta_{j} )'\Sigma_{j}(\Lambda + \Sigma_{j})^{-2}
                                 \Sigma_{j}(y_{j}-\theta_{j} )\right) \\
                               & \textstyle\leq  \frac{1}{J^{2}}  \sum_{j=1}^{J} \tr (\Sigma_{j}),
\end{align*}
where the inequality follows from von Neumann's trace inequality and the fact
that $s_{1}(\Sigma_{j}(\Lambda + \Sigma_{j})^{-2}
\Sigma_{j}) \leq 1. $ Moreover, by Assumption \ref{assum:bounded}, I have $\sup_{j}\tr (\Sigma_{j})\leq T
\sup_{j}s_{1}(\Sigma_{j}) < \infty$, which implies $\frac{1}{J^{2}}  \sum_{j=1}^{J}
\tr (\Sigma_{j}) \to 0$. This establishes that $H_{J}(\Lambda)$
converges to zero in $L^{2}$, and thus in probability.

To show that this
convergence is uniform over $\Lambda \in \mathcal{S}_{T}^{+}$, by a similar argument as in
the proof of Theorem \ref{thm:URE}, it suffices to show that
$\tilde{H}_{J} := H_{J}\circ h^{-1}$ satisfies a Lipschitz condition, i.e.,
\begin{equation}\label{eq:Lip_condn_cent}
  \lvert \tilde{H}_{J}(\tilde{\Lambda}) -
  \tilde{H}_{J}(\tilde{\Lambda}^{\dagger}) \rvert \leq B_{H,J} \lVert
  \tilde{\Lambda} - \tilde{\Lambda}^{\dagger} \rVert
\end{equation}
for all $\tilde{\Lambda}, \tilde{\Lambda}^{\dagger}\in
\tilde{\mathcal{L}}$, where $B_{H,J} = O_{p}(1)$. Define
$\tilde{A}_{j} = \tilde{\Lambda}^{-1} +( \Sigma_{j} -
\underline{s}_{\Sigma}I_{T})$ and $\tilde{A}_{j}^{\dagger}$ likewise with
$\tilde{\Lambda}$ replaced with $\tilde{\Lambda}^{\dagger}$. By the
triangle inequality (and its reverse), we have
\begin{align}
  \begin{aligned}
    \lvert \tilde{H}_{J}(\tilde{\Lambda}) -
    \tilde{H}_{J}(\tilde{\Lambda}^{\dagger}) \rvert%
    \leq & \textstyle \frac{1}{J}
    \sum_{j=1}^{J}s_{1}(\tilde{A}_{j}^{-1} - \tilde{A}_{j}^{\dagger
      -1} ) \lVert \Sigma_{j}(y_{j}-\theta_{j} ) \rVert,\label{eq:Lips_lambdatil}
  \end{aligned}
\end{align}
where the first inequality follows from the reverse triangle inequality and the
second by the triangle inequality and the definition of the operator
norm. Observe that
\begin{equation}
    \tilde{A}_{j}^{-1} - \tilde{A}^{\dagger -1}_{j} =%
    \tilde{A}^{\dagger-1}_{j}
    \tilde{\Lambda}^{-1}(\tilde{\Lambda} - \tilde{\Lambda}^{\dagger}
    )\tilde{\Lambda}^{\dagger -1}\tilde{A}_{j}^{-1},\label{eq:Lambda_til_decomp}
\end{equation}
which implies
\begin{align}
  \begin{aligned}
    s_{1}(\tilde{A}_{j}^{-1} - \tilde{A}^{\dagger -1}_{j}) %
    \leq & \textstyle ( \sup_{\tilde{\Lambda} \in \tilde{\mathcal{L}}}
    s_{1}((\tilde{\Lambda}^{-1} +( \Sigma_{j} -
    \underline{s}_{\Sigma}I_{T}))^{-1}\tilde{\Lambda}^{-1})^{2} ) s_{1}(\tilde{\Lambda}-
    \tilde{\Lambda}^{\dagger}).\label{eq:bound_on_Atil}
  \end{aligned}
\end{align}
The inequality follows from \eqref{eq:Lambda_til_decomp}, the fact that the
operator norm is submultiplicative and the fact that
$s_{1}(C) = s_{1}(C')$ for any matrix $C$. Furthermore, I have shown
in \eqref{eq:bound_on_LambdtilAtil} that
$s_{1}((\tilde{\Lambda}^{-1} +( \Sigma_{j} -
\underline{s}_{\Sigma}I_{T}))^{-1}\tilde{\Lambda}^{-1})^{2}$ is bounded
above by $1$. Hence, I obtain
\begin{equation*}
  s_{1}(\tilde{A}_{j}^{-1} - \tilde{A}^{\dagger -1}_{j}) \leq
  s_{1}(\tilde{\Lambda} - \tilde{\Lambda}^{\dagger}) \leq \lVert
  \tilde{\Lambda} - \tilde{\Lambda}^{\dagger} \rVert,
\end{equation*}
where the last inequality follows from the fact that the operator norm of a
matrix is less than or equal to its Frobenius norm.

Plugging this bound into
\eqref{eq:Lips_lambdatil}, it follows that
\begin{equation*}
  \lvert \tilde{H}_{J}(\tilde{\Lambda}) -
  \tilde{H}_{J}(\tilde{\Lambda}^{\dagger}) \rvert \leq \textstyle  \left(\frac{1}{J}
  \sum_{j=1}^{J} \lVert \Sigma_{j}(y_{j}-\theta_{j} ) \rVert\right) \lVert
  \tilde{\Lambda} - \tilde{\Lambda}^{\dagger} \rVert.
\end{equation*}
Therefore, it remains to show
$\frac{1}{J} \sum_{j=1}^{J} \lVert \Sigma_{j}(y_{j}-\theta_{j} ) \rVert = O_{p}(1)$
to establish \eqref{eq:Lip_condn_cent}. Observe that
\begin{equation*}
  \textstyle \textstyle \frac{1}{J} \sum_{j=1}^{J} \lVert \Sigma_{j}(y_{j}-\theta_{j} ) \rVert 
  =   \textstyle \frac{1}{J}
      \sum_{j=1}^{J} (\lVert \Sigma_{j}(y_{j}-\theta_{j} ) \rVert - \E  \lVert
      \Sigma_{j}(y_{j}-\theta_{j} ) \rVert ) + \frac{1}{J}
      \sum_{j=1}^{J} \E \lVert \Sigma_{j}(y_{j}-\theta_{j} ) \rVert.
\end{equation*}
Since
$ \sup_{j}\var (\lVert \Sigma_{j}(y_{j}-\theta_{j} ) \rVert) \leq \sup_{j} E
\lVert \Sigma_{j}(y_{j}-\theta_{j} ) \rVert^{2} = \sup_{j} \tr (\Sigma_{j})^{3}
< \infty $ the first term on the right hand side converges to zero in
probability by an application of Chebyshev's inequality. Also, because
$\sup_{j}\E \lVert \Sigma_{j}(y_{j}-\theta_{j} ) \lvert \leq \sup_{j}
s_{1}(\Sigma_{j})( \E \lVert y_{j} \rVert + \lVert \theta_{j} \rVert)
<\infty$ by Assumption \ref{assum:bounded}, the second term is $O(1)$. This
establishes
$\frac{1}{J} \sum_{j=1}^{J} \lVert \Sigma_{j}(y_{j}-\theta_{j} ) \rVert =
O_{p}(1)$, and thus $\sup_{\Lambda}H_{J}(\Lambda) \overset{p}{\to} 0$.

Now, to show that $\sup_{\Lambda}H_{J}(\Lambda)$ converges to zero in $L^{2}$,
it is enough to show that
$\left\{ \sup_{\Lambda}H_{J}(\Lambda)^{2} \right\}_{J\leq 1}$ is uniformly
integrable. A sufficient condition for this to hold is
\begin{equation*}
  \textstyle \sup_{J} \E \sup_{\Lambda}H_{J}(\Lambda)^{2+\delta} < \infty,
\end{equation*}
for some $\delta > 0$. First, I derive an upper bound of $H_{J}(\Lambda)$,
\begin{align*}
  H_{J}(\Lambda) %
  \leq \textstyle  \frac{1}{J} \sum_{j=1}^{J} s_{1}((\Lambda + \Sigma_{j})^{-1}
         \Sigma_{j}))\lVert y_{j}-\theta_{j} \rVert 
  \leq  \textstyle \frac{1}{J} \sum_{j=1}^{J}\lVert y_{j}-\theta_{j} \rVert,
\end{align*}
where the first inequality follows from the triangle inequality and the
definition of the operator norm, and the second inequality follows because
\begin{equation*}
  s_{1}((\Lambda + \Sigma_{j})^{-1} \Sigma_{j})^{2} = s_{1}(\Sigma_{j}(\Lambda + \Sigma_{j})^{-2}
  \Sigma_{j}) = s_{1}(I_{T}+\Sigma_{j}^{-1/2}\Lambda \Sigma_{j}^{-1/2})\leq 1.
\end{equation*}
Therefore, I have
\begin{equation}
\textstyle    \sup_{\Lambda} H_{J}(\Lambda)^{2+\delta} %
    \leq  \frac{1}{J} \sum_{j=1}^{J}\lVert y_{j}-\theta_{j}
    \rVert^{2+\delta}    \leq  \frac{1}{J} \sum_{j=1}^{J}2^{1+\delta}(\lVert
    y_{j}\rVert^{2+\delta} + \lVert\theta_{j} \rVert^{2+\delta}),\label{eq:unif_int_center}
\end{equation}
where the first inequality follows from Jensen's inequality, and the last
inequality follows from the triangle inequality and the fact that
$(a+b)^{p} \leq 2^{p-1}(a^{p}+b^{p})$ for any $a,b \geq 0$ and $p \geq
1$. Taking expectations shows that, for any $\delta \in [0, 2]$,
\begin{equation*}
  \textstyle \limsup_{J}  \E  \sup_{\Lambda} H_{J}(\Lambda)^{2+\delta} < \infty,
\end{equation*}
and thus $\sup_{J}  \E  \sup_{\Lambda} H_{J}(\Lambda)^{2+\delta} < \infty$.

\subsection{Proof of Theorem \ref{thm:URE_cov}}
\label{appsec:proof-theorem-URE_cov}

By the same calculation given in \eqref{eq:two_app_cs}, it is
enough to show
\begin{align*} 
& \textstyle \underset{J \to \infty}{\limsup} \, \E   \sup_{\gamma \in \Gamma} \lVert \gamma \rVert^{2} <
    \infty  \text{ and } \\ & \textstyle \lim_{J \to \infty}  \E   \sup_{\Lambda} \left\lVert  \frac{1}{J} \sum_{j=1}^{J}  Z_{j}'(\Lambda + \Sigma_{j})^{-1}
    \Sigma_{j}(y_{j}-\theta_{j} ) \right\rVert^{2} = 0,
\end{align*}
The first inequality is equivalent to showing
\begin{equation}\label{eq:cov_equiv_bound}
  \underset{J \to \infty}{\limsup} \, \textstyle \E  (\sum_{j=1}^{J}y_{j}'Z_{j})
  (\sum_{j=1}^{J}Z_{j}'Z_{j})^{-2}(\sum_{j=1}^{J}Z_{j}'y_{j}) < \infty.
\end{equation}
Simple algebra show that
\begin{align*}
  & \textstyle \E  (\sum_{j=1}^{J}y_{j}'Z_{j})
    (\sum_{j=1}^{J}Z_{j}'Z_{j})^{-2}(\sum_{j=1}^{J}Z_{j}'y_{j}) \\
  = & \textstyle (\sum_{j=1}^{J} \theta_{j}'Z_{j})
      (\sum_{j=1}^{J}Z_{j}'Z_{j})^{-2}(\sum_{j=1}^{J}Z_{j}'\theta_{j}) + \E  (\sum_{j=1}^{J}\varepsilon_{j}'Z_{j})
      (\sum_{j=1}^{J}Z_{j}'Z_{j})^{-2}(\sum_{j=1}^{J}Z_{j}' \varepsilon_{j}),
\end{align*}
where the last equality follows because the ``cross terms'' are zero due to the
conditional mean independence assumption. I show that the first and second term
of the last line is $O(1)$ and $o(1)$, respectively, which in turn will imply
\eqref{eq:cov_equiv_bound}. Note that
\begin{align*}
  \textstyle  \lVert (\sum_{j=1}^{J}
    \theta_{j}'Z_{j})(\sum_{j=1}^{J}Z_{j}'Z_{j})^{-2}(\sum_{j=1}^{J}Z_{j}'\theta_{j})
    \rVert%
  \leq  \textstyle   s_{1}((\frac{1}{J}\sum_{j=1}^{J}Z_{j}'Z_{j})^{-2})
         (\frac{1}{J}\sum_{j=1}^{J}s_{1}(Z_{j}) \lVert\theta_{j} \rVert)^{2},
\end{align*}
with
$s_{1}((\frac{1}{J}\sum_{j=1}^{J}Z_{j}'Z_{j})^{-2}) \to
s_{1}((\E Z_{j}'Z_{j})^{-2})$ and
$\limsup_{J\to \infty}\frac{1}{J}\sum_{j=1}^{J}s_{1}(Z_{j})
\lVert\theta_{j} \rVert < \infty$. This shows that
$\limsup_{J\to \infty}(\sum_{j=1}^{J} \theta_{j}'Z_{j})
(\sum_{j=1}^{J}Z_{j}'Z_{j})^{-2}(\sum_{j=1}^{J}Z_{j}'\theta_{j}) < \infty$.

It remains to show
$E (\sum_{j=1}^{J}\varepsilon_{j}'Z_{j})
(\sum_{j=1}^{J}Z_{j}'Z_{j})^{-2}(\sum_{j=1}^{J}Z_{j}' \varepsilon_{j}) \to 0.$
Because
\begin{align*}
  & \textstyle (\sum_{j=1}^{J}\varepsilon_{j}'Z_{j})
    (\sum_{j=1}^{J}Z_{j}'Z_{j})^{-2}(\sum_{j=1}^{J}Z_{j}' \varepsilon_{j})   %
  =  \textstyle  \tr ((\sum_{j=1}^{J}Z_{j}'Z_{j})^{-2}(\sum_{j=1}^{J}\sum_{\ell=1}^{J}Z_{j}'
      \varepsilon_{j}\varepsilon_{\ell}'Z_{\ell})),
\end{align*}
it follows that
\begin{align*}
  \textstyle \E  (\sum_{j=1}^{J}\varepsilon_{j}'Z_{j})
    (\sum_{j=1}^{J}Z_{j}'Z_{j})^{-2}(\sum_{j=1}^{J}Z_{j}' \varepsilon_{j}) 
  =  \textstyle \frac{1}{J}\tr ((\frac{1}{J}\sum_{j=1}^{J}Z_{j}'Z_{j})^{-2}( 
      \frac{1}{J}\sum_{j=1}^{J}Z_{j}'\Sigma_{j}Z_{j})).
\end{align*}
Again, note that
$(\frac{1}{J}\sum_{j=1}^{J}Z_{j}'Z_{j})^{-2} \to (\mu_{Z,2})^{-2}, $ and
\begin{equation*}
\textstyle   \limsup_{J \to \infty}  \left\lVert
  \frac{1}{J}\sum_{j=1}^{J}Z_{j}'\Sigma_{j}Z_{j} \right\rVert \leq \limsup_{J \to \infty}
  \frac{1}{J}\sum_{j=1}^{J}s_{1}(Z_{j})^{2}s_{1}(\Sigma_{j}) < \infty.
\end{equation*}
This shows that
$\frac{1}{J}\tr ((\frac{1}{J}\sum_{j=1}^{J}Z_{j}'Z_{j})^{-2}( 
\frac{1}{J}\sum_{j=1}^{J}Z_{j}'\Sigma_{j}Z_{j})) \to 0$, which concludes the proof
for \eqref{eq:cov_equiv_bound}.

To show
$\lim_{J \to \infty} \E  \textstyle \sup_{\Lambda} \lVert \frac{1}{J}
\sum_{j=1}^{J} Z_{j}'(\Lambda + \Sigma_{j})^{-1} \Sigma_{j}(y_{j}-\theta_{j} )
\rVert^{2} = 0,$ I follow the lines of argument given in the proof of Theorem
\ref{thm:URE_opt_m_g} (ii) carefully. The main difference is that now the summand is
multiplied by $Z_{j}'$. Write
$H_{Z,J}(\Lambda) := \lVert \frac{1}{J} \sum_{j=1}^{J} Z_{j}'(\Lambda +
\Sigma_{j})^{-1} \Sigma_{j}(y_{j}-\theta_{j} ) \rVert$. First, I show
$EH_{Z,J}(\Lambda)^{2} \to 0,$ which implies
$H_{Z,J}(\Lambda) \overset{p}{\to} 0$. Write
$\overline{\sigma}_{Z} := \sup_{j}s_{1}(Z_{j})$, and note that
$\sup_{j}s_{1}(Z_{j}Z_{j}') = \sup_{j}s_{1}(Z_{j}'Z_{j}) =
\overline{\sigma}_{Z}^{2}$. I have
\begin{align*}
  EH_{Z,J}(\Lambda)^{2} %
                        &=  \textstyle \frac{1}{J^{2}} \E   (\sum_{j=1}^{J} (y_{j}-\theta_{j}
                          )'\Sigma_{j}(\Lambda + \Sigma_{j})^{-1}Z_{j}Z_{j}' (\Lambda + \Sigma_{j})^{-1}
                          \Sigma_{j}(y_{j}-\theta_{j} )) \\
                        & \leq \textstyle \frac{1}{J^{2}}  \sum_{j=1}^{J} s_{1}(\Sigma_{j}(\Lambda + \Sigma_{j})^{-1}Z_{j}Z_{j}' (\Lambda + \Sigma_{j})^{-1}
                          \Sigma_{j}) \tr (\Sigma_{j}) \\
                        & \leq \textstyle \frac{1}{J^{2}}  \sum_{j=1}^{J} \overline{\sigma}^{2}_{Z} \tr (\Sigma_{j}),
\end{align*}
where the first inequality follows from von Neumann's trace inequality and the
last equality from the fact that the operator norm is submultiplicative and the
bound $s_{1}(\Sigma_{j}(\Lambda + \Sigma_{j})^{-2} \Sigma_{j}) \leq 1. $
Since $\sup_{j}\tr (\Sigma_{j})\leq T \sup_{j}s_{1}(\Sigma_{j}) < \infty$,
I conclude that
$\frac{1}{J^{2}} \sum_{j=1}^{J} \overline{\sigma}^{2}_{Z}\tr (\Sigma_{j}) \to
0$. This establishes that $H_{Z,J}(\Lambda)$ converges to zero in $L^{2}$, and
thus in probability.

To show that this
convergence is uniform over $\Lambda \in \mathcal{S}_{T}^{+}$, by a similar argument as in
the proof of Theorem \ref{thm:URE}, it suffices to show that
$\tilde{H}_{Z,J} := H_{Z,J}\circ h^{-1}$ satisfies a Lipschitz condition, i.e.,
\begin{equation}\label{eq:Lip_condn_cov}
  \lvert \tilde{H}_{Z,J}(\tilde{\Lambda}) -
  \tilde{H}_{Z,J}(\tilde{\Lambda}^{\dagger}) \rvert \leq B^{Z}_{H,J} \lVert
  \tilde{\Lambda} - \tilde{\Lambda}^{\dagger} \rVert
\end{equation}
for all $\tilde{\Lambda}, \tilde{\Lambda}^{\dagger}\in
\tilde{\mathcal{L}}$, where $B^{Z}_{H,J} = O_{p}(1)$. Define
$\tilde{A}_{j} = \tilde{\Lambda}^{-1} +( \Sigma_{j} -
\underline{s}_{\Sigma}I_{T})$ and $\tilde{A}_{j}^{\dagger}$ likewise with
$\tilde{\Lambda}$ replaced with $\tilde{\Lambda}^{\dagger}$. Observe
that by the triangle inequality (and its reverse),
\begin{align}
  \begin{aligned}
     \lvert \tilde{H}_{Z,J}(\tilde{\Lambda}) -
    \tilde{H}_{Z,J}(\tilde{\Lambda}^{\dagger}) \rvert%
    \textstyle \leq   \frac{1}{J}
    \sum_{j=1}^{J}s_{1}(Z_{j})s_{1}(\tilde{A}_{j}^{-1} -
    \tilde{A}_{j}^{\dagger -1} ) \left\lVert \Sigma_{j}(y_{j}-\theta_{j} )
    \right\rVert,\label{eq:Lips_lambdatil_cov}
  \end{aligned}
\end{align}
where the first inequality follows from the reverse triangle inequality and the
second by the triangle inequality and the definition of the operator norm.

In the proof of Theorem \ref{thm:URE_opt_m_g} (ii), I showed that
\begin{equation*}
  s_{1}(\tilde{A}_{j}^{-1} - \tilde{A}^{\dagger -1}_{j}) \leq \lVert
  \tilde{\Lambda} - \tilde{\Lambda}^{\dagger} \rVert.
\end{equation*}
Plugging this bound into
\eqref{eq:Lips_lambdatil_cov}, I obtain
\begin{equation*}
  \lvert \tilde{H}_{Z,J}(\tilde{\Lambda}) -
  \tilde{H}_{Z,J}(\tilde{\Lambda}^{\dagger}) \rvert \leq \textstyle (\frac{1}{J}\sum_{j=1}^{J} s_{1}(Z_{j}) \lVert \Sigma_{j}(y_{j}-\theta_{j} ) \rVert) \lVert
  \tilde{\Lambda} - \tilde{\Lambda}^{\dagger} \rVert.
\end{equation*}
Furthermore, I have $  \sum_{j=1}^{J} s_{1}(Z_{j}) \lVert \Sigma_{j}(y_{j}-\theta_{j} ) \rVert \leq \overline{\sigma}_{Z} \sum_{j=1}^{J}  \lVert \Sigma_{j}(y_{j}-\theta_{j} ) \rVert,$
and I have already shown $\frac{1}{J} \sum_{j=1}^{J} \lVert
\Sigma_{j}(y_{j}-\theta_{j} ) \rVert = O_{p}(1)$ in the proof of Theorem
\ref{thm:URE_opt_m_g} (ii). This establishes \eqref{eq:Lip_condn_cov}, and thus
$\sup_{\Lambda}H_{Z,J}(\Lambda) \overset{p}{\to} 0$.

Now, to show that $\sup_{\Lambda}H_{Z,J}(\Lambda)$ converges to zero in $L^{2}$,
it is enough to show that $\left\{ \sup_{\Lambda}H_{Z,J}(\Lambda)^{2}  \right\}_{J\leq
  1}$ is uniformly integrable. A sufficient condition for this is
\begin{equation*}
  \textstyle \sup_{J} \E \sup_{\Lambda}H_{Z,J}(\Lambda)^{2+\delta} < \infty,
\end{equation*}
for some $\delta > 0$. An upper bound of $H_{Z,J}(\Lambda)$ is given by
\begin{align*}
  H_{Z,J}(\Lambda) %
  \leq  \textstyle  \frac{1}{J} \sum_{j=1}^{J}s_{1}(Z_{j}) s_{1}((\Lambda + \Sigma_{j})^{-1}
         \Sigma_{j}))\lVert y_{j}-\theta_{j} \rVert 
  \leq  \textstyle \overline{\sigma}_{Z} \frac{1}{J} \sum_{j=1}^{J}\lVert y_{j}-\theta_{j} \rVert,
\end{align*}
where the first inequality follows from the triangle inequality and the
definition of the operator norm, and the second inequality follows because
$ s_{1}((\Lambda + \Sigma_{j})^{-1} \Sigma_{j}) \leq 1.$ Therefore, following \eqref{eq:unif_int_center}, I have
\begin{align*}
\textstyle    \sup_{\Lambda} H_{Z,J}(\Lambda)^{2+\delta} \leq \overline{\sigma}_{Z}^{2+\delta}\frac{1}{J} \sum_{j=1}^{J}2^{1+\delta}(\lVert y_{j}\rVert^{2+\delta} + \lVert\theta_{j}
  \rVert^{2+\delta}),
\end{align*}
Taking expectations, I obtain
\begin{equation*}
  \textstyle \limsup_{J}  \E  \sup_{\Lambda} H_{J}(\Lambda)^{2+\delta} < \infty,
\end{equation*}
for any $\delta \in [0, 2]$, and thus $\sup_{J}  \E  \sup_{\Lambda} H_{Z,J}(\Lambda)^{2+\delta} < \infty$. This
concludes the proof for $\sup_{\Lambda} H_{Z,J}(\Lambda) \overset{L^{2}}{\to} 0$.

\subsection{Proof of Theorem \ref{thm:forecasting}}
\label{appsec:proof-theorem-forecasting}

I first give details on the derivation of the UPE. Note that
\begin{align*}
  & \E (B(\Lambda, \Sigma_{j,-T})'y_{j,-T} - \theta_{jT})^{2} \\
  = & \E (B(\Lambda, \Sigma_{j,-T})'y_{j,-T} -y_{jT})^{2} + \E (y_{jT} -
      \theta_{jT})^{2}  - 2 \E [(y_{jT} -B(\Lambda, \Sigma_{j,-T})'y_{j,-T}  )(y_{jT} - \theta_{jT})].
\end{align*}
The cross term can be simplified to $\Sigma_{j,T} - B(\Lambda, \Sigma_{j,-T})'\Sigma_{j,T,-T}$.
Hence, it follows that
\begin{align*}
   \E (B(\Lambda, \Sigma_{j,-T})'y_{j,-T} - \theta_{jT})^{2} = \E (B(\Lambda, \Sigma_{j,-T})'y_{j,-T}
      -y_{jT})^{2}  - \Sigma_{jT} + 2B(\Lambda, \Sigma_{j,-T})'\Sigma_{j,T,-T},
\end{align*}
which shows the UPE is indeed unbiased.

Now, I prove Theorem \ref{thm:forecasting}. Consider the following bound,
\begin{align}\label{eq:decomp_ure_loss_fc}
  \begin{aligned} 
    & \textstyle \left\lvert \frac{1}{J} \sum_{j=1}^{J}\left( (B(\Lambda,
        \Sigma_{j,-T})'y_{j,-T} -y_{jT})^{2} -\Sigma_{j,T} \right) - \frac{1}{J}
      \sum_{j=1}^{J}(B(\Lambda,
      \Sigma_{j,-1})'y_{j,-1} - \theta_{j,T+1})^{2} \right\rvert\\
    \leq & \textstyle \left\lvert \frac{1}{J} \sum_{j=1}^{J}\left( (B(\Lambda,
        \Sigma_{j,-T})'y_{j,-T} -y_{jT})^{2} -\Sigma_{j,T} \right) - \frac{1}{J}
      \sum_{j=1}^{J}(B(\Lambda,
      \Sigma_{j,-T})'y_{j,-T} - \theta_{j,T})^{2} \right\rvert \\
    &  + \textstyle \left\lvert \frac{1}{J} \sum_{j=1}^{J}(B(\Lambda,
      \Sigma_{j,-T})'y_{j,-T} - \theta_{j,T})^{2} - \frac{1}{J}
      \sum_{j=1}^{J}(B(\Lambda,\Sigma_{j,-1})'y_{j,-1} - \theta_{j,T+1})^{2}
    \right\rvert,
  \end{aligned}
\end{align}
which is by the triangle inequality.

Further algebra gives
\begin{align*}
  &  \textstyle \frac{1}{J} \sum_{j=1}^{J}(B(\Lambda,\Sigma_{j,-T})'y_{j,-T} -
    \theta_{j,T})^{2} \\
  = &  \textstyle \frac{1}{J} \sum_{j=1}^{J} ((B(\Lambda,\Sigma_{j,-T})'y_{j,-T} -
      y_{j,T})^{2} + (y_{jT} - \theta_{j,T})^{2})  \\
  &  \textstyle - 2 \frac{1}{J} \sum_{j=1}^{J} (y_{j,T} - B(\Lambda,\Sigma_{j,-T})'y_{j,-T})(y_{jT} - \theta_{j,T}).
\end{align*}
The cross term can be decomposed as
\begin{align*}
  & \textstyle \frac{1}{J} \sum_{j=1}^{J} (y_{j,T} - B(\Lambda,\Sigma_{j,-T})'y_{j,-T})(y_{jT} - \theta_{j,T}) \\
  = & \textstyle \frac{1}{J} \sum_{j=1}^{J} (y_{jT} -\theta_{jT})^{2} - \frac{1}{J} \sum_{j=1}^{J}B(\Lambda,\Sigma_{j,-T})'(y_{j,-T} - \theta_{j,-T})(y_{jT} -
      \theta_{jT}) \\
  & + \textstyle \frac{1}{J} \sum_{j=1}^{J}  (\theta_{jT} -   B(\Lambda,\Sigma_{j,-T})'\theta_{j,-T})(y_{jT} -
    \theta_{jT}).      
\end{align*}
Plugging this into the first term of right-hand side in
\eqref{eq:decomp_ure_loss_fc}, it follows that
\begin{align*}
  & \textstyle \left\lvert \frac{1}{J} \sum_{j=1}^{J}\left( (B(\Lambda,
    \Sigma_{j,-T})'y_{j,-T} -y_{jT})^{2} -\Sigma_{j,T} \right) - \frac{1}{J}
    \sum_{j=1}^{J}(B(\Lambda,
    \Sigma_{j,-T})'y_{j,-T} - \theta_{j,T})^{2} \right\rvert \\
  \leq & \textstyle \left\lvert \frac{1}{J} \sum_{j=1}^{J}\left((y_{jT}
         -\theta_{jT})^{2}  -\Sigma_{j,T} \right) \right\rvert %
    \textstyle + \left\lvert \frac{2}{J} \sum_{j=1}^{J}B(\Lambda,\Sigma_{j,-T})'((y_{j,-T} - \theta_{j,-T})(y_{jT} -
    \theta_{jT}) -\Sigma_{j,T,-T}) \right\rvert \\
  &\textstyle + \left\lvert  \frac{2}{J} \sum_{j=1}^{J}  \theta_{jT}(y_{jT} -
    \theta_{jT}) \right\rvert + \left\lvert  \frac{2}{J} \sum_{j=1}^{J}   B(\Lambda,\Sigma_{j,-T})'\theta_{j,-T}(y_{jT} -
    \theta_{jT}) \right\rvert \\
  := & \mathrm{(I)}_{J} + \mathrm{(II)}_{J} + \mathrm{(III)}_{J} + \mathrm{(IV)}_{J}.
\end{align*}

The aim is to show that each of the four terms in the last line converges to $0$
in $L^{1}$, uniformly over $\Lambda \in \mathcal{L}$. In fact, I show uniformity
over
$ (\Lambda_{T,-T}, \Lambda_{-T}) \in
\overline{\mathcal{L}}:=\overline{\mathcal{L}}_{T,-T} \times
\overline{\mathcal{L}}_{-T},$ where
\begin{align*}
  &\overline{\mathcal{L}}_{T,-T} = \{\Lambda_{T,-T} \in \mathbf{R}^{T-1}: \lVert
    \Lambda_{T,-T} \rVert \leq K_{T,-T}\}, \text{ and} \\
  &  \overline{\mathcal{L}}_{-T} = \{\Lambda_{-T} \in S_{T-1}^{+}: \lVert
    \Lambda_{-T} \rVert \leq K_{-T}\}.
\end{align*}
Here, $K_{T,-T}$ and $K_{-T}$ are positive numbers large enough so that
$\left\{ \Lambda_{T,-T}: \Lambda \in \mathcal{L} \right\} \subset
\overline{\mathcal{L}}_{T,-T}$ and
$\left\{ \Lambda_{-T}: \Lambda \in \mathcal{L} \right\} \subset
\overline{\mathcal{L}}_{-T}$, which exist due to the fact that $\mathcal{L}$ is
bounded. Note that $\Lambda \in \mathcal{L}$ implies
$(\Lambda_{T,-T}, \Lambda_{T}) \in \overline{\mathcal{L}}$, and thus
establishing convergence uniformly over the latter is sufficient. Note that
\begin{equation}
    \sup_{(\Lambda_{T,-T},\Lambda_{T}) \in \overline{\mathcal{L}}} %
    \leq  K_{T,-T} s_{T}^{-1}(\Sigma_{j}),\label{eq:bound_BLam}
\end{equation}
where the inequality follows because the relationship between eigenvalues of a
matrix and the eigenvalues of its principal submatrices (see, for example,
Theorem 4.3.15 of \cite{horn1990MatrixAnalysis}). In some of the derivations
later on, it is useful to make clear that $B(\Lambda,\Sigma_{j,-T})$ depends on
$\Lambda$ only through $(\Lambda_{T,-T}, \Lambda_{-T})$. When this fact has be
highlighted, I write
$B(\Lambda_{T,-T}, \Lambda_{-T},\Sigma_{j,-T}) := B(\Lambda,\Sigma_{j,-T}).$
Now, I condition all random quantities on a sequence
$\{((\theta_{j}',\theta_{j,T+1})', \Sigma_{j})\}_{j=1}^{\infty}$. Note that by
Assumption \ref{assum:bounded_rand}, now I can assume that Assumption
\ref{assum:bounded} holds.

To show that $\mathrm{(I)}_{J}$ converges to zero in $L^{2}$, and thus in
$L^{1}$, note that
\begin{equation*}
\textstyle   \E   \left\lvert \frac{1}{J} \sum_{j=1}^{J}\left((y_{jT}
    -\theta_{jT})^{2}  -\Sigma_{j,T} \right) \right\rvert^{2} %
  \leq   \frac{1}{J^{2}} \sum_{j=1}^{J}8(\E y^{4}_{jT}
         +\theta_{jT}^{4}).
\end{equation*}
The summand in the last line is uniformly bounded over $j,$ which establishes
the convergence.

Similarly, $\mathrm{(III)}_{J}\overset{L^{2}}{\to} 0$ can be easily shown by
noting that
\begin{equation*}
\textstyle  \E \left\lvert  \frac{2}{J} \sum_{j=1}^{J}  \theta_{jT}(y_{jT} -
    \theta_{jT}) \right\rvert^{2}  \leq  \frac{4}{J^{2}} \sum_{j=1}^{J}\theta_{jT}^{2}\Sigma_{jT},
\end{equation*}
and the summand of the right-hand side is bounded uniformly over $j.$

To show that
$\sup_{\overline{\mathcal{L}}}\mathrm{(II)}_{J} \overset{L^{1}}{\to} 0$ and
$\sup_{\overline{\mathcal{L}}}\mathrm{(IV)}_{J} \overset{L^{1}}{\to} 0 $, I again
use a result by \cite{andrews1992GenericUniformConvergence}, which will
establish convergence in probability, and then show a uniform integrability
condition to show that convergence holds in $L^{1}$ as well. Here, I write
$\sup_{\overline{\mathcal{L}}}$ as a shorthand for
$\sup_{(\Lambda_{T,-T}, \Lambda_{-T}) \in \overline{\mathcal{L}}}$. I start with
$ \mathrm{(II)}_{J}$. For pointwise convergence (in $L^{2}$), note that
\begin{align*}
  & \textstyle   \E  \left\lvert \frac{2}{J} \sum_{j=1}^{J}B(\Lambda,\Sigma_{j,-T})'((y_{j,-T} - \theta_{j,-T})(y_{jT} -
    \theta_{jT}) -\Sigma_{j,T,-T}) \right\rvert^{2} \\
  \leq &   \textstyle  \frac{4}{J^{2}} \sum_{j=1}^{J}K_{T,-T}^{2} s_{T}^{-2}(\Sigma_{j})\tr (\var ((y_{j,-T} -
         \theta_{j,-T})(y_{jT} -\theta_{jT}))),
\end{align*}
where the inequality follows from von Neumann's trace inequality and the
fact that $s_{1}(xx')= s_{1}(x'x) = \lVert x \rVert^{2}$ for any
$x \in \mathrm{R}^{T-1}$, and the last inequality from
\eqref{eq:bound_BLam}. Moreover, by Cauchy-Schwarz
\begin{equation*}
   \tr (\var ((y_{j,-T} -
    \theta_{j,-T})(y_{jT} -\theta_{jT}))) 
  \leq   \textstyle
          \sum_{t=1}^{T-1}(\E (y_{jt}-\theta_{jt})^{4}(y_{jT}-\theta_{jT})^{4})^{1/2},
\end{equation*}
The term on the right hand side is bounded uniformly over $j$, which establishes
$\mathrm{(II)_{J}}\overset{L^{2}}{\to} 0.$ It remains to establish a Lipschitz
condition. Define
\begin{equation*}
    G_{J}(\Lambda_{T,-T}, \Lambda_{-T}) = \textstyle \frac{2}{J}  \sum_{j=1}^{J}B(\Lambda,\Sigma_{j,-T})'((y_{j,-T} - \theta_{j,-T})(y_{jT} -
  \theta_{jT}) -\Sigma_{j,T,-T}).
\end{equation*}
I show that $G_{J}(\Lambda_{T,-T}, \Lambda_{-T})$ is Lipschitz in
$\Lambda_{T,-T}$ and $\Lambda_{-T}$, respectively, with Lipschitz constants
bounded in probability that do not depend on the other parameter held fixed,
which will establish that $G_{J}(\Lambda_{T,-T}, \Lambda_{-T})$ is Lipschitz
with respect to $(\Lambda_{T,-T}, \Lambda_{-T})$.  Note that, for any
$\Lambda_{T,-T}, \tilde{\Lambda}_{T,-T} \in \overline{\mathcal{L}}_{T,-T}$,
\begin{align}
  \begin{aligned}\label{eq:B_Lip_T-T}
    \lVert B(\Lambda_{T,-T},\Lambda_{-T},\Sigma_{j,-T}) -
    B(\tilde{\Lambda}_{T,-T},\Lambda_{-T},\Sigma_{j,-T}) \rVert %
    \leq  \underline{s}_{\Sigma}^{-1} \lVert \Lambda_{T,-T} -
    \tilde{\Lambda}_{T,-T} \rVert.
  \end{aligned}
\end{align}
Also, for any
$\Lambda_{-T}, \tilde{\Lambda}_{-T} \in \overline{\mathcal{L}}_{-T}$, I
have
\begin{align*}
  \begin{aligned}
     & \lVert B(\Lambda_{T,-T},\Lambda_{-T},\Sigma_{j,-T}) -
    B({\Lambda}_{T,-T},\tilde{\Lambda}_{-T},\Sigma_{j,-T}) \rVert  \\
    \leq& K_{T,-T} s_{1}((\Sigma_{j,-T} + \Lambda_{-T})^{-1} - (\Sigma_{j,-T} +
    \tilde{\Lambda}_{-T})^{-1})
  \end{aligned}
\end{align*}
To derive a bound for
$s_{1}((\Sigma_{j,-T} + \Lambda_{-T})^{-1} - (\Sigma_{j,-T} +
\tilde{\Lambda}_{-T})^{-1})$, note that
\begin{align*}
  (\Sigma_{j,-T} + \Lambda_{-T})^{-1} - (\Sigma_{j,-T} +
    \tilde{\Lambda}_{-T})^{-1} 
  \leq (\Sigma_{j,-T} +
      \tilde{\Lambda}_{-T})^{-1} (\tilde{\Lambda}_{-T}- \Lambda_{-T})(\Sigma_{j,-T} +
      \Lambda_{-T})^{-1}.
\end{align*}
This implies
$ s_{1}((\Sigma_{j,-T} + \Lambda_{-T})^{-1} - (\Sigma_{j,-T} +
\tilde{\Lambda}_{-T})^{-1}) \leq \underline{s}_{\Sigma}^{-2} \lVert
\Lambda_{-T} - \tilde{\Lambda}_{-T} \rVert,$ which in turn implies the
following Lipschitz condition,
\begin{equation} \label{eq:B_Lip_-T} \lVert B(\Lambda_{T,-T},\Lambda_{-T},
  \Sigma_{j,-T}) - B({\Lambda}_{T,-T},\tilde{\Lambda}_{-T},\Sigma_{j,-T}) \rVert \leq
  \underline{s}_{\Sigma}^{-2} \lVert \Lambda_{-T} - \tilde{\Lambda}_{-T}
  \rVert.
\end{equation}
Now, combining \eqref{eq:B_Lip_T-T} and \eqref{eq:B_Lip_-T}, we have for any
$(\Lambda_{T,-T}, \Lambda_{-T}), (\tilde{\Lambda}_{T,-T},
\tilde{\Lambda}_{-T}) \in \overline{\mathcal{L}}$,
\begin{align}
  \label{eq:B_Lip}
  \begin{aligned}
    & \lVert B(\Lambda_{T,-T},\Lambda_{-T},\Sigma_{j,-T}) -
    B(\tilde{\Lambda}_{T,-T},\tilde{\Lambda}_{-T},\Sigma_{j,-T}) \rVert  \\
    \leq & (\underline{s}_{\Sigma}^{-1} \vee \underline{s}_{\Sigma}^{-2})
    (\lVert \Lambda_{T,-T} - \tilde{\Lambda}_{T,-T} \rVert + \lVert
    \Lambda_{-T} - \tilde{\Lambda}_{-T} \rVert).
  \end{aligned}
\end{align}
Because
$\lVert (\Lambda_{T,-T}, \Lambda_{-T}) \rVert := \lVert \Lambda_{T,-T} \rVert +
\lVert \Lambda_{-T} \rVert $ defines a norm on the product space
$\overline{\mathcal{L}}$, this shows that $B(\cdot,\cdot,\Sigma_{j,-T})$ is Lipshitz
on $\overline{\mathcal{L}}$.

It follows that
\begin{align*}
  & \lvert G_{J}(\Lambda_{T,-T}, \Lambda_{-T}) -
    G_{J}(\tilde{\Lambda}_{T,-T}, \tilde{\Lambda}_{-T}) \rvert \\
  \leq &\textstyle  ((\underline{s}_{\Sigma}^{-1} \vee \underline{s}_{\Sigma}^{-2})\frac{2}{J} \sum_{j=1}^{J} \lVert(y_{j,-T} -
         \theta_{j,-T})(y_{jT} - \theta_{jT}) -\Sigma_{j,T,-T} \rVert) \lVert
         (\Lambda_{T,-T}, \Lambda_{-T}) - (\tilde{\Lambda}_{T,-T}, \tilde{\Lambda}_{-T})\rVert.
\end{align*}
Hence, now it suffices to show
\begin{equation}\label{eq:B_Lip_con_bounded}
   \textstyle  \frac{1}{J} \sum_{j=1}^{J} \lVert(y_{j,-T} -
  \theta_{j,-T})(y_{jT} - \theta_{jT}) -\Sigma_{j,T,-T} \rVert = O_{p}(1).
\end{equation}
A bound for the left-hand side is given by
\begin{align*}
    & \textstyle \frac{1}{J} \sum_{j=1}^{J} (\lVert(y_{j,-T} -
      \theta_{j,-T})(y_{jT} - \theta_{jT}) \rVert - \E \lVert(y_{j,-T} -
      \theta_{j,-T})(y_{jT} - \theta_{jT}) \rVert) \\
   & \textstyle + \frac{1}{J}\sum_{j=1}^{J}( \E  \lVert(y_{j,-T} -
    \theta_{j,-T})(y_{jT} - \theta_{jT}) \rVert +  \lVert\Sigma_{j,T,-T} \rVert) 
  =: \mathrm{(A)}_{J} + \mathrm{(B)}_{J}.
\end{align*}
I show that $\mathrm{(A)}_{J} = o_{p}(1)$ and $\mathrm{(B)}_{J}= O(1)$, from
which \eqref{eq:B_Lip_con_bounded} will follow. 

To show $\mathrm{(A)}_{J} \overset{p}{\to} 0$, it suffices to show that the
variance of the summand is bounded over $j$, since then it converges to
zero in $L^{2}$. Observe that
\begin{align}
  \begin{aligned}
    & \var (\lVert(y_{j,-T} -  \theta_{j,-T})(y_{jT} - \theta_{jT}) \rVert) \\
    \leq &  4 \textstyle \sum_{t=1}^{T-1} ((\E \lvert y_{jT} \rvert^{4}\E \lvert
    y_{j,t} \rvert^{4})^{1/2} + \lvert \theta_{jT} \rvert^{2}\E \lvert y_{j,t}
    \rvert^{2} + \lvert \theta_{j,t} \rvert^{2}\E \lvert y_{jT} \rvert^{2} +
    \lvert \theta_{jT} \rvert^{2}\lvert \theta_{j,t} \rvert^{2}),\label{eq:11}
  \end{aligned}
\end{align}
where the inequality follows by Cauchy-Schwarz. The expression in the last line
is bounded uniformly over $j$, and thus the variance term is as well. By
Jensen's inequality, this also establishes
$ \limsup_{J \to \infty}\mathrm{(B)}_{J} < \infty$. This concludes the proof for
$ \sup_{\overline{\mathcal{L}}}\lvert G_{J}(\Lambda_{T,-T}, \Lambda_{-T}) \rvert
\overset{p}{\to} 0.$

Now, I show that the convergence is in fact in $L^{1}$ by establishing uniform
integrability of
$\sup_{\overline{\mathcal{L}}}\lvert G_{J}(\Lambda_{T,-T}, \Lambda_{-T})
\rvert$. To this end, I verify a sufficient condition,
\begin{equation}\label{eq:suff_condn_UI_B_GJ}
\textstyle   \sup_{j}
  \E  \left( \sup_{\overline{\mathcal{L}}}\lvert G_{J}(\Lambda_{T,-T},
    \Lambda_{-T}) \rvert  \right)^{2}
  < \infty.
\end{equation}
First, I bound $\lvert G_{J}(\Lambda_{T,-T}, \Lambda_{-T}) \rvert$. By the
triangle inequality, Cauchy-Schwarz and \eqref{eq:bound_BLam}, we have
\begin{align*}
   \lvert G_{J}(\Lambda_{T,-T}, \Lambda_{-T}) \rvert%
  \leq  \textstyle K_{T,-T} \underline{s}^{-1}_{M} \frac{2}{J} \sum_{j=1}^{J}\lVert (y_{j,-T} - \theta_{j,-T})(y_{jT} -
         \theta_{jT}) -\Sigma_{j,T,-T}) \rVert,
\end{align*}
Hence, it follows that
\begin{equation*}
   \textstyle \E  \sup_{\overline{\mathcal{L}}}\lvert G_{J}(\Lambda_{T,-T}, \Lambda_{-T})
    \rvert^{2} 
  \leq \textstyle K_{T,-T}^{2} \underline{s}^{-2}_{M}  \frac{8}{J} \sum_{j=1}^{J}(\E \lVert (y_{j,-T} - \theta_{j,-T})(y_{jT} -
         \theta_{jT})\rVert^{2} + \lVert\Sigma_{j,T,-T} \rVert^{2}),
\end{equation*}
Since I have shown that the summand in the last line is bounded over $j$ in
\eqref{eq:11}, we have \eqref{eq:suff_condn_UI_B_GJ}. We conclude that
$ \sup_{\overline{\mathcal{L}}}\lvert G_{J}(\Lambda_{T,-T}, \Lambda_{-T}) \rvert
\overset{L^{1}}{\to} 0.$

I follow these same steps for $\mathrm{(IV)}_{J}$. First, define
\begin{equation*}
  H_{J}(\Lambda_{T,-T},\Lambda_{-T},\Sigma_{j,-T}) := \textstyle \frac{2}{J} \sum_{j=1}^{J}   B(\Lambda,\Sigma_{j,-T})'\theta_{j,-T}(y_{jT} -
  \theta_{jT}).
\end{equation*}
For pointwise convergence, note
that by Cauchy-Schwarz and \eqref{eq:bound_BLam}, we have
\begin{align*}
  \var (B(\Lambda,\Sigma_{j,-T})'\theta_{j,-T}y_{jT}) %
  \leq & K_{T,-T}^{2} \underline{s}^{-2}_{M} \lVert \theta_{j,-T} \rVert^{2}
         \Sigma_{jT},
\end{align*}
 The right hand side is bounded over $j$,
and thus $H_{J}(\Lambda_{T,-T},\Lambda_{-T},\Sigma_{j,-T}) \overset{L^{2}}{\to} 0.$

Now, I show that $H_{J}(\Lambda_{T,-T},\Lambda_{-T},\Sigma_{j,-T})$
satisfies a Lipschitz condition. I have
\begin{align*}
  & \lvert H_{J}(\Lambda_{T,-T}, \Lambda_{-T}) -
    H_{J}(\tilde{\Lambda}_{T,-T}, \tilde{\Lambda}_{-T}) \rvert \\
  \leq & \textstyle \frac{2}{J} \sum_{j=1}^{J}\lVert B(\Lambda,\Sigma_{j,-T}) -
         B(\tilde{\Lambda},\Sigma_{j,-T}) \rVert \lVert \theta_{j,-T} \rVert \lvert y_{jT} - \theta_{jT} \rvert \\
  \leq & \textstyle ((\underline{s}_{\Sigma}^{-1} \vee \underline{s}_{\Sigma}^{-2})\frac{2}{J} \sum_{j=1}^{J} \lVert \theta_{j,-T} \rVert \lvert y_{jT} - \theta_{jT} \rvert ,
\end{align*}
where the second inequality is by Cauchy-Schwarz and the third inequality
follows from \eqref{eq:B_Lip}. The fact that
$\frac{2}{J} \sum_{j=1}^{J} \lVert \theta_{j,-T} \rVert \lvert y_{jT} -
\theta_{jT} \rvert = O_{p}(1)$ follows from similar, but simpler, steps we have
taken to show \eqref{eq:B_Lip_con_bounded}. This implies
$\sup_{\mathcal{L}}\lvert H_{J}(\Lambda_{T,-T}, \Lambda_{-T}) \rvert
\overset{p}{\to} 0$. Again, following the same arguments we have used to show
\eqref{eq:suff_condn_UI_B_GJ}, we can easily show that
$\sup_{\mathcal{L}}\lvert H_{J}(\Lambda_{T,-T}, \Lambda_{-T}) \rvert$ is
uniformly integrable, from which it follows that
$\sup_{\mathcal{L}}\lvert H_{J}(\Lambda_{T,-T}, \Lambda_{-T}) \rvert
\overset{L^{1}}{\to} 0$. This concludes the proof for the first term of the
right-hand side of \eqref{eq:decomp_ure_loss_fc} converging to zero in $L^{1}$.

For the second term of the
right-hand side of \eqref{eq:decomp_ure_loss_fc}, note that
\begin{align*}
  &  (B(\Lambda,\Sigma_{j,-T})'y_{j,-T} - \theta_{j,T})^{2}  \\
  = & (B(\Lambda,\Sigma_{j,-T})'y_{j,-T} - B(\Lambda,\Sigma_{j,-T})'\theta_{j,-T})^{2} 
      +(B(\Lambda,\Sigma_{j,-T})'\theta_{j,-T}  - \theta_{j,T})^{2} \\
  &+ 2 (B(\Lambda,\Sigma_{j,-T})'y_{j,-T} - B(\Lambda,\Sigma_{j,-T})'\theta_{j,-T})(B(\Lambda,\Sigma_{j,-T})'\theta_{j,-T}  - \theta_{j,T}).
\end{align*}
Furthermore, I have
\begin{align*}
   \E  (B(\Lambda,\Sigma_{j,-T})'y_{j,-T} - B(\Lambda,\Sigma_{j,-T})'\theta_{j,-T})^{2} %
  = B(\Lambda,\Sigma_{j,-T})'\Sigma_{j,-T}B(\Lambda,\Sigma_{j,-T}).
\end{align*}
Hence, it follows that
\begin{align*}
  &\textstyle  \left\lvert \frac{1}{J} \sum_{j=1}^{J}(B(\Lambda,
    \Sigma_{j,-T})'y_{j,-T} - \theta_{j,T})^{2} - \frac{1}{J}
    \sum_{j=1}^{J}(B(\Lambda,\Sigma_{j,-1})'y_{j,-1} - \theta_{j,T+1})^{2}
    \right\rvert \\
  \leq & \textstyle \left\lvert \frac{1}{J}\sum_{j=1}^{J}\left(  (B(\Lambda,\Sigma_{j,-T})'(y_{j,-T} -
         \theta_{j,-T}))^{2} - B(\Lambda,\Sigma_{j,-T})'\Sigma_{j,-T}B(\Lambda,\Sigma_{j,-T}) \right) \right\rvert \\
  & \textstyle+ 2 \left\lvert \frac{1}{J} \sum_{j=1}^{J} 
    (B(\Lambda,\Sigma_{j,-T})'(y_{j,-T} - \theta_{j,-T}))(B(\Lambda,\Sigma_{j,-T})'\theta_{j,-T}  - \theta_{j,T})
    \right\rvert \\ 
  &  \textstyle + \left\lvert \frac{1}{J}\sum_{j=1}^{J} \left( (B(\Lambda,\Sigma_{j,-1})'(y_{j,-1} -
    \theta_{j,-1}))^{2} - B(\Lambda,\Sigma_{j,-1})'\Sigma_{j,-1}B(\Lambda,\Sigma_{j,-1}) \right) \right\rvert \\
  &\textstyle +2 \left\lvert \frac{1}{J} \sum_{j=1}^{J} 
    (B(\Lambda,\Sigma_{j,-1})'(y_{j,-1} - \theta_{j,-1}))(B(\Lambda,\Sigma_{j,-1})'\theta_{j,-1}  -
    \theta_{j,-1})  \right\rvert \\
  & \textstyle+ \left\lvert \frac{1}{J}\sum_{j=1}^{J}\left( B(\Lambda,\Sigma_{j,-T})'\Sigma_{j,-T}B(\Lambda,
    \Sigma_{j,-T}) - B(\Lambda,\Sigma_{j,-1})'\Sigma_{j,-1}B(\Lambda,\Sigma_{j,-1}) \right) \right\rvert\\
  & \textstyle + \left\lvert \frac{1}{J}\sum_{j=1}^{J}\left( (B(\Lambda,
    \Sigma_{j,-T})'\theta_{j,-T}  - \theta_{j,T})^{2} - (B(\Lambda,
    \Sigma_{j,-1})'\theta_{j,-1}  - \theta_{j,T+1})^{2} \right)  \right\rvert \\
  =& \mathrm{(I)}_{J} + \mathrm{(II)}_{J} +\mathrm{(III)}_{J} + \mathrm{(IV)}_{J} +
     \mathrm{(V)}_{J} + \mathrm{(VI)}_{J}.
\end{align*}

I show that each of the six terms converges to zero uniformly over $\mathcal{L}$
in the $L^{1}$ sense. The proof for the first four terms are extremely
similar. Hence, I provide a proof for only $\mathrm{(I)_{J}}$, and a sketch for
the other three terms. Note that the terms $\mathrm{(V)}_{J}$ and
$\mathrm{(VI)}_{J}$ are nonrandom. %

Note that the summand in $\mathrm{(I)}_{J}$ can be written as
\begin{equation*}
  \tr  (B(\Lambda,\Sigma_{j,-T})B(\Lambda,\Sigma_{j,-T})'((y_{j,-T} -
  \theta_{j,-T})(y_{j,-T} -
  \theta_{j,-T})' -\Sigma_{j,-T} )),
\end{equation*}
which has mean zero. Hence, if the expectation of the square of this term is
bounded over $j$, then $\mathrm{(I)}_{J} \overset{L^{2}}{\to} 0$. We have
\begin{align*}
  &\lvert  \tr  (B(\Lambda,\Sigma_{j,-T})B(\Lambda,\Sigma_{j,-T})'((y_{j,-T} -
    \theta_{j,-T})(y_{j,-T} -
    \theta_{j,-T})' -\Sigma_{j,-T} )) \rvert \\
  \leq & \lvert  \tr  (B(\Lambda,\Sigma_{j,-T})B(\Lambda,\Sigma_{j,-T})'(y_{j,-T} -
         \theta_{j,-T})(y_{j,-T} -
         \theta_{j,-T})') \rvert \\
  & + \lvert \tr  (B(\Lambda,\Sigma_{j,-T})B(\Lambda,\Sigma_{j,-T})'\Sigma_{j,-T} )
    \rvert \\
  \leq & \lVert B(\Lambda,\Sigma_{j,-T}) \rVert^{2}( \lVert y_{j,-T} - \theta_{j,-T}
         \rVert^{2} + \tr (\Sigma_{j,-T})),
\end{align*}
where the last inequality follows from von Neumann's trace inequality and the
equivalence between the largest singular value of the outer product of a vector
and its squared $L^{2}$ norm. It follows that
\begin{align*}
  & \E  \, \tr  (B(\Lambda,\Sigma_{j,-T})B(\Lambda,\Sigma_{j,-T})'((y_{j,-T} -
    \theta_{j,-T})(y_{j,-T} -
    \theta_{j,-T})' -\Sigma_{j,-T} ))^{2} \\
  \leq & \E  \, \lVert B(\Lambda,\Sigma_{j,-T}) \rVert^{4}( \lVert y_{j,-T} - \theta_{j,-T}
         \rVert^{2} + \tr (\Sigma_{j,-T}))^{2} \\
  \leq & \E  \, K_{T,-T}^{4}\underline{s}_{\Sigma}^{-4}( 8\lVert y_{j,-T}\rVert^{4} + 8\lVert\theta_{j,-T}
         \rVert^{4} + \tr (\Sigma_{j,-T})^{2} + 4(\lVert y_{j,-T}\rVert^{2} + \lVert\theta_{j,-T} \rVert^{2})\tr (\Sigma_{j,-T})),
\end{align*}
where the term in the last line is bounded over $j$. This shows that $\mathrm{(I)}_{J} \overset{L^{2}}{\to}
0$.

Now, to obtain a uniform convergence result, write
\begin{align*}
  & G_{I,J}(\Lambda_{T,-T}, \Lambda_{-T}) \\
  = & \textstyle \frac{1}{J} \sum_{j=1}^{J} \tr  (B(\Lambda,\Sigma_{j,-T})B(\Lambda,\Sigma_{j,-T})'((y_{j,-T} -
      \theta_{j,-T})(y_{j,-T} -
      \theta_{j,-T})' -\Sigma_{j,-T} )).
\end{align*}

For any two $x,\tilde{x} \in \mathrm{R}^{T-1}$, we have
\begin{equation*}
  \lVert xx' - \tilde{x}\tilde{x}' \rVert \leq \lVert x - \tilde{x}
  \rVert ( \lVert x \rVert +
  \lVert \tilde{x} \rVert),
\end{equation*}
where the inequality holds by adding and subtracting $x\tilde{x}'$, applying
the triangle inequality, and then Cauchy-Schwarz. This, combined with
\eqref{eq:bound_BLam} and \eqref{eq:B_Lip}, gives
\begin{align}
  \begin{aligned}\label{eq:B2_Lip}
    & \lVert  B(\Lambda,\Sigma_{j,-T})B(\Lambda,\Sigma_{j,-T})' - B(\tilde{\Lambda},
    \Sigma_{j,-T})B(\tilde{\Lambda},\Sigma_{j,-T})' \rVert \\
    \leq & 2 \underline{s}_{\Sigma}^{-1}K_{T,-T} \lVert B(\Lambda,\Sigma_{j,-T}) -
    B(\tilde{\Lambda},\Sigma_{j,-T}) \rVert \\
    \leq & 2 \underline{s}_{\Sigma}^{-1}(\underline{s}_{\Sigma}^{-1} \vee
    \underline{s}_{\Sigma}^{-2}) K_{T,-T}
    \lVert (\Lambda_{T,-T}, \Lambda_{-T}) -  (\tilde{\Lambda}_{T,-T}, \tilde{\Lambda}_{-T}) \rVert,
  \end{aligned}
\end{align}
which shows that $B(\Lambda,\Sigma_{j,-T})B(\Lambda,\Sigma_{j,-T})$ is Lipschitz. This
will translate into a Lipschitz condition on $G_{I,J}(\Lambda_{T,-T},
\Lambda_{-T}).$ For simplicity, I write
\begin{equation*}
  B^{2}(\Lambda,\Sigma_{j,-T}) = B(\Lambda,\Sigma_{j,-T})B(\Lambda,\Sigma_{j,-T})'.
\end{equation*}
Observe that
\begin{align*}
  & \lvert G_{I,J}(\Lambda_{T,-T}, \Lambda_{-T}) -
    G_{I,J}(\tilde{\Lambda}_{T,-T}, \tilde{\Lambda}_{-T}) \rvert \\
  \leq & \textstyle \left \lvert \frac{1}{J} \sum_{j=1}^{J} \tr  ((B^{2}(\Lambda,\Sigma_{j,-T})-B^{2}(\tilde{\Lambda},\Sigma_{j,-T}))'(y_{j,-T} -
         \theta_{j,-T})(y_{j,-T} -
         \theta_{j,-T})' ) \right\rvert \\
  &  + \textstyle\left\lvert \frac{1}{J} \sum_{j=1}^{J} \tr  ((B^{2}(\Lambda,
    \Sigma_{j,-T})-B^{2}(\tilde{\Lambda},\Sigma_{j,-T}))'\Sigma_{j,-T} ) \right\rvert \\
  \leq &  \textstyle \frac{1}{J} \sum_{j=1}^{J}s_{1}(B^{2}(\Lambda,\Sigma_{j,-T})-B^{2}(\tilde{\Lambda},\Sigma_{j,-T}))\tr ((y_{j,-T} -
         \theta_{j,-T})(y_{j,-T} -
         \theta_{j,-T})' )  \\
  &  +\textstyle \frac{1}{J} \sum_{j=1}^{J} s_{1}(B^{2}(\Lambda,\Sigma_{j,-T})-B^{2}(\tilde{\Lambda},\Sigma_{j,-T}))\tr (\Sigma_{j,-T}),
\end{align*}
where the second inequality follows from the triangle inequality, von
Neumann's trace inequality, and the fact that the sum of the eigenvalues of
asymmetric matrix equals its trace. Now, using the fact that the operator norm
is bounded by the Frobenius norm, we obtain
\begin{align*}
  & \lvert G_{I,J}(\Lambda_{T,-T}, \Lambda_{-T}) -
    G_{I,J}(\tilde{\Lambda}_{T,-T}, \tilde{\Lambda}_{-T}) \rvert \\
  = &  2 \underline{s}_{\Sigma}^{-1}(\underline{s}_{\Sigma}^{-1} \vee
      \underline{s}_{\Sigma}^{-2}) K_{T,-T} B_{J}
      \lVert (\Lambda_{T,-T}, \Lambda_{-T}) -  (\tilde{\Lambda}_{T,-T}, \tilde{\Lambda}_{-T}) \rVert,
\end{align*}
where
$B_{J} = \frac{1}{J} \sum_{j=1}^{J}\tr ((y_{j,-T} -
\theta_{j,-T})(y_{j,-T} - \theta_{j,-T})' +\Sigma_{j,-T})$, which is $O_{p}(1)$ by
the law of large numbers. This establishes that
$\sup_{\overline{\mathcal{L}}}\lvert G_{I,J}(\Lambda_{T,-T}, \Lambda_{-T}) \rvert
\overset{p}{\to} 0.$ Again, the mode of convergence can be strengthened to
$L^{1}$ by verifying a uniform integrability conditions. To this end, note that
the summand in the definition of $G_{I,J}(\Lambda_{T,-T}, \Lambda_{-T})$ can be
bounded by 
\begin{align*}
  & \lvert \tr  (B(\Lambda,\Sigma_{j,-T})B(\Lambda,\Sigma_{j,-T})'((y_{j,-T} -
    \theta_{j,-T})(y_{j,-T} -
    \theta_{j,-T})' -\Sigma_{j,-T} )) \rvert \\
  \leq & K_{T,-T}^{2} \underline{s}_{\Sigma}^{-2}\tr ((y_{j,-T} -
         \theta_{j,-T})(y_{j,-T} -
         \theta_{j,-T})' +\Sigma_{j,-T} )),
\end{align*}
which follows by the same steps used when showing the Lipschitz condition. Since
the expectation of the square of the right-hand side is bounded uniformly over
$j$, it follows that $\sup_{j}\E \sup_{\overline{\mathcal{L}}}\lvert G_{I,J}(\Lambda_{T,-T},
\Lambda_{-T}) \rvert^{2} < \infty$. This concludes the proof for
$\mathrm{(I)}_{J}$, and the exact same steps with ``$-T$ replaced with $-1$''
also shows that $\mathrm{(III)}_{J}$ converges to zero uniformly over
$\overline{\mathcal{L}}$ in $L^{1}$.

For $\mathrm{(II)}_{J}$, note that
\begin{align*}
  & \left\lvert \textstyle \frac{1}{J} \sum_{j=1}^{J} 
    (B(\Lambda,\Sigma_{j,-T})'(y_{j,-T} - \theta_{j,-T}))(B(\Lambda,\Sigma_{j,-T})'\theta_{j,-T}  - \theta_{j,T})
    \right\rvert \\
  \leq & \textstyle\left\lvert \frac{1}{J} \sum_{j=1}^{J} 
         \theta_{j,-T}' B(\Lambda,\Sigma_{j,-T})B(\Lambda,\Sigma_{j,-T})'(y_{j,-T} -
         \theta_{j,-T}) \right \rvert \\
  & + \textstyle\left\lvert \frac{1}{J} \sum_{j=1}^{J} 
    \theta_{j,T}B(\Lambda,\Sigma_{j,-T})'(y_{j,-T} - \theta_{j,-T})
    \right\rvert.
\end{align*}
Note that the summand of the first term on the right-hand side can be written as
\begin{equation*}
  \tr ( B(\Lambda,\Sigma_{j,-T})B(\Lambda,\Sigma_{j,-T})'(y_{j,-T} -
  \theta_{j,-T})\theta_{j,-T}'),
\end{equation*}
which is very similar to the summand of $\mathrm{(I)}_{J}$. The same steps used
there go through without any added difficulty. The second term is even simpler,
and extremely similar to $\mathrm{(IV)}_{J}$ above in the decomposition of the
first term on the right-hand side of \eqref{eq:decomp_ure_loss_fc}. The same
lines of argument used to establish convergence of such term can be used here to
show the desired convergence result. Note that none of the convergence results
depend on the choice of sequence $\{((\theta_{j}', \theta_{j,T+1})',
\Sigma_{j})\}_{j=1}^{\infty}$ under Assumption \ref{assum:bounded_rand}.

Now, it remains to show that $\mathrm{(V)}_{J}$ and $\mathrm{(VI)}_{J}$
converges to zero uniformly over $\overline{\mathcal{L}}$, for almost all
sequences $\{((\theta_{j}', \theta_{j,T+1})',
\Sigma_{j})\}_{j=1}^{\infty}$. Here, it is convenient to treat
$\{((\theta_{j}', \theta_{j,T+1})', \Sigma_{j})\}_{j=1}^{\infty}$. I denote by
$\E _{f}$ the expectation with respect to the random sequence
$\{((\theta_{j}', \theta_{j,T+1})', \Sigma_{j})\}_{j=1}^{\infty}$. All almost
sure assertions in the remainder of the proof is with respect to the randomness
of $\{((\theta_{j}', \theta_{j,T+1})', \Sigma_{j})\}_{j=1}^{\infty}$. It follows
that
\begin{align*}
  &\textstyle \left\lvert  \frac{1}{J}\sum_{j=1}^{J}\left( B(\Lambda,\Sigma_{j,-T})'\Sigma_{j,-T}B(\Lambda,
    \Sigma_{j,-T}) - B(\Lambda,\Sigma_{j,-1})'\Sigma_{j,-1}B(\Lambda,\Sigma_{j,-1}) \right)
    \right\rvert \\
  \leq & \textstyle\left\lvert  \frac{1}{J}\sum_{j=1}^{J}\left( B(\Lambda,\Sigma_{j,-T})'\Sigma_{j,-T}B(\Lambda,
         \Sigma_{j,-T}) -\E _{f}B(\Lambda,\Sigma_{j,-T})'\Sigma_{j,-T}B(\Lambda,
         \Sigma_{j,-T}) \right) \right\rvert \\
  &+\textstyle \left\lvert   \frac{1}{J}\sum_{j=1}^{J}\left( B(\Lambda,
    \Sigma_{j,-1})'\Sigma_{j,-1}B(\Lambda,\Sigma_{j,-1}) -  \E _{f} B(\Lambda,
    \Sigma_{j,-1})'\Sigma_{j,-1}B(\Lambda,\Sigma_{j,-1}) \right)
    \right\rvert,
\end{align*}
where in the inequality I use Assumption \ref{assu:stationarity_main} and use the
fact that the expectations are equal. I show that the first term on the
right-hand side converges to $0$ almost surely, uniformly over
$\overline{\mathcal{L}}$. The same result can be shown for the second term using
the exact same argument. Define
\begin{align*}
  & G_{\mathrm{V}, J}(\Lambda_{T,-T}, \Lambda_{-T}) \\
  = & \textstyle \frac{1}{J}\sum_{j=1}^{J}\left( B(\Lambda,\Sigma_{j,-T})'\Sigma_{j,-T}B(\Lambda,
      \Sigma_{j,-T}) -\E _{f}B(\Lambda,\Sigma_{j,-T})'\Sigma_{j,-T}B(\Lambda,
      \Sigma_{j,-T}) \right).
\end{align*}
Note that since $\E _{f} s_{1}(\Sigma_{j})$ exists, we have
\begin{equation*}
  \E _{f}B(\Lambda,\Sigma_{j,-T})'\Sigma_{j,-T}B(\Lambda,
  \Sigma_{j,-T}) \leq  K_{T,-T}^{-2}\underline{s}_{\Sigma}^{-2}\E _{f}s_{1}(\Sigma_{j}) < \infty.
\end{equation*}
Hence, by the strong law of large numbers, we have $G_{\mathrm{V},
  J}(\Lambda_{T,-T}, \Lambda_{-T}) \to 0$ almost surely. For uniformity over
$\overline{\mathcal{L}}$, again I verify a Lipschitz condition for
$G_{\mathrm{V}, J}(\Lambda_{T,-T}, \Lambda_{-T})$. We have
\begin{align*}
  &\textstyle \left\lvert G_{\mathrm{V}, J}(\Lambda_{T,-T}, \Lambda_{-T}) - G_{\mathrm{V},J} (\tilde{\Lambda}_{T,-T}, \tilde{\Lambda}_{-T}) \right\rvert \\
  \leq & \textstyle \frac{1}{J}\sum_{j=1}^{J} \lvert \tr ((B^{2}(\Lambda,\Sigma_{j,-T})- B^{2}(\tilde{\Lambda},\Sigma_{j,-T}))\Sigma_{j,-T}) \rvert \\
  & +\textstyle \frac{1}{J}\sum_{j=1}^{J}  \E _{f}\lvert \tr ((B^{2}(\Lambda,
    \Sigma_{j,-T}) - B^{2}(\tilde{\Lambda},\Sigma_{j,-T}))\Sigma_{j,-T})\rvert \\
  \leq  & \textstyle 2 \underline{s}_{\Sigma}^{-1}(\underline{s}_{\Sigma}^{-1} \vee
          \underline{s}_{\Sigma}^{-2}) K_{T,-T} B_{J}
          \lVert (\Lambda_{T,-T}, \Lambda_{-T}) -  (\tilde{\Lambda}_{T,-T}, \tilde{\Lambda}_{-T}) \rVert 
\end{align*}
where the first inequality follows from multiple applications of the triangle
inequality, and the second inequality follows with
$B_{J} = \frac{1}{J}\sum_{j=1}^{J}(\tr (\Sigma_{j,-T}) + E
\tr (\Sigma_{j,-T}))$ from von Neumann's trace inequality and
\eqref{eq:B2_Lip}. Let $\overset{a.s.}{\to}$ denote almost sure convergence with
respect to the density $f_{(\theta',\theta_{T+1})',M}$. By the strong law of
large numbers, it follows that
$B_{J} \overset{a.s.}{\to} 2\E _{f} \tr (\Sigma_{j,-t}) $. Hence, by Lemma 1 of
\cite{andrews1992GenericUniformConvergence}, I conclude that
$\sup_{\overline{\mathcal{L}}}\lvert G_{\mathrm{V}, J}(\Lambda_{T,-T},
\Lambda_{-T}) \rvert \overset{a.s.}{\to} 0$. 

For $\mathrm{(VI)_{J}}$, the triangle inequality gives
\begin{align*}
  & \textstyle \left\lvert \frac{1}{J}\sum_{j=1}^{J}\left( (B(\Lambda,
    \Sigma_{j,-T})'\theta_{j,-T}  - \theta_{j,T})^{2} -  (B(\Lambda,
    \Sigma_{j,-1})'\theta_{j,-1}  - \theta_{j,T+1})^{2} \right)  \right\rvert \\
  \leq & \textstyle \left\lvert \frac{1}{J}\sum_{j=1}^{J}\left( (B(\Lambda,
         \Sigma_{j,-T})'\theta_{j,-T}  - \theta_{j,T})^{2} -  \E _{f}(B(\Lambda,
         \Sigma_{j,-T})'\theta_{j,-T}  - \theta_{j,T})^{2} \right)  \right\rvert \\
  & +\textstyle \left\lvert \frac{1}{J}\sum_{j=1}^{J}\left(   (B(\Lambda,
    \Sigma_{j,-1})'\theta_{j,-1}  - \theta_{j,T+1})^{2}- \E _{f}(B(\Lambda,
    \Sigma_{j,-1})'\theta_{j,-1}  - \theta_{j,T+1})^{2} \right)  \right\rvert.
\end{align*}
Again, I show that the desired convergence result only for the first term since
the result for the second term will follow from the exact same steps. Define
\begin{align*}
  &G_{\mathrm{VI},J}(\Lambda_{T,-T}, \Lambda_{-T}) \\
  = & \textstyle \frac{1}{J}\sum_{j=1}^{J}\left( (B(\Lambda,
      \Sigma_{j,-T})'\theta_{j,-T}  - \theta_{j,T})^{2} -  \E _{f}(B(\Lambda,
      \Sigma_{j,-T})'\theta_{j,-T}  - \theta_{j,T})^{2} \right).
\end{align*}
To show $G_{\mathrm{VI},J}(\Lambda_{T,-T}, \Lambda_{-T}) \overset{a.s.}{\to} 0$,
note that
\begin{align*}
  & \E _{f}(B(\Lambda,\Sigma_{j,-T})'\theta_{j,-T}  - \theta_{j,T})^{2} \\
  \leq & 2 \E _{f}\tr  (B^{2}(\Lambda,
         \Sigma_{j,-T})\theta_{j,-T}\theta_{j,-T}') + \E _{f}\theta_{jT}^{2} \\
  \leq &\textstyle 2 K_{T,-T}^{2}\underline{s}_{\Sigma}^{-2} \sum_{t=1}^{T-1}
         \E _{f}\theta_{jt}^{2} + \E _{f}\theta_{jT}^{2} < \infty,
\end{align*}
where the second inequality follows because
\begin{equation*}
  \tr  (B^{2}(\Lambda,
  \Sigma_{j,-T})\theta_{j,-T}\theta_{j,-T}') \leq s_{1}(B^{2}(\Lambda,
  \Sigma_{j,-T})) \tr (\theta_{j,-T}\theta_{j,-T}')
\end{equation*}
due to von Neumann's trace inequality. Hence, by the strong law of large
numbers, we have $G_{\mathrm{VI},J}(\Lambda_{T,-T}, \Lambda_{-T}) \overset{a.s.}{\to} 0$.

Once again, I verify a Lipschitz condition to show that this convergence is in
fact uniform over $\overline{\mathcal{L}}$. %
By \eqref{eq:B_Lip} and \eqref{eq:B2_Lip},
\begin{align*}
  &\lvert  (B(\Lambda,\Sigma_{j,-T})'\theta_{j,-T}  - \theta_{j,T})^{2}
    - (B(\tilde{\Lambda},\Sigma_{j,-T})'\theta_{j,-T}  - \theta_{j,T})^{2} \rvert \\
  = & \lvert \tr  ((B^{2}(\Lambda,
      \Sigma_{j,-T}) - B^{2}(\tilde{\Lambda},
      \Sigma_{j,-T}))\theta_{j,-T}\theta_{j,-T}') \rvert \\
  & +2 \lvert (B(\Lambda,\Sigma_{j,-T}) - B(\Lambda,
    \Sigma_{j,-T}))'\theta_{j,-T}\theta_{jT} \rvert \\
  \leq & \textstyle 2 \underline{s}_{\Sigma}^{-1}(\underline{s}_{\Sigma}^{-1} \vee
         \underline{s}_{\Sigma}^{-2}) K_{T,-T} \lVert \theta_{j,-T} \rVert^{2}
         \lVert (\Lambda_{T,-T}, \Lambda_{-T}) -
         (\tilde{\Lambda}_{T,-T}, \tilde{\Lambda}_{-T}) \rVert \\
  & + 2 (\underline{s}_{\Sigma}^{-1} \vee
    \underline{s}_{\Sigma}^{-2})\lVert \theta_{j,-T}\theta_{jT} \rVert \lVert (\Lambda_{T,-T}, \Lambda_{-T}) -
    (\tilde{\Lambda}_{T,-T}, \tilde{\Lambda}_{-T}) \rVert \\
  := & B_{j}\lVert (\Lambda_{T,-T}, \Lambda_{-T}) -
       (\tilde{\Lambda}_{T,-T}, \tilde{\Lambda}_{-T}) \rVert
\end{align*}
Likewise, we have
\begin{align*}
  &\lvert  \E _{f}(B(\Lambda,\Sigma_{j,-T})'\theta_{j,-T}  - \theta_{j,T})^{2}
    - \E _{f}(B(\tilde{\Lambda},\Sigma_{j,-T})'\theta_{j,-T}  - \theta_{j,T})^{2} \rvert \\
  \leq & \textstyle 2 \underline{s}_{\Sigma}^{-1}(\underline{s}_{\Sigma}^{-1} \vee
         \underline{s}_{\Sigma}^{-2}) K_{T,-T} \E _{f}\lVert \theta_{j,-T} \rVert^{2}
         \lVert (\Lambda_{T,-T}, \Lambda_{-T}) -
         (\tilde{\Lambda}_{T,-T}, \tilde{\Lambda}_{-T}) \rVert \\
  & + 2 (\underline{s}_{\Sigma}^{-1} \vee
    \underline{s}_{\Sigma}^{-2})\lVert \E _{f}\theta_{j,-T}\theta_{jT} \rVert \lVert (\Lambda_{T,-T}, \Lambda_{-T}) -
    (\tilde{\Lambda}_{T,-T}, \tilde{\Lambda}_{-T}) \rVert \\
  := & B_{{E}_{f}}  \lVert (\Lambda_{T,-T}, \Lambda_{-T}) -
       (\tilde{\Lambda}_{T,-T}, \tilde{\Lambda}_{-T}) \rVert.
\end{align*}
Combining the two inequalities, we have
\begin{align*}
  & \lvert G_{\mathrm{VI},J}(\Lambda_{T,-T}, \Lambda_{-T}) -
    G_{\mathrm{VI},J}(\tilde{\Lambda}_{T,-T}, \tilde{\Lambda}_{-T}) \rvert \\
  \leq & \textstyle \left( \frac{1}{J}\sum_{j=1}^{J}B_{j} + B_{E_{f}} \right) \lVert (\Lambda_{T,-T}, \Lambda_{-T}) -
         (\tilde{\Lambda}_{T,-T}, \tilde{\Lambda}_{-T}) \rVert
\end{align*}
Hence, it suffices to show that $\frac{1}{J}\sum_{j=1}^{J}B_{j}
\overset{a.s.}{\to} B $ for some fixed $B$. By the strong law of large numbers,
we have
\begin{align*}
  & \textstyle \frac{1}{J}\sum_{j=1}^{J} \lVert \theta_{j,-T} \rVert^{2} 
    \overset{a.s.}{\to} \E \lVert \theta_{j,-T} \rVert^{2} <\infty \\
  & \textstyle \frac{1}{J}\sum_{j=1}^{J} \lVert \theta_{j,-T}\theta_{jT} \rVert 
    \overset{a.s.}{\to} \E \lVert \theta_{j,-T}\theta_{jT} \rVert \leq
    (\E \theta_{j,T}^{2}\E \lVert \theta_{j,-T} \rVert^{2})^{\frac{1}{2}} <\infty,
\end{align*}
which establishes that $\frac{1}{J}\sum_{j=1}^{J}B_{j}$ indeed converges almost
surely to a finite value, which concludes the proof.

\section{Sufficient Condition for Assumption \ref{assu:bdd_quan}}
\label{sec:suff-cond-assumpt}

To show that the term
$\E  \textstyle \sup_{\mu \in \mathcal{M}_{j}} \lVert \mu \rVert^{2} $ is
bounded, note that
\begin{equation*}
\E  \textstyle \sup_{\mu \in \mathcal{M}_{j}} \lVert \mu \rVert^{2} =
\sum_{t=1}^{T} \E  q_{1-\tau}(\{ y_{jt}^{2} \}_{j=1}^{J})
\end{equation*}
by the definition of $\mathcal{B}$. Hence, it suffices to show
$\E  q_{1- \tau}(\{y_{jt}^{2}\}_{j=1}^{J}) = O(1)$ for each $t \leq
T$. To control the sample quantile behavior of $\{y_{jt}^{2}\}_{j=1}^{J}$, I
impose an additional condition. Write $\varepsilon_{jt} := y_{jt} - \theta_{jt}$
so that $\E\varepsilon_{jt} = 0$ and $\E\varepsilon_{jt}^{2} = \sigma^{2}_{jt}$,
where $\sigma^{2}_{jt}$ denotes the $t$th diagonal entry of $\Sigma_{j}$. Note
that $\overline{\sigma}^{2}_{t} := \sup_{j}\sigma^{2}_{jt} < \infty$ by
Assumption \ref{assum:bounded}. I assume that the distribution of
$\varepsilon_{jt}$ belongs to a scale family with finite fourth moments.

\begin{assumption}[Scale family]\label{assu:scale_fam}
  For each $t$, $ \varepsilon_{jt}/\sigma_{jt}\overset{i.i.d.}{\sim} F_{t} $ for
  $j=1, \dots, J$, where $F_{t}$ is a distribution function with finite fourth
  moments.
\end{assumption}
Note that the assumption is notably weaker than requiring that the noise vectors
$\varepsilon_{j}$ for $j=1, \dots, J$ belong to a multivariate scale family,
which restricts the joint distribution across $t$ in a much more stringent
way. Here, I instead require that the error terms belong to a scale family only
for each period. It can be shown that, by Assumption \ref{assum:bounded}, the
problem of bounding $\E  q_{1- \tau}(\{y_{jt}^{2}\}_{j=1}^{J})$ boils
down to the problem of bounding
$\E  q_{1-
  \tau}(\{(\varepsilon_{jt}/\sigma_{jt})^{2}\}_{j=1}^{J})$. Then, by Assumption
\ref{assu:scale_fam}, this simplifies to bounding the mean of the sample
quantile of an i.i.d. sample. I use a result given by
\cite{okolewski2001SharpDistributionfreeBounds} to derive a bound on this
quantity without having to further impose conditions on the distribution
$F_{t}$.

\begin{example}[Teacher value-added]
  In the teacher value-added example, Assumption \ref{assu:scale_fam} is
  satisfied as long as the idiosyncratic error terms are i.i.d across $j$ and
  have finite fourth moment. Hence, this assumption is almost always satisfied
  in teacher value-added models, or in linear panel data models in general.
\end{example}

\begin{lemma}
  Suppose Assumptions \ref{assum:bounded} and \ref{assu:scale_fam} hold. Then,
  Assumption \ref{assu:bdd_quan} holds. 
\end{lemma}

\begin{proof}
Observe that
\begin{align*}
  q_{1-\tau}^{2}(\{\lvert y_{jt}\rvert \}_{j=1}^{J})
  =  \, q_{1-\tau}(\{ y_{jt}^{2}\}_{j=1}^{J}) 
  \leq & \, q_{1-\tau}(\{2\theta_{jt}^{2}  + 2 \varepsilon_{jt}^{2}\}_{j=1}^{J}) \\
  \leq & \, 2 q_{ 1-\tau/2}(\{ \theta_{jt}^{2} \}_{j=1}^{J}) + 2 q_{1-\tau/2}(\{ \varepsilon_{jt}^{2}\}_{j=1}^{J}),
\end{align*}
where the last inequality follows from a property of a quantile that the
$1-\tau $ quantile of the sum of two random variables are bounded by the sum of
the $1 - \tau/2$ quantiles of those two random variables. It follows that
$ q_{1-\tau/2}(\{ \theta_{jt}^{2} \}_{j=1}^{J}) < \sup_{j}\theta_{jt}^{2} <
\infty$, and thus it suffices to show that
$ \limsup_{J \to \infty} Eq_{1-\tau/2}(\{ \varepsilon_{jt}^{2}\}_{j=1}^{J}) <
\infty$. I have
\begin{equation*}
  q_{1-\tau/2}(\{\varepsilon_{jt}^{2}\}_{j=1}^{J}) 
  =  q_{1-\tau/2}(\{\sigma_{jt}^{2}\eta_{jt}^{2}\}_{j=1}^{J}) 
  \leq   q_{1-\tau/2}(\{\overline{\sigma}_{t}^{2}\eta_{jt}^{2}\}_{j=1}^{J}) 
  =   \overline{\sigma}_{t}^{2} q_{1-\tau/2}(\{\eta_{jt}^{2}\}_{j=1}^{J}),
\end{equation*}
where the first equality holds by Assumption \ref{assu:scale_fam} and the
inequality holds because replacing $\sigma_{jt}^{2}$ by
$\overline{\sigma}_{t}^{2}$ makes all the sample points larger, and thus the
sample quantile larger. Define
$\underline{\tau} =2 (1 - \lceil{J (1-\tau/2)}\rceil / J)$, which is the largest
$\underline{\tau} \leq \tau$ such that $J(1-\underline{\tau}/2)$ is an
integer. By a result on the bias of sample quantiles given by
\cite{okolewski2001SharpDistributionfreeBounds}, it follows that
\begin{equation*}
  \sup_{J} \E  \, q_{1-\underline{\tau}/2}(\{\eta_{jt}^{2}\}_{j=1}^{J}) \leq \left(
    \frac{\var (\eta_{jt}^{2})}{(1-\underline{\tau} / 2)\underline{\tau}/2} \right)^{1/2} +
  F_{t}^{-1}(1-\underline{\tau}/2) <\infty,
\end{equation*}
and
$q_{1-{\tau}/2}(\{\eta_{jt}^{2}\}_{j=1}^{J}) \leq
q_{1-\underline{\tau}/2}(\{\eta_{jt}^{2}\}_{j=1}^{J})$ because
$\underline{\tau} \leq \tau$. This establishes that
$ \E  \sup_{\mu} \lVert \mu \rVert^{2} < \infty$, which concludes the
proof.
\end{proof}

\end{appendices}
\end{document}

%% file: preamble.tex
\usepackage[utf8]{inputenc}
\usepackage[T1]{fontenc}

\usepackage{amsmath}
\usepackage{amsfonts}
\usepackage{amssymb}
\usepackage{amsthm}
\usepackage{graphicx}
\usepackage{setspace}
\usepackage{verbatim}
\usepackage{natbib}
\usepackage{booktabs}
\usepackage{hyperref}
\usepackage{cleveref}
\usepackage{subcaption}

\usepackage[margin=1.25in]{geometry}
\usepackage[title]{appendix}

\crefname{appsec}{Appendix}{Appendices}
\crefname{sappsec}{Supplemental Appendix}{Supplemental Appendices}

\newtheorem{theorem}{Theorem}[section]
\newtheorem{lemma}{Lemma}[section]

\newtheorem{assumption}{Assumption}[section]

\theoremstyle{definition}

\newtheorem{remark}{Remark}[section]
\newtheorem{example}{Example}[section]